\documentclass[a4paper,11pt,oneside]{article}

\usepackage[left=2cm,right=2cm,top=2.5cm,bottom=3cm]{geometry}
\usepackage[centertags]{mathtools}
\usepackage{amssymb, mathrsfs, amsthm}
\usepackage{amsmath}
\usepackage{mathrsfs}
\usepackage{graphicx}
\usepackage{multicol}
\usepackage{multirow}
\usepackage{xcolor}
\usepackage{csquotes}
\usepackage{bbm}
\usepackage{bm}
\usepackage{mwe}
\usepackage{pdflscape}
\usepackage[noblocks,affil-sl]{authblk} 
\usepackage{caption}
\usepackage{subcaption}
\usepackage{hhline}
\usepackage{adjustbox}
\usepackage{dsfont}
\usepackage{colortbl}
\usepackage[round]{natbib}
\usepackage{cases}
\usepackage{soul}
\usepackage{cancel}
\usepackage{float}
\usepackage{enumerate}
\usepackage{xfrac}
\mathtoolsset{showonlyrefs}
\usepackage{hyperref}
\usepackage{makecell}
\usepackage{comment}
\usepackage{orcidlink}
\usepackage[table]{xcolor}

\counterwithin{figure}{section}
\counterwithin{table}{section}
\numberwithin{equation}{section}

\allowdisplaybreaks

\theoremstyle{plain}
\newtheorem{thm}{Theorem}[section]
\newtheorem{lem}[thm]{Lemma}
\newtheorem{prop}[thm]{Proposition}

\newtheorem{defn}[thm]{Definition}

\newtheorem{remark}[thm]{Remark}

\newtheorem{example}[thm]{Example}

\makeatletter
\def\@makefnmark{\hbox{\@textsuperscript{\normalfont(\@thefnmark)}}}
\makeatother 
\renewcommand{\P}{\mathbb{P}}
\newcommand{\Q}{\mathbb{Q}}

\newcommand{\mcA}{\mathcal{A}}
\newcommand{\mcC}{\mathcal{C}}
\newcommand{\mcL}{\mathcal{L}}
\newcommand{\mcLC}{\mathcal{L}^\mathbf{C}}

\newcommand{\mcR}{\mathcal{R}}
\newcommand{\mcN}{\mathcal{N}}
\newcommand{\mcD}{\mathcal{D}}

\newcommand{\bfc}{\mathbf{c}}
\newcommand{\bfC}{\mathbf{C}}
\newcommand{\bfe}{\mathbf{e}}
\newcommand{\bfa}{\mathbf{a}}
\newcommand{\bfL}{\mathbf{L}}
\newcommand{\bfS}{\mathbf{S}}

\newcommand{\bfW}{\mathbf{W}}
\newcommand{\bmpi}{\bm{\pi}}
\newcommand{\whbmpi}{{\widehat{\bmpi}}}
\newcommand{\bmmu}{\bm{\mu}}
\newcommand{\1}{\bm{1}}

\newcommand{\ind[1]}{\mathds{1}_{\{#1\}}}

\newcommand{\bfSigma}{\mathbf{\Sigma}}
\newcommand{\bfZ}{\mathbf{Z}}
\newcommand{\bfD}{\mathbf{D}}

\newcommand{\whbfZ}{\widehat{\bfZ}}

\newcommand{\whbfC}{\widehat{\bfC}}
\newcommand{\whC}{\widehat{C}}
\newcommand{\whX}{\widehat{X}}

\newcommand\de{\mathrm{d}}
\newcommand{\E}{\mathbb{E}}
\newcommand{\R}{\mathbb{R}}
\newcommand{\N}{\mathbb{N}}
\newcommand{\diag}{\text{diag}}

\newcommand{\less}{<}
\newcommand{\g}{>}
\DeclareMathOperator*{\argmax}{arg\,max} 

\newcommand{\Var}{\mathrm{Var}}

\begin{document}
\author[1]{Katia Colaneri, \orcidlink{0000-0003-3933-3788} \thanks{Corresponding author: katia.colaneri@uniroma2.it}}
\author[2]{Alessandra Cretarola, \orcidlink{0000-0003-1324-9342} \thanks{alessandra.cretarola@unich.it}}
\author[1]{Edoardo Lombardo, \orcidlink{0009-0009-1296-9891}\thanks{edoardo.lombardo@uniroma2.it}}
\author[3]{Daniele Mancinelli, \orcidlink{0000-0003-3854-2414}\thanks{daniele.mancinelli@polimi.it}}
\affil[1]{Department of Economics and Finance, University of Rome Tor Vergata.}
\affil[2]{Department of Economic Studies, University ``G. D'Annunzio" of Chieti-Pescara.}
\affil[3]{Department of Mathematics, Politecnico di Milano.}
\title{Carbon-Sensitive Fund Construction and Hedging for Green Unit-Linked Life Insurance}

\date{}
\maketitle
\vspace{-0.5cm}
\begin{abstract}
\noindent We study the problem of  hedging unit linked life insurance policies whose benefits depend on an investment fund that incorporates environmental criteria in its selection process. Offering these products poses
two key challenges: constructing a green investment fund and developing a hedging strategy for policies written on that fund. We address these two problems separately. First, we design a portfolio selection rule driven by firms’ carbon intensity that endogenously selects assets and avoids ad hoc pre-screens based on ESG scores. The effectiveness of our new portfolio selection method is tested using real market data. Second, we consider an insurance company issuing unit linked policies written on this fund. Such contracts are exposed to market, carbon, and mortality risk, which the insurance company seeks to hedge. Due to market incompleteness, we address the hedging problem via a quadratic approach aimed at minimizing the variance of the hedging costs. Finally, we also make a numerical analysis to assess the performance of the hedging strategy. For our simulation study, we use an efficient weak second-order scheme that allows for variance reduction.
\end{abstract}

\section{Introduction}
Climate change has emerged as a major source of systemic risk, with far-reaching implications for financial institutions and, in particular, for the insurance industry. Recent empirical evidence highlights how climate-related factors increasingly affect asset prices and portfolio allocation decisions (\citet{hartzmark2019investors, lagerkvist2020preferences, anquetin2022scopes}). In response, institutional investors have progressively incorporated environmental criteria into their investment strategies. For example, \cite{peng2024optimal} reports that the Chinese Government Pension Investment Fund has allocated $163$ trillion yen to passive ESG index products, while the California Public Employees’ Retirement System follows a ``social change investment” approach aligned with ESG principles. This trend has extended to insurance companies, especially in the design of \emph{insurance-based investment products} (IBIPs). Among these, \emph{unit-linked life insurance policies} with sustainability features have experienced rapid growth, driven both by investor demand and regulatory frameworks such as the Sustainable Finance Disclosure Regulation (SFDR). In particular, \cite{eiopa2023costs} reported a $24$\% increase in the number of IBIPs classified under Articles $8$ and $9$ of SFDR. Moreover, a study carried out by \cite{ivass2024report}, based on a sample of $106$ IBIPs offered by $18$ insurance companies, found that $92$\% of unit linked life insurance policies are classified under Article $8$ of the SFDR.\\
From an actuarial viewpoint, the introduction of environmental considerations into unit-linked contracts raises new modeling and risk management challenges. The payoff of these policies depends on the performance of an underlying investment fund, while the timing of payments is contingent on the policyholder’s lifetime. When the underlying fund incorporates environmental characteristics -- such as exposure to firms’ carbon emissions -- the insurer is simultaneously exposed to financial risk, mortality risk, and carbon risk. Since the latter two sources of risk are not spanned by traded assets, the resulting market is incomplete, and perfect replication of liabilities is not attainable . This calls for the development of hedging strategies that explicitly account for residual, non-hedgeable risks.\\
The actuarial literature provides well-established tools for dealing with hedging in incomplete markets. In particular, quadratic hedging criteria such as mean variance hedging and risk minimization have been widely and successfully applied within the life insurance context, especially for hedging products like unit-linked policies (see, e.g., \citet{moller1998risk, schweizer2001, moller2001risk,  vandaele2008locally, ceci2015hedging, biagini2016risk, ceci2017unit, biagini2017risk}). Under these approaches, one can derive optimal hedging strategies that minimize the variance of the hedging costs arising from market incompleteness, mainly due to mortality risk. However, the integration of climate-related risk factors—such as carbon emissions—into these frameworks remains largely unexplored.\\
At the same time, a growing strand of financial literature has focused on portfolio selection under sustainability constraints. Early contributions, such as \citet{andersson2016hedging}, propose excluding high-carbon assets to reduce portfolio exposure while maintaining low tracking error. \citet{bolton2022net} extend this approach by incorporating dynamic decarbonization constraints aligned with climate targets. Alternative methodologies impose sustainability constraints directly within the optimization problem (\citet{LeGuenedalRoncalli2023, de2023esg}), while preference-based approaches embed ESG considerations into investors’ utility functions (\citet{pastor2021sustainable, pedersen2021responsible}). We discuss the comparison between our approach and this literature more in detail in Section \ref{sec:comparison}. Despite these advances, the existing literature largely focuses on asset management and does not address the implications for insurance products whose liabilities depend on such portfolios.\\
This paper contributes to bridging this gap by developing a unified framework that combines carbon-sensitive portfolio construction with the hedging of unit-linked life insurance contracts. We adopt the perspective of an insurance company issuing policies whose benefits depend on a fund that incorporates carbon intensity into its investment decisions. First, we propose a portfolio selection methodology in which carbon risk is introduced endogenously through a penalization of terminal wealth, allowing for a flexible trade-off between financial performance and environmental impact. Second, we analyze the hedging problem faced by the insurer in an incomplete market characterized by the coexistence of financial, mortality, and carbon risks. To this end, we employ a quadratic hedging criterion and derive risk-minimizing strategies for standard life insurance contracts, including pure endowment, term insurance, and endowment policies.\\
Our approach extends the classical actuarial framework for unit-linked contracts by incorporating climate-related risk as an additional source of incompleteness. In doing so, it provides a tractable setting for assessing the interaction between sustainability considerations and insurance risk management, and contributes to the emerging literature at the intersection of actuarial science and sustainable finance.\\

The remainder of the paper is organized as follows. Section \ref{sec:setup} lays out the market model and its main assumptions and derives the optimal carbon-penalised investment portfolio, which provides the underlying of the unit linked contract. Section \ref{sec:examples} presents several plausible dynamics for the carbon intensity process and the corresponding characterisation of the value function. Section \ref{sect:EMPIRICAL_FIRST_PART} reports the numerical analysis of the optimal investment fund, combining a static empirical study based on real data with a dynamic simulation exercise. In Section \ref{sect:locally_risk_minimizing}, we shift to the insurance company’s perspective and study risk minimizing hedging strategies for green unit-linked life insurance policies in an incomplete market. Section \ref{sec:hedging} is devoted to the numerical implementation of such hedging strategies. Finally, Section \ref{sect:conclusions} concludes.
\section{A carbon-penalised investment fund}\label{sec:setup}
The hedging problem for unit-linked life insurance contracts is intrinsically linked to the dynamics of the underlying investment fund. Since contract payoffs are written on the value of a reference portfolio, the insurer’s liability is directly driven by the portfolio selection rule, making asset allocation an key component of the overall risk management problem. In this paper, the underlying fund is constructed by explicitly incorporating firms’ carbon aversion into the investment decision. This introduces an additional source of risk (i.e. carbon risk) which, together with financial and mortality risks, contributes to market incompleteness. As some of these sources of uncertainty are not spanned by traded assets, perfect replication of the insurance liability is not feasible. Consequently, the residual risk faced by the insurer depends not only on the hedging strategy, but also on the structural properties of the underlying portfolio. Therefore it is natural to address both the problem of the fund construction and the problem of hedging within a unified framework.\\
To do this, we first model the investment problem within in a continuous-time financial market, where portfolio selection is formulated as a stochastic control problem with a carbon-dependent penalisation of terminal wealth. The optimal portfolio obtained in this setting defines the dynamics of the reference fund underlying the unit-linked contract. This construction provides the basis for the subsequent analysis of risk-minimizing hedging strategies in the combined financial–insurance market, carried in Sections \ref{sect:locally_risk_minimizing} and \ref{sec:hedging}.
\subsection{The financial market}
In this section we define the modeling framework for the construction of the investment fund. We let $\left(\Omega^M,\mathcal{F},\P^M\right)$ be a fixed probability space and $T$ a finite time horizon coinciding with the terminal time of an investment. We also introduce a $\P^M$-complete and right-continuous filtration $\mathbb{F}=\left\lbrace\mathcal{F}_t\right\rbrace_{t\in[0,T]}$, and we assume that all processes below are $\mathbb{F}$-adapted. We consider a financial market model consisting of $d$ stocks with price processes $\bfS=\left\lbrace\bfS_t\right\rbrace_{t\in[0,T]}$ in $\mathbb{R}^d_+$, and one risk-free asset $S_0=\{S^0_t\}_{t\in[0,T]}$, that are traded continuously on $[0,T]$. The dynamics of the risk-free asset are given by 
\begin{equation}\label{eq:cash_account}
\de S^0_t=rS^0_t\de t,\quad S^0_0=1, 
\end{equation}
where $r\in\mathbb{R}_{+}$ denotes the constant risk-free interest rate. The price dynamics of the vector of risky assets $\bfS$ are given by
\begin{equation}\label{eq:risk_assets_dyn}
\de\bfS_t=\diag(\bfS_t)\left(\bmmu\de t+\bfSigma\de\bfZ_t\right),\quad\bfS_0\in\mathbb{R}^{+},
\end{equation}
where $\bmmu\in\mathbb{R}^d$ is the vector of constant drift rates, $\bfZ=\left\lbrace\bfZ_t\right\rbrace_{t\in[0,T]}$ is a standard $d$-dimensional $\mathbb{F}$-Brownian motion with independent components, and $\bfSigma\bfSigma^\top$ is the variance-covariance matrix of log-returns. In particular, $\bfSigma=\bfL\bfD$ where $\bfD=\diag\left(\sigma_1,\dots,\sigma_d\right)$ with $\sigma_i>0$ for all $i=1,\dots,d$, and $\bfL$ is the lower triangular matrix obtained through Cholesky decomposition of the correlation matrix $\bm{\rho}=\left(\rho_{i,j}\right)_{i,j=1,\dots,d}$, so that $\bm{\rho}=\bfL\bfL^\top$.\\
Stocks are assumed to be issued by firms with different levels of carbon emissions, measured by carbon intensity. A firm’s carbon intensity is defined as the ratio between the total greenhouse gas emissions
in metric tonnes of CO$_2$ and total revenues (in USD millions). Let $\bfC=(C_1,\dots,C_d)^\top$, where each component $C_i=\{C_{i,t}\}_{t\in[0,T]}$ denotes the carbon intensity process of the $i$-th firm. We assume that $\bfC$ is a $d$-dimensional Markov process with mutually independent components taking values in $\mcD$. The 
specific choice of $\mcD$ depends on the selected model for the carbon-intensity 
process; in the examples considered below, $\mcD$ is either a finite subset of 
$\mathbb R^d$, $\mathbb R^d$, or $\mathbb R_+^d$. Moreover, we assume that $C_i$ is independent of $\bfZ$, for every $i=1,\dots, d$.
\begin{remark}
The assumption that $\bfC$ has mutually independent components reflects the idea that firm-level carbon intensity evolves largely as an idiosyncratic process, shaped by heterogeneous technologies, sector-specific constraints, and individual decarbonization strategies. Moreover, the independence between $C_i$ and $\bfZ$ for every $i=1,\dots,d$ is supported by empirical findings from \cite{bolton2021investors}, which shows that stock performance is significantly related to the absolute level of firm emissions, but not to carbon intensity. In particular, carbon intensity does not appear to be priced by the market, indicating that it is largely uncorrelated with the drivers of financial returns. 
\end{remark}

\subsection{Optimal carbon-penalised investment portfolio}\label{sec:optimal_inv}
Next, our goal is to construct an investment fund that accounts for carbon risk. Let $X^{\bmpi}=\left\lbrace X^{\bmpi}_t\right\rbrace_{t\in[0,T]}$ be the value of a self-financing portfolio and let $\bmpi=(\pi_{1},\dots,\pi_{d})^\top$, where $\pi_i=\{\pi_{i,t}\}_{t\in[0,T]}$, denote fraction of the wealth invested in the risky assets $\bfS$. Consequently, the percentage of wealth invested in the riskless asset is $1-\bmpi_t^\top\1$, for every $t\in[0,T]$,  where $\1=(1,\ldots,1)^\top \in \R^d$. We introduce now the suitable set of strategies.
\begin{defn} 
A $\mathbb{F}$-admissible investment portfolio is a self-financing, $\mathbb{F}$-predictable vector of weights $\bmpi$ such that 
\begin{equation}\label{eq:adm_cond}
\mathbb{E}^{\mathbb{P}^M}\left[\int_0^T\|\bmpi_s\|^2\de s\right]\less\infty.
\end{equation}
The set of $\mathbb{F}$-admissible investment portfolios is denoted by $\mcA$.
\end{defn}
For every $\bmpi\in\mcA$, the dynamics of the associated investment portfolio $X^{\bmpi}$ is given by
\begin{equation}
\dfrac{\de X_t^{\bmpi}}{X_t^{\bmpi}}=\left[r+\bmpi_t^\top(\bmmu-r\1)\right]\de t+\bmpi_t^\top\bfSigma\de \bfZ_t,\quad X_0^{\bmpi}=x_0.
\end{equation}
The objective is to maximise the expected utility of terminal wealth under a carbon-penalised criterion, which discourages investment in assets with high carbon intensity. Such penalisation is applied to the terminal value of the investment portfolio, and it is assumed to be proportional to the riskiness of stocks, measured according to their realized volatility. The idea of penalisation is inspired by  \cite{rogers2013optimal}, under a slightly different definition (see Section \ref{sec:comparison}). The carbon-penalised investment portfolio value at maturity is given by 
\begin{equation}
\tilde{X}_T^{\bmpi}=X^{\bmpi}_T\exp\left(-\dfrac{1}{2}\int_0^T\bmpi_s^\top\left(\bfe(s,\bfC_s)\odot\bfD\bfD^\top\right)\bmpi_s\de s\right),
\end{equation}
where $\odot$ denotes the elementwise (Hadamard) product, and  $\bfe\left(t,\bfC_t\right)=\text{diag}\left(\varepsilon_1(t,C_{1,t}),\dots,\varepsilon_d(t,C_{d,t})\right)$ denotes the diagonal matrix whose diagonal contains carbon aversion coefficients associated with each risky asset in the investment portfolio. In particular, the carbon aversion of the $i$-th stock is given by 
\begin{equation}\label{eq:penalisation}
\varepsilon_i(t,C_{i,t})=\alpha_i(t) C_{i,t}\ind[C_{i,t}>0],\quad t\in[0,T],
\end{equation}
where $\alpha_i(\cdot):[0,T]\to[0,+\infty)$ for every $i=1,\dots,d$. We assume that 
\begin{equation}\label{eq:int_on_epsilon}
\E^{\P^M}\left[\int_0^T\alpha_i(s)C_{i,s}\ind[C_{i,s}>0]\de s\right]<\infty,
\end{equation}
for every $i=1,\dots,d$. Specifically, $\alpha_i$ captures the time-varying environmental preferences with respect to the $i$-th firm, reflecting how strongly the carbon emissions of that firm are penalized in the investment decision. A higher $\alpha_i$ corresponds to a greater carbon aversion with respect to the $i$-th asset and thus a stronger shift toward low-emission investments. The proposed penalisation scheme is not necessarily asset-specific. For instance, one may choose to assign the same penalisation to all assets, or adopt a sector-based approach where the same $\alpha$ is applied to all stocks within a given sector. Moreover, $\alpha_i$ is time-varying, allowing the penalisation scheme to evolve consistently with long-term environmental targets. In particular, it can be aligned with exogenous sustainability goals, such as regulatory guidelines, enabling the investment strategy to dynamically adjust its sensitivity to carbon risk as these objectives become more (or less) binding over time.
\begin{remark}\label{remark:negative_CI}
Recent contributions in the literature acknowledge the possibility that firms may report net-negative carbon emissions for instance due to the deployment of carbon removal technologies, such as Direct Air Capture (DAC) or Bioenergy with Carbon Capture and Storage (BECCS); the implementation of nature-based solutions, including afforestation, reforestation; the purchase of high-quality certified offsets that correspond to verifiable and additional carbon removals (see, e.g., \cite{bhatia2025emission} and \cite{verbist2025carbon}). Motivated by this evidence, our model allows the carbon intensity process $C_{i,t}$  to take negative values, and no penalisation is applied when this occurs (see equation \eqref{eq:penalisation}). From a modeling perspective, one may also extend our study to include incentives for stocks with negative carbon intensity. However, such an approach may be subject to criticism, as it could result in greenwashing practices.  
\end{remark}
It follows from It\^{o}’s formula
that the dynamics of $\tilde{X}^{\bmpi}=\{ \tilde{X}^{\bmpi}_t\}_{t\in[0,T]}$ is given by
\begin{equation}\label{eq:carbon_pen_wealth}
\dfrac{\de\tilde{X}_t^{\bmpi}}{\tilde{X}_t^{\bmpi}}=\left[r+\bmpi_t^\top\left(\bmmu- r\1\right)-\dfrac{1}{2}\bmpi_t^\top\left(\bfe(t,\bfC_t)\odot\bfD\bfD^\top\right)\bmpi_t\right]\de t+\bmpi_t^\top\bfSigma\de\bfZ_t,\quad\tilde{X}_0^{\bmpi}=x_0.
\end{equation}
Let $U$ be a CRRA utility function. We then formulate the following optimization problem
\begin{align}\label{eq:OPT_PROBLEM_CRRA_UTILITY}
\mbox{Maximise }\E^{\P^M}_{t,x,\bfc}[U(\tilde{X}^{\bmpi}_T)],\mbox{ over all }\bmpi\in\mcA,
\end{align}
where 
\begin{equation}
U(x)=\begin{cases}
\dfrac{x^{1-\delta}}{1-\delta},&\delta\in\left(0,1\right)\cup\left(1,+\infty\right),\\
\log(x),&\delta=1,
\end{cases}
\end{equation}
with $\delta$ being the risk aversion parameter. Here, $\E_{t,x,\bfc}^{\P^M}$ denotes the conditional expectation under $\P^M$ given $\tilde{X}_t=x$ and $\bfC_t=\bfc$. The value function of the optimization problem
\eqref{eq:OPT_PROBLEM_CRRA_UTILITY}, is given by
\begin{equation}\label{eq:value_function}
v\left(t,x,\bfc\right):=\sup_{\bmpi\in\mcA}\E^{\P^M}_{t,x,\bfc}[U(\tilde{X}_T^{\bmpi})].
\end{equation}
\subsubsection{Comparison with existing literature}\label{sec:comparison}
The proposed penalisation incorporates sustainability directly into preferences by increasing her effective risk aversion to stocks with high carbon intensity. Unlike \cite{rogers2013optimal}, this adjustment does not act on the variance-covariance matrix. The reason is that negative correlations among stocks with high carbon intensity would attenuate the penalisation, thereby misrepresenting the true risk preferences. To avoid such distortions, the penalisation relies solely on each asset’s realized variance. Moreover, in contrast to the approaches such as \cite{bolton2022net} and \cite{LeGuenedalRoncalli2023}), we do not impose a binding sustainability constraint on the investment portfolio. Instead, our method enables a trade-off between carbon intensity and market risk: a stock with high carbon intensity may still be held in the portfolio if its low volatility and large return sufficiently offset its carbon risk.\\
The proposed carbon penalisation admits two complementary interpretations. First, it can be viewed as a proportional cost applied to positions in carbon-intensive assets. Under this interpretation, the portfolio selection problem reflects a trade-off between a higher risk premium and the reputational or regulatory cost of holding assets with high carbon intensity that may prevent the fund from being classified as “green”. Second, the penalisation can be viewed as an endogenous adjustment of the investor’s risk aversion. In this case, the optimization problem can be recast  as  one with stochastic risk aversion, where the effective risk aversion coefficient is the sum of the utility-based parameter and the carbon penalisation term. As a result, exposure to carbon-intensive assets is naturally reduced. The implications of this mechanism for optimal strategies are discussed in details in Example \ref{example_2stocks} and Section \ref{sect:empirical_findings} below.\\
Moreover, our carbon penalisation mechanism is also related to the ESG preference frameworks of \cite{pastor2021sustainable} and \cite{pedersen2021responsible}. In \cite{pastor2021sustainable}, sustainability enters investor utility through an additive non pecuniary term depending on portfolio holdings, whereas in \cite{pedersen2021responsible} it enters through an additional reward term linked to the ESG score of the aggregate portfolio, so that the investor chooses along an enlarged mean-variance frontier that also accounts for sustainability. In our framework, instead, sustainability enters through a multiplicative exponential penalisation of terminal value of the investment portfolio, whose exponent can be interpreted as the cumulative carbon disutility generated by the trading strategy. Under logarithmic utility ($\delta=1$), equation \eqref{eq:OPT_PROBLEM_CRRA_UTILITY} becomes
\begin{equation}
\mathbb{E}^{\mathbb{P}^M}_{t,x,\bfc}\left[\log\left(X^{\bmpi}_T\right)-\dfrac{1}{2}\int_0^T\bmpi_s^\top\left(\bfe(s,\bfC_s)\odot\bfD\bfD^\top\right)\bmpi_s\de s\right],
\end{equation}
so that the carbon penalisation appears as an additive preference wedge in the objective function, in line with \cite{pastor2021sustainable}. For general CRRA preferences, this structure extends naturally to a nonlinear utility setting, in which the same carbon term acts as an exponential discount factor on utility. Hence, our contribution generalises, to some extent, the preference based approaches of \cite{pastor2021sustainable} and \cite{pedersen2021responsible} by introducing a penalisation that is both dynamic and state dependent, as it rely on the whole trading path, firms' carbon intensity, and asset specific variance.
\subsection{Solution to the optimisation problem}
The problem \eqref{eq:OPT_PROBLEM_CRRA_UTILITY} is solved by employing the dynamic programming principle. We denote by $\mcL$ the infinitesimal generator of the process $(\tilde{X}^{\bmpi},\bfC)$, that is
\begin{align}
\mcL F(t,x,\bfc)=&x\left[r+\bmpi^\top\left(\bmmu- r\1\right)-\dfrac{1}{2}\bmpi^\top\left(\bfe(t,\bfc)\odot\bfD\bfD^\top\right)\bmpi\right]\partial_{x}F(t,x,\bfc)\\
&+\frac{x^2}{2}\bmpi^\top\bfSigma\bfSigma^\top\bmpi \partial^2_{x}F(t,x,\bfc)+\mcLC F(t,x,\bfc),
\end{align}
for every sufficiently regular function $F$ and any constant control $\bmpi\in\mathbb{R}^d$, where $\mcLC$ denotes the infinitesimal generator of $\bfC$. We consider the following Hamilton-Jacobi-Bellman (HJB) equation
\begin{equation}\label{eq:HJB}
\begin{cases}
\displaystyle\sup_{\bmpi\in\R^{d}}\mcL v(t,x,\bfc)+\partial_tv(t,x,\bfc)=0, &(t,x,\bfc)\in[0,T)\times\R_+\times\mcD,\\
v(T,x,\bfc)=U(x), &(x,\bfc)\in\mathbb\R_+\times\mcD.
\end{cases}
\end{equation}

In the sequel, we prove that the value function, defined in equation \eqref{eq:value_function}, solves the equation \eqref{eq:HJB}. We begin our analysis with a verification result. In the discussion of the verification theorem we assume quite general dynamics for the carbon intensity process covering several interesting examples, discussed in Section \ref{sec:examples}. In particular, we assume that $\bfC$ has a semimartingale decomposition given by 
\begin{equation}
\bfC_t= \bfC_0 +\int_0^t \bm\Gamma_s \de s + \bm{M}_t,
\end{equation}
where $\{\bm\Gamma_t\}_{t \in [0,T]}$ is an $\R^d$-valued $\mathbb F$-predictable process such that $\E^{\P^M}[\int_0^T\|\bm\Gamma_s\| {\rm d} s] < \infty$, and $\{\bm{M}_t\}_{t \in [0,T]}$ is a $d$-dimensional square integrable $\mathbb{F}$-martingale which may have a continuous and a discontinuous parts, e.g. $\bm{M}_t=\bm{M}^{cont}_t + \bm{M}^{disc}_t$. 
To avoid technicalities we assume that possible jumps  of  $\mathbf{C}$ are bounded. 

\begin{thm}[Verification Theorem]\label{thm:ver_thm}
Suppose that $w(t,x, \bfc)$ is a classical solution of the HJB equation \eqref{eq:HJB} and that $|w(t, x,\bfc)|\le\bar{k}\left(1+|x|+|x|^{1-\delta}\right)$, for some constant $\bar{k}>0$.
Then:
\begin{itemize}
\item[(i)] $w(t,x,\bfc)\ge v(t, x, \bfc)$ for all $0\le t\le T$, $(x, \bfc) \in \R_+ \times \mathcal D$; 
\item[(ii)]  if there exists a measurable map $\bmpi^\star:[0,T]\times \mathcal D\to\mathbb R^d$ such that,
for every $(t,x,\bm c)\in[0,T)\times\mathbb R_+\times \mathcal D$, the value
$\bmpi^\star(t,\bm c)$ attains the supremum over $\bmpi\in\mathbb R^d$ in the HJB equation \eqref{eq:HJB}, and if the feedback strategy $
\bmpi_s^\star=\bmpi^\star(s,\bm C_s)$, for every $s\in[t,T]$,
is admissible, then
$w(t,x,\bm c)=v(t,x,\bm c)$. In particular, $\bmpi^\star$ is an optimal portfolio strategy.
\end{itemize}
\end{thm}
\begin{proof}
See Section \ref{sect:proof_of_ver_thm} in Appendix. 
\end{proof}
\begin{thm}\label{thm:general_thm}
Let $\varphi(t,\bfc)$ be the unique classical solution to the Cauchy problem
\begin{equation}\label{eq:PDE}
\begin{cases}
\partial_t\varphi(t,\bfc)+\mcLC\varphi(t,\bfc)+H(t,\bfc)\varphi(t,\bfc)=f(t,\bfc),&(t,\bfc)\in[0,T)\times\mcD,\\
\varphi(T,\bfc)=\ind[\delta\in(0,1)\cup(1,+\infty)],&\bfc\in\mcD,
\end{cases}
\end{equation}
where the functions $H(t,\bfc)$ and $f(t,\bfc)$ are defined as follows: 
\begin{align}
H(t,\bfc)&:=(1-\delta)\left[r+\dfrac{1}{2}\left(\bmmu-r\1\right)^\top \left(\delta\bfSigma\bfSigma^\top+\bfe(t,\bfc)\odot\bfD\bfD^\top\right)^{-1}\left(\bmmu-r\1\right)\right],\\
f(t,\bfc)&:=-\left[r+\dfrac{1}{2}\left(\bmmu-r\1\right)^\top\left(\bfSigma\bfSigma^\top+\bfe(t,\bfc)\odot\bfD\bfD^\top\right)^{-1}\left(\bmmu-r\1\right)\right]\ind[\delta=1].
\end{align}
Then, 
the optimal investment strategy is given by
\begin{equation}\label{eq:sol_CRRA}
\bmpi^\star(t,\bfc) =\left(\delta\bfSigma\bfSigma^\top+\bfe(t,\bfc)\odot\bfD \bfD^\top\right)^{-1} \left(\bmmu-r\1\right),
\end{equation}
and the value function satisfies
\begin{equation}
v(t,x,\bfc) =
\begin{cases}
\dfrac{x^{1-\delta}}{1-\delta}\varphi(t,\bfc), &\delta\in(0,1)\cup(1,+\infty),\\
\log(x)+\varphi(t,\bfc),&\delta=1.
\end{cases}
\end{equation}
\end{thm}
\begin{proof}
See Section \ref{sect:proof_general_thm} in Appendix.
\end{proof}
\begin{example}\label{example_2stocks}
To analyze the optimal investment strategy, we consider a case in which only two stocks, $S_1$ and $S_2$, are traded in the market, where $S_1$ is characterized by a positive carbon intensity $(C_1>0)$ and $S_2$ has zero carbon intensity $(C_2=0)$. Hence, equation \eqref{eq:sol_CRRA} reads as
\begin{align}
\label{eq:alloc_on_pi_1}\pi^\star_1(t,c_1)&=\dfrac{\mu_1-r}{\sigma_1^2\left[\delta\left(1-\rho^2\right)+\alpha_1c_1\right]}-\dfrac{\rho\left(\mu_2-r\right)}{\sigma_1\sigma_2\left[\delta\left(1-\rho^2\right)+\alpha_1c_1\right]},\\
\label{eq:alloc_on_pi_2}\pi^\star_2(t,c_1)&=-\dfrac{\rho\left(\mu_1-r\right)}{\sigma_1\sigma_2\left[\delta\left(1-\rho^2\right)+\alpha_1c_1\right]}+\dfrac{\left(\delta+\alpha_1c_1\right)\left(\mu_2-r\right)}{\delta\sigma_2^2\left[\delta\left(1-\rho^2\right)+\alpha_1c_1\right]}.
\end{align}
for every $(t,c_1)\in[0,T]\times\mathbb{R}_+$. When the carbon penalisation parameter $\alpha$ is set to zero, we recover the classical myopic Merton solution, namely
\begin{align}
\pi^\star_1(t,c_1)&=\dfrac{1}{\delta}\dfrac{\mu_1-r}{\sigma_1^2(1-\rho^2)}-\dfrac{1}{\delta} \dfrac{\rho(\mu_2-r)}{\sigma_1\sigma_2(1-\rho^2)}, \\
\pi^\star_2(t,c_1)&=-\dfrac{1}{\delta}\dfrac{\rho(\mu_1-r)}{\sigma_1\sigma_2(1-\rho^2)}+\dfrac{1}{\delta}\dfrac{\mu_2-r}{\sigma_2^2(1-\rho^2)},
\end{align}
for every $(t,c_1)\in[0,T]\times\mathbb{R}_+$. In the limiting case, where $\alpha\to\infty$, $\pi_1^\star(t,c_1)=0$ and $\pi_2^\star(t,c_1)=\frac{1}{\delta}\dfrac{\mu_2-r}{\sigma_2^2}$, meaning that $S_1$ is fully divested and the entire risky allocation is shifted to $S_2$ with zero carbon intensity. It is also interesting to comment on the case in which the correlation $\rho$ between the two stocks is equal to zero. Under this specification, equations \eqref{eq:alloc_on_pi_1} and \eqref{eq:alloc_on_pi_2} become, respectively:
\begin{equation}
\pi^\star_1(t,c_1)=\dfrac{1}{\delta+\alpha_1c_1}\dfrac{\mu_1-r}{\sigma_1^2},\quad\pi^\star_2(t,c_1)=\dfrac{1}{\delta}\dfrac{\mu_2-r}{\sigma_2^2}.
\end{equation}
Consequently, any increase in $\alpha$ leads to a reduction in the exposure to the carbon-intensive asset, while the allocation to the zero-carbon emission asset remains unchanged, resulting in an increase in the exposure to the risk-free asset.
\end{example}
\section{Examples: models for carbon intensities and value functions}\label{sec:examples}
The limited data available on asset carbon intensity do not allow us to reliably specify a model for this process. Therefore, we discuss several plausible dynamics for the carbon intensity process and, for each case, provide the corresponding characterization of the value function.  
\subsection{Continuous time finite-state Markov chain}
As a first stylized specification, one may assume that the carbon intensity process is piecewise constant and subject only to occasional but significant changes. Under this interpretation, a natural modeling choice is to represent $\bfC$ as a continuous time finite state Markov chain. This should be understood as a regime-switching approximation, where the states represent carbon-intensity regimes or classes.
More precisely, we assume that \(\mathbf C\) has 
$\ell$ possible states, with state space $\mathcal{D}=\{\bfa_1, \dots, \bfa_\ell\}$ and $\bfa_k\in\mathbb{R}^d$ for $k\in\{1,\dots,\ell\}$. We denote by $Q=(q_{k,l})_{k,l\in\{1,\dots,\ell\}}$ the infinitesimal generator of $\bfC$, with $q_{k,l}\ge 0$ and $q_{k,k}=-\sum_{l\ne k}q_{k,l}$, and let $\Pi=(\Pi_1,\dots,\Pi_\ell)$ be its initial distribution. Hence, the infinitesimal generator of $\bfC$ is given by
\begin{equation}
\mcLC F(\bfa_k)=\sum_{l} q_{k, l} F(\bfa_l),
\end{equation}
for every function $F: \mathcal{D} \to \mathbb{R}$.
\begin{lem}\label{cor:MARKOV_CHAIN_P}
Let $\alpha_i(t)$ be continuous functions on $[0,T]$ for every $i=1,\dots,d$. Then, $\varphi_k(t):=\varphi(t,\bfa_k)\in\mathcal{C}^{1}([0,T])$ are the unique solutions to the system of ODEs
\begin{equation}\label{eq:ODE_system}
\partial_t\varphi_k(t)+\sum_{l}q_{k,l}\varphi_l(t)+H_k(t)\varphi_k(t)=f_k(t),
\end{equation}
with terminal conditions $\varphi_k(T)=\ind[\delta\in(0,1)\cup(1,+\infty)]$ for all $k\in\{1,\dots,\ell\}$, where $H_k(t):=H(t,\bfa_k)$ and $f_k(t):=f(t,\bfa_k)$.
\end{lem}
\begin{proof}
Equation \eqref{eq:ODE_system} defines a system of linear ODEs. The result then follows from \cite[Theorem 3.9]{teschl2012ordinary}.
\end{proof}

\subsection{Ornstein-Uhlenbeck process}
A natural continuous-time specification for carbon intensity is provided by a mean-reverting Ornstein-Uhlenbeck (OU) process. Indeed, the few available data points suggest that carbon intensity tends to fluctuate around a stable level. This behavior may be attributed to the fact that a substantial shift in a firm’s carbon intensity typically results from a technological shock, a change in production management, or the adoption of a green agenda aimed at lowering emissions. In particular, the implementation of a green agenda may even result in net negative emission levels (see Remark \ref{remark:negative_CI}). We therefore model the dynamics of carbon intensity as follows:
\begin{equation}\label{eq:OU_PROCESS}
\de C_{i,t}=\kappa_i\left(\bar{C}_i(t)-C_{i,t}\right)\de t+\lambda_i\de W_{i,t},\quad C_{i,0}=c_{i,0}, 
\end{equation}
where $\bar{C}_i(\cdot):[0,T]\to[0,+\infty)$ is a decreasing function representing the long-term carbon intensity level, $\kappa_i>0$ governs the strength of attraction toward the time-dependent target $\bar{C}_i$, and $\lambda_i$ is the volatility of the $i$-th carbon intensity, for every $i=1,\dots,d$. Specifically, the decreasing time-dependent target $\bar{C}_i(t)$ reflects the long-term reduction in carbon intensity that firms typically pursue through decarbonisation strategies in order to align with Net Zero emissions commitments. Such firms’ decarbonisation strategies lower their expected future carbon intensity and therefore attenuate the penalization in the optimal portfolio problem, which, for a given risk-return trade-off, may translate into higher optimal portfolio weights (see equation \eqref{eq:sol_CRRA}). Furthermore, the corresponding rate of attraction $\kappa_i$ around the target reflects how strongly the firm reacts to environmental pressures or incentives -- such as carbon regulations, investor scrutiny, or reputational concerns -- that drive decarbonization efforts over time. $\bfW=\left(W_{1,t},\dots,W_{d,t}\right)^\top$ is a vector of uncorrelated standard $\mathbb{F}$-Brownian motions independent of $\bfZ$, as required in order to satisfy the assumptions made on $\bfC$ in Section \ref{sec:setup}. Under this model specification, the infinitesimal generator of $\bfC_t$ is given by
\begin{equation}\label{eq:generatorOU}
\mcLC F(\bfc)=\sum_{i=1}^d\kappa_i\left(\bar{C}_i(t)-c_{i}\right)\partial_{c_i}F(\bfc)+\dfrac{1}{2}\sum_{i=1}^d\lambda_i^2\partial^2_{c_i}F(\bfc),
\end{equation}
for every $F\in\mathcal{C}^{2}(\mathcal{D})$, with $\mathcal{D}=\R^d$.
\begin{lem}\label{cor:OU_P}
Let $\mcLC$ as in equation \eqref{eq:generatorOU} with $\lambda_i>0$ for every $i\in{1, \dots, d}$, and let $\alpha_i(t),\,\bar{C}_i(t)$ be Lipschitz continuous functions on $[0,T]$ for every $i\in{1, \dots, d}$. Then, the problem \eqref{eq:PDE} has a unique solution $\varphi(t,\bfc)\in \mathcal{C}^{1,2}([0,T)\times\R^d) \cap \mathcal{C}([0,T]\times \R^d)$. 
\end{lem}
\begin{proof}
Note that, under the assumption that $\alpha_i(t)$ is Lipschitz continuous on $[0,T]$ for every $i\in{1, \dots, d}$, it holds that $H(t,\bfc)$ and $f(t, \bfc)$ are Lipschitz continuous on $[0,T]\times \R^d$. Hence the result follows from  \cite[Theorem 8.2.1, Chapter 8]{krylov1996lectures}. 
\end{proof}
\subsection{Exponential of Ornstein-Uhlenbeck process}
Empirical observations on the released carbon intensities of assets listed in the S\&P $500$, which we use for the empirical analysis (see Section \ref{sect:EMPIRICAL_FIRST_PART}), suggest that these values may remain positive over time. Hence, one could retain the mean-reversion property introduced in equation \eqref{eq:OU_PROCESS} while ensuring that the carbon intensity process remains positive over time. In this case, a natural modeling choice is to assume that $C_{i,t}$ follows exponential of the OU process given by equation \eqref{eq:OU_PROCESS}. Specifically, its dynamics are given by
\begin{equation}
\dfrac{\de C_{i,t}}{C_{i,t}}=\left[\kappa_i\left(\bar{C}_i(t)-\ln C_{i,t}\right)+\dfrac{\lambda_i^2}{2}\right]\de t+\lambda_i\de W_{i,t},\quad C_{i,0}=c_{i,0}, 
\end{equation}
for every $i=1,\dots,d$. Hence, 
\begin{equation}\label{eq:generator_expOU}
\mcLC F(\bfc)=\sum_{i=1}^d c_i\left[\kappa_i\left(\bar C_i(t)-\ln c_{i}\right)+\dfrac{\lambda_i^2}{2}\right]\partial_{c_i}F(\bfc)+\dfrac{1}{2}\sum_{i=1}^d\lambda_i^2c_i^2\partial^2_{c_i}F(\bfc),
\end{equation}
for every $F\in\mathcal{C}^{2}(\mathcal{D})$, with $\mathcal{D}=\R^d_+$.
\begin{lem}\label{cor:EXP_OU_P}
Let $\mcLC$ as in equation \eqref{eq:generator_expOU} with $\lambda_i>0$ and let $\alpha_i(t)$, $\bar{C}_i(t)$ be Lipschitz continuous functions on $[0,T]$ for every $i\in{1, \dots, d}$. Then the problem \eqref{eq:PDE} has a unique solution $\varphi\in\mcC^{1,2}[0,T)\times\mathbb{R}^d_+)\cap\mcC([0,T]\times\mathbb{R}_+^d)$.
\end{lem}
\begin{proof}
The proof of this result follows from the same lines of the proof of Lemma \ref{cor:CIR_P}, see also Section \ref{sect:proof_CIR_P} in Appendix. 
\end{proof}
\subsection{Cox-Ingersoll-Ross process}\label{sect:CIR_case}
An alternative positive mean-reverting specification to the exponential Ornstein-Uhlenbeck process is provided by the Cox-Ingersoll-Ross (CIR) process, whose dynamics are given by
\begin{equation}\label{eq:CIR_PROCESS}
\de C_{i,t}=\kappa_i\left(\bar{C}_i(t)-C_{i,t}\right)\de t+\lambda_i\sqrt{C_{i,t}}\de W_{i,t},\quad C_{i,0}=c_{i,0}, 
\end{equation}
for every $i=1,\dots,d$, with $\bar C_i(\cdot):[0,T]\to[0,+\infty)$ and $\kappa_i,\,\lambda_i>0$. In this case, the infinitesimal generator of $\bfC_t$ reads as
\begin{equation}\label{eq:generatorCIR}
\mcLC F(\bfc)=\sum_{i=1}^d\kappa_i\left(\bar{C}_i(t)-c_{i}\right)\partial_{c_i}F(\bfc)+\dfrac{1}{2}\sum_{i=1}^d\lambda_i^2c_i\partial^2_{c_i}F(\bfc),
\end{equation}
for every $f\in\mathcal{C}^{2}(\mathcal{D})$, with $\mathcal{D}=\R_+^d$.
\begin{lem}\label{cor:CIR_P}
Let $\mcLC$ as in equation \eqref{eq:generatorCIR} and let $\alpha_i(t),\,\bar{C}_i(t)$ be Lipschitz continuous functions on $[0,T]$ for every $i\in{1, \dots, d}$. Assume that Feller conditions hold, that is, $2\kappa_i\inf_{t\in[0,T]}\bar{C}_i(t)\ge\lambda_i^2$ for all $i\in\{1,\ldots,d\}$. Then the problem \eqref{eq:PDE} has a unique solution $\varphi\in\mcC^{1,2}[0,T)\times\mathbb{R}^d_+)\cap\mcC([0,T]\times\mathbb{R}_+^d)$.
\end{lem}
\begin{proof}
See Section \ref{sect:proof_CIR_P} in Appendix.
\end{proof}
\section{Static empirical results and dynamic numerical analysis on the optimal investment fund}\label{sect:EMPIRICAL_FIRST_PART}
Because of the limited availability of carbon intensity data, we distinguish between two different types of numerical analysis. In Section \ref{sect:empirical_findings}, we conduct a static empirical analysis in which carbon intensity is not modeled dynamically, but is fixed at its most recent available value for each stock. In Section \ref{sect:dynamic}, we instead present an illustrative simulation exercise in which carbon intensity is assumed to evolve according to a CIR process specified in equation \eqref{eq:CIR_PROCESS}, in order to show how the carbon penalized-portfolio allocation mechanism operates over time.
\subsection{Numerical analysis based on real data}\label{sect:empirical_findings}
In this section, we provide a static empirical analysis of the optimal carbon-penalised portfolio. The empirical implementation is based on a sample of $d=34$ stocks from the S\&P 500 index, and the stock-price dynamics in equation \eqref{eq:risk_assets_dyn} are calibrated using daily closing prices over the period from January $1$, $2015$ to December $31$, $2024$, sourced from the Morningstar database. The resulting market parameters are then used to compute the optimal carbon-penalised strategy $\bmpi^\star$ at $t=0$. Owing to the limited availability of carbon intensity data, we do not attempt to calibrate $\bfC$; instead, for each stock we take its most recent available carbon intensity observation, released on December $31$, $2024$.\\
Moreover, we take a constant and homogeneous penalty for each stock in the portfolio,  that is $\alpha_i(t)=\alpha$ for every $t\in[0,T]$, and every $i=1,\dots,d$.\\

\begin{figure}[htbp]
\centering
\includegraphics[width=0.40\linewidth]{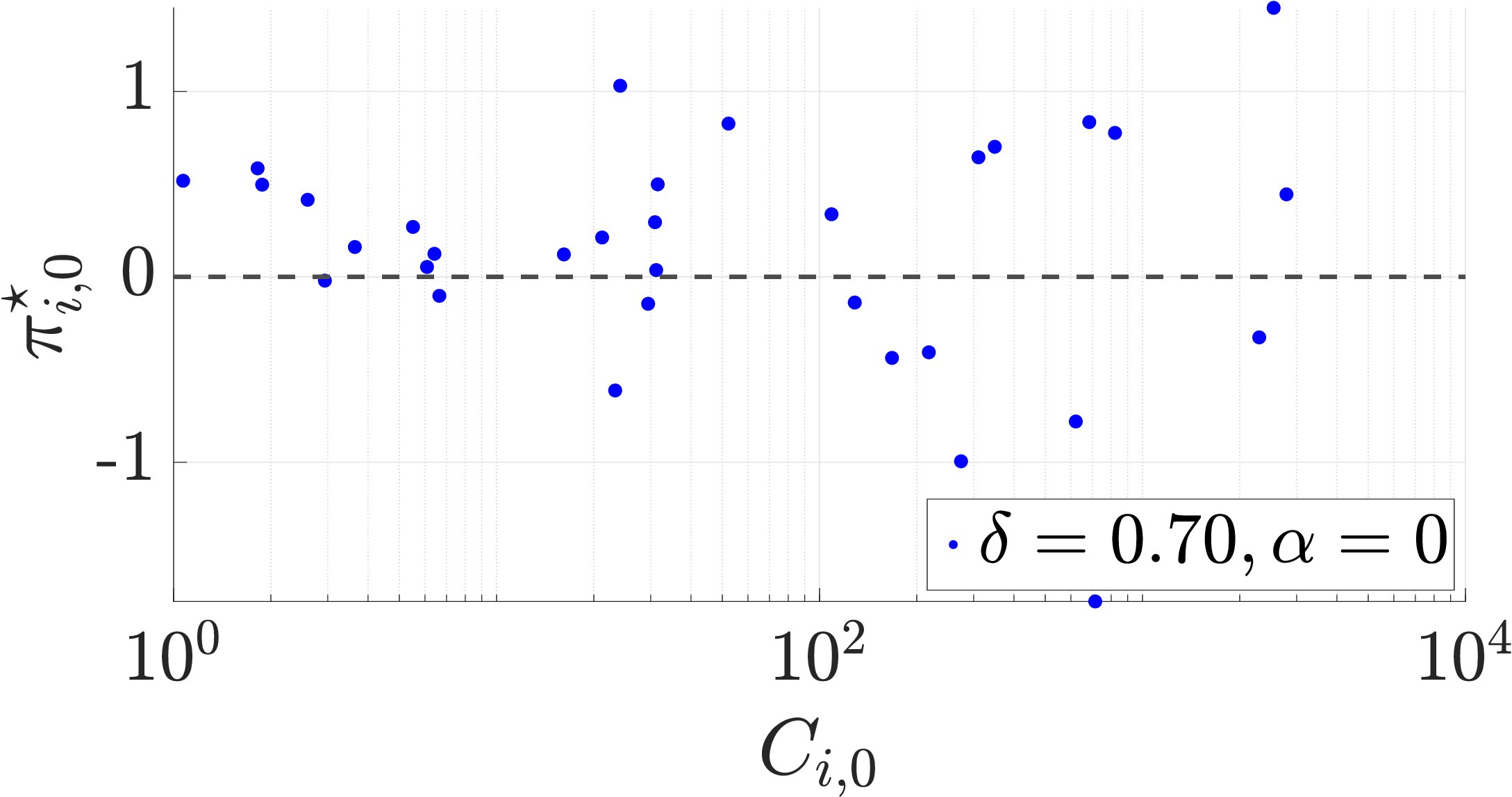}
\hspace{0.02\linewidth}
\includegraphics[width=0.40\linewidth]{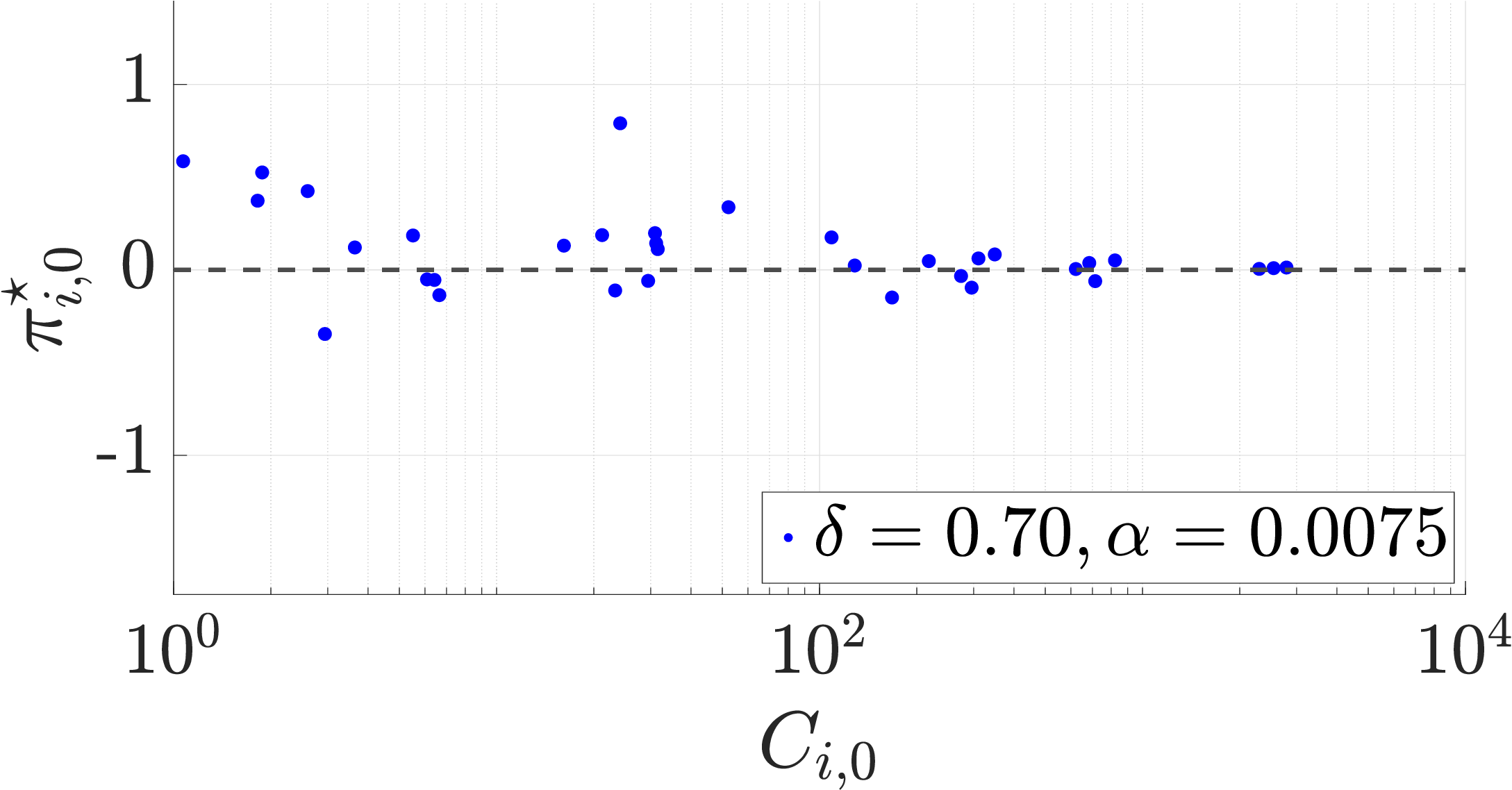}
\vspace{0.5cm}
\includegraphics[width=0.40\linewidth]{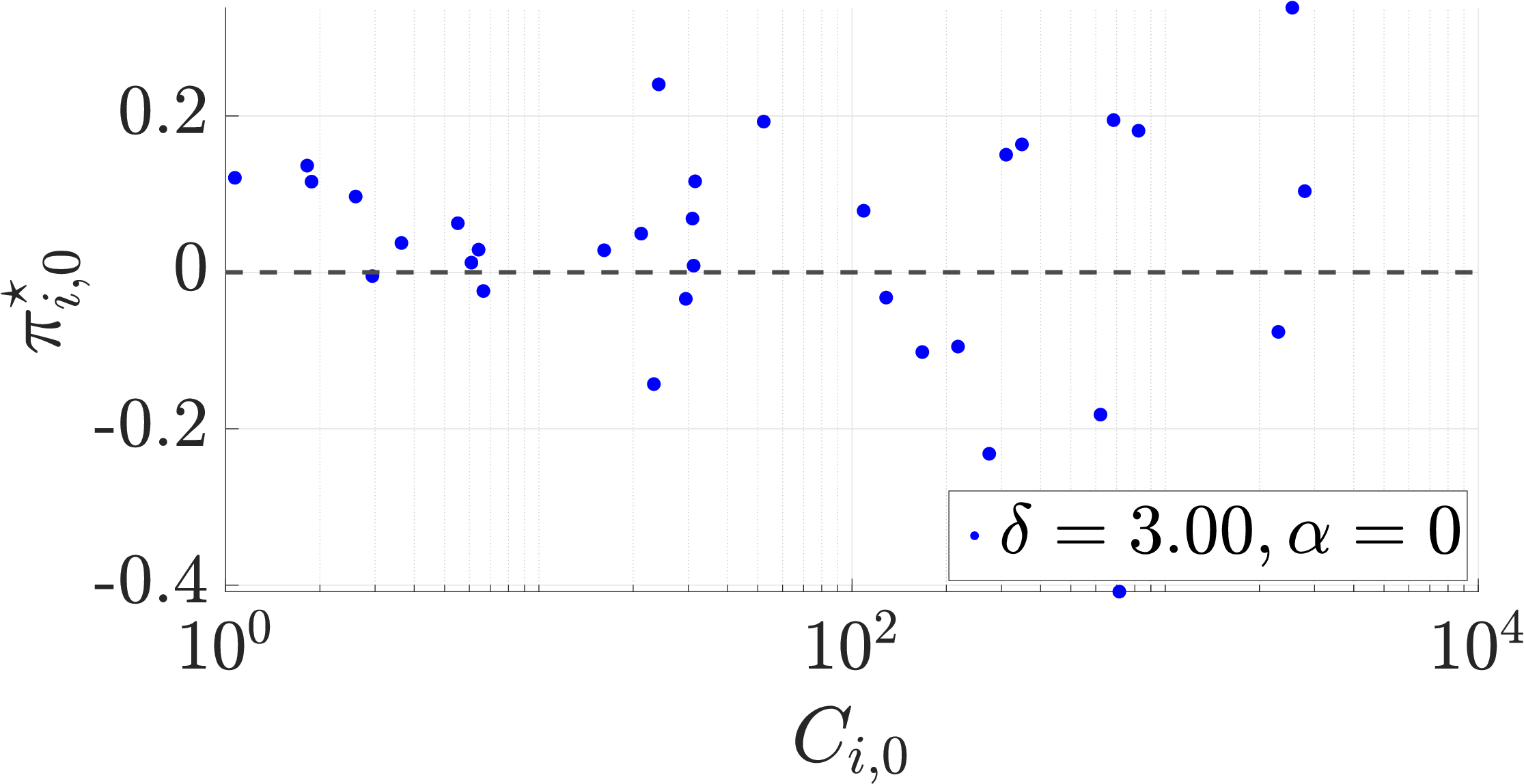}
\hspace{0.02\linewidth}
\includegraphics[width=0.40\linewidth]{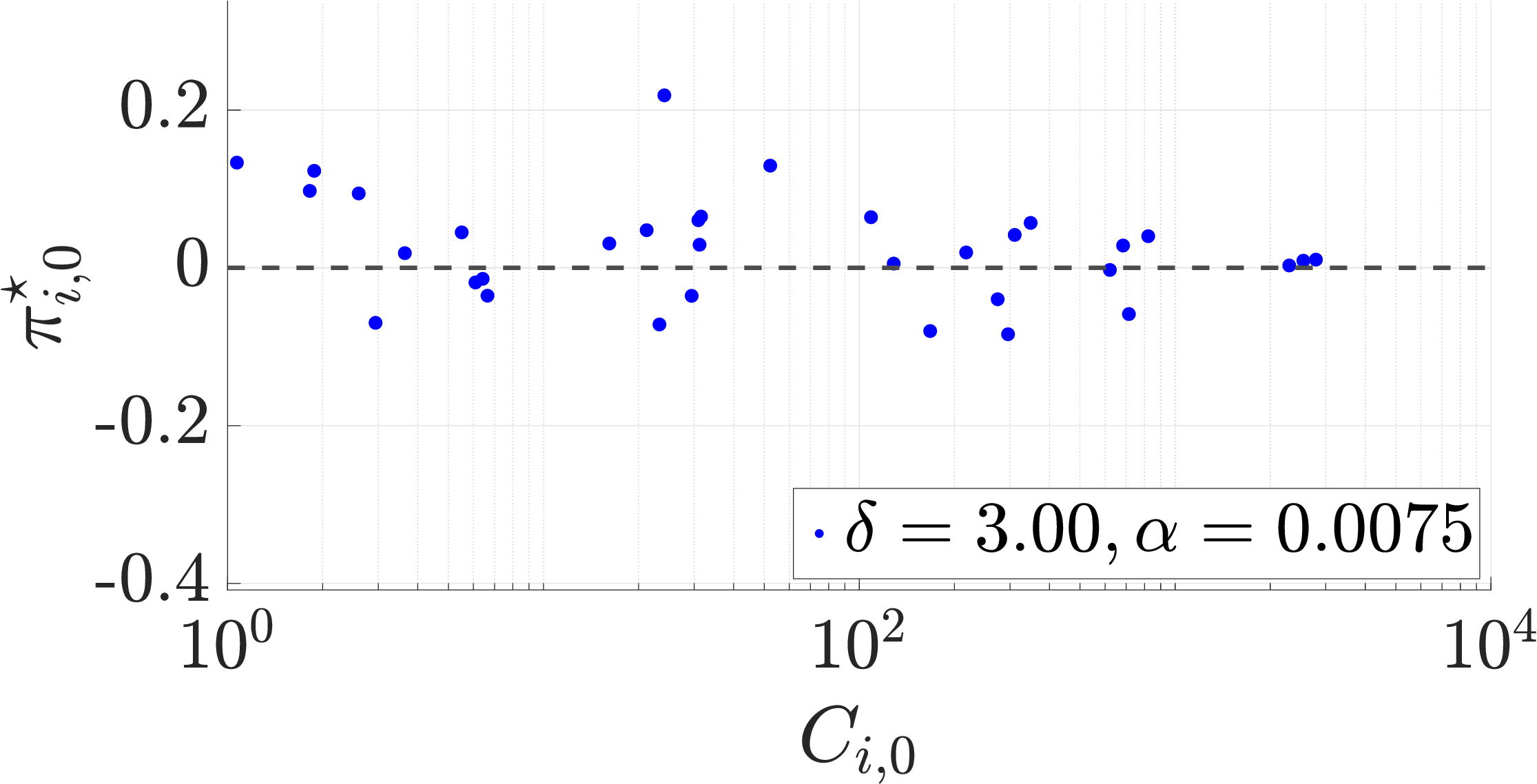}
\caption{Scatter plots displaying the carbon intensity of the $i$-th stock ($x$-axis on logarithmic scale) and the optimal portfolio weights ($y$-axis) for different levels of $\delta$ and $\alpha$.}\label{fig:sens_analysis_optimal_weights}
\end{figure}
Scatter plots in Figure \ref{fig:sens_analysis_optimal_weights} illustrate the relationship between the carbon intensity $C_{i,0}$ on the $x$-axis and the corresponding optimal portfolio weight $\pi^\star_{i,0}$ on the $y$-axis of the $i$-th stock for every $i=1,\dots,d$, for different levels of the  risk aversion $\delta$ and carbon penalisation $\alpha$. Each panel corresponds to a specific combination of $\delta\in\{0.7,\,3\}$ and $\alpha\in\{0,\,0.0075\}$, allowing for a comparative analysis of how increasing values of carbon aversion affect the allocation strategy. The results clearly show that, as the carbon penalisation $\alpha$ increases, the optimal weights of carbon-intensive stocks decrease in absolute value, regardless of the level of the risk aversion parameter $\delta$.
Increasing values of $\delta$ do not modify the impact of carbon penalisation; instead, they reduce the investment portfolio's exposure to risky assets, which in turn leads to lower leverage.
\begin{figure}[htbp]
\centering
\includegraphics[width=0.45\linewidth]{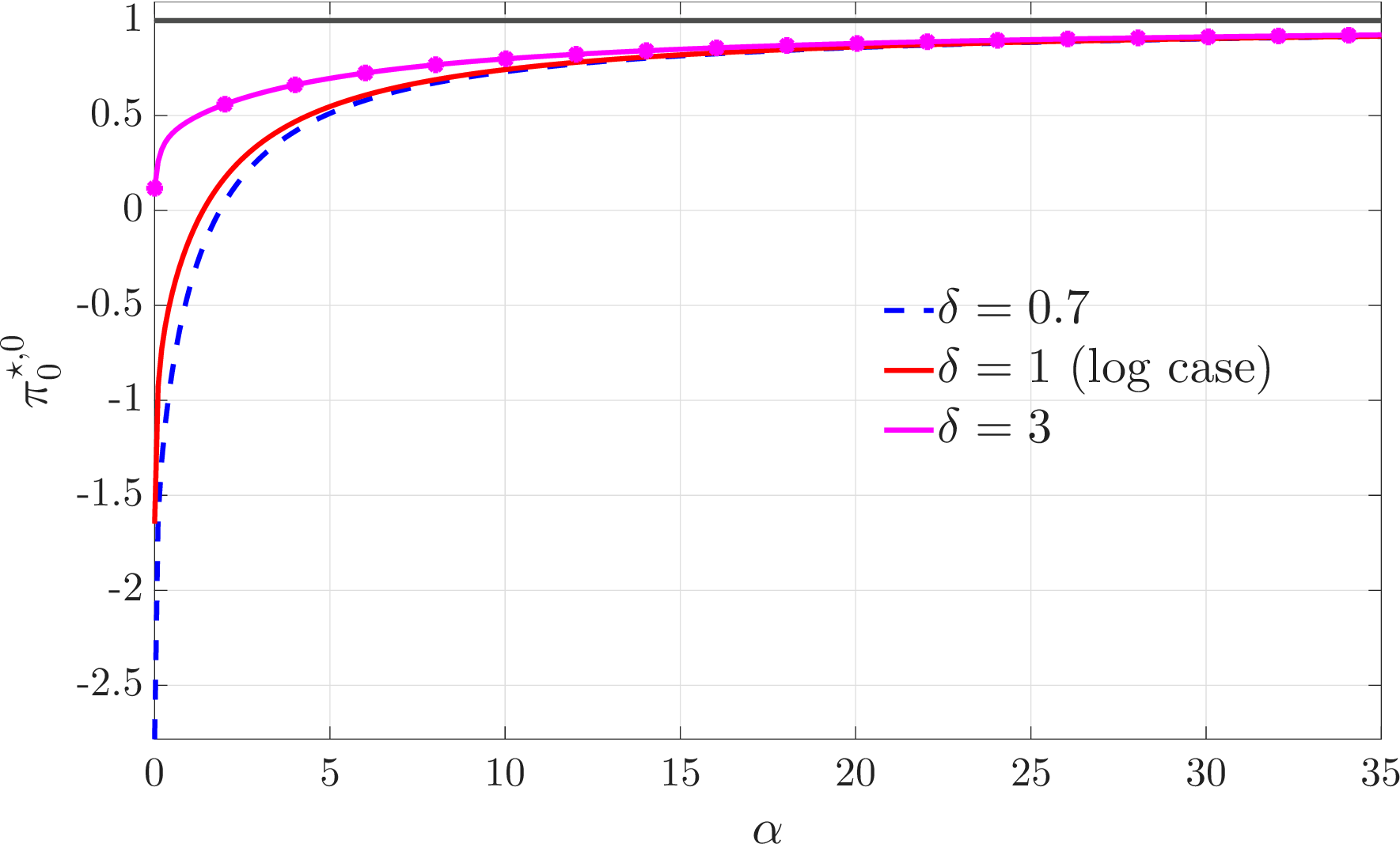}
\hfill 
\includegraphics[width=0.45\linewidth]{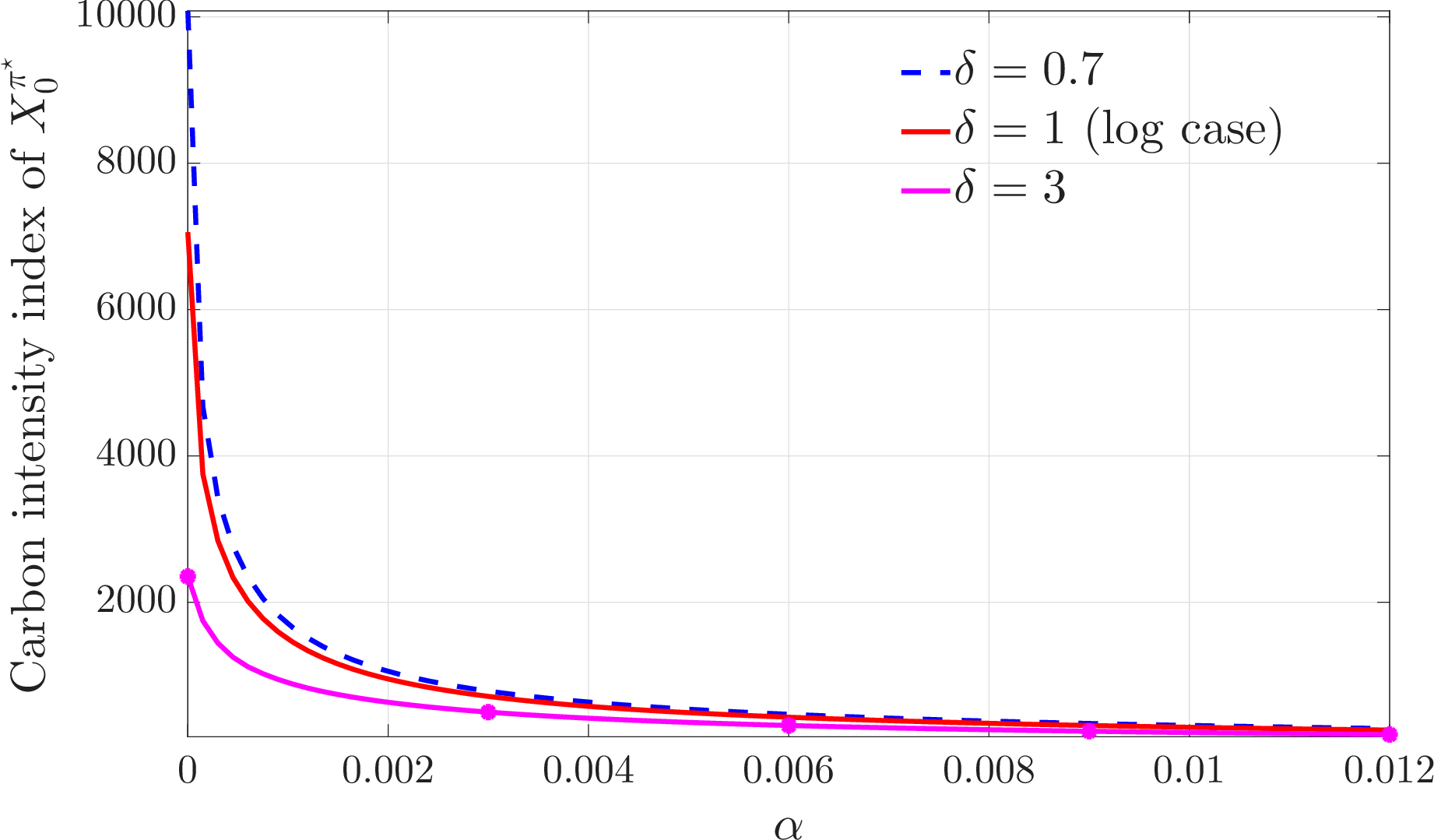}
\caption{Optimal portfolio weight in the bank account, $\pi^{0,\star}_0=1-\sum_{i=1}^d\pi^{\star}_{i,0}$, as a function of the carbon penalisation $\alpha$, for different levels of market risk aversion $\delta$ (left panel). Carbon intensity index of the optimal portfolio at $t=0$ as a function of carbon penalisation $\alpha$, for different levels of market risk aversion $\delta$. The portfolio carbon intensity index is computed as $\sum_{i=1}^d|\pi^\star_{i,0}|C_{i,0}$ (right panel).
\label{fig:opt_weight_B}}
\end{figure}
This is also confirmed by the left panel of Figure \ref{fig:opt_weight_B}. We denote by
$\pi^{0,\star}_0 := 1-\mathbf 1^\top \bmpi^\star_0$
the optimal weight invested in the bank account at time $t=0$. The left panel of Figure \ref{fig:opt_weight_B}
illustrates the relationship between carbon penalisation $\alpha$ and the corresponding optimal portfolio allocation to bank account $\pi_0^{0,\star}$. When $\alpha=0$, or takes small values, the portfolio may exhibit leverage, but only in case of low risk aversion or logarithmic preferences; for higher levels of risk aversion (e.g., $\delta=3$), no leverage is observed. As $\alpha$ increases, the investment in the bank account rises, implying that the overall exposure to risky assets $\1^\top\bmpi^{\star}_0$ decreases, regardless of the level of the risk aversion parameter $\delta$. If all assets have strictly positive carbon intensity, then in the limit as $\alpha \to \infty$ the risky allocation vanishes and $\pi_0^{0,\star}$ tends to $1$ implying that the entire investment portfolio is held in the bank account.\\

To further quantify the impact of $\alpha$ on the environmental sustainability of the investment portfolio, the right panel of Figure \ref{fig:opt_weight_B} reports the carbon intensity index of the optimal portfolio as a function of $\alpha$, for different levels of $\delta$. These plots confirm that the carbon intensity index decreases as $\alpha$ increases and that a very small penalisation is enough to considerably diminish the carbon intensity index of the portfolio. Moreover, consistent with the results shown in Figure \ref{fig:sens_analysis_optimal_weights}, lower levels of $\delta$ lead the strategy to take on greater exposure to risky assets and, for a given $\alpha$, result in a smaller reduction in carbon intensity.
\subsection{Numerical analysis based on simulated data}\label{sect:dynamic}
In this section, we complement the previous empirical exercise with an illustrative simulation study aimed at analyzing the dynamics of the optimal portfolio weights $\bmpi^\star$. Since the available carbon intensity data do not allow for a reliable calibration, we consider a stylized setting in which the sustainable investment fund consists of four stocks, whose corresponding carbon intensity processes follow the CIR dynamics specified in equation \eqref{eq:CIR_PROCESS}, with a time-independent long-run mean, that is, $\bar{C}(t)=\bar{C}\g 0$ for every $t\in[0,T]$. The parameters of the four stocks and of the associated carbon intensity processes are reported in Table \ref{Market_parameters}. Moreover, we assume that $\delta=1$ and, for the sake of simplicity, as in Section \ref{sect:empirical_findings}, we adopt a constant and homogeneous penalty across all stocks in the portfolio, that is, $\alpha_i(t)=\alpha=0.0025$ for every $t\in[0,T]$ and every $i=1,\dots,d$.\\
\begin{table}[h]
\centering
\begin{subtable}{0.24\textwidth}
\centering
\begin{tabular}{ccc} 
\Xhline{1.2pt}
            & $\mu$     & $\sigma$      \\ 
\hline
$S_1$       & $0.25$ & $0.30$ \\    
$S_2$       & $0.15$ & $0.25$ \\ 
$S_3$       & $0.10$ & $0.20$\\
$S_4$       & $0.08$ & $0.16$\\
\Xhline{1.2pt}
\end{tabular}
\caption{Parameters for the stock prices.}
\label{Stocks_parameters}
\end{subtable}
\hfill
\begin{subtable}{0.33\textwidth}
\centering
\[\bm{\rho} = \begin{bmatrix}
1.00 & 0.44 & 0.39 & 0.32 \\
0.44 & 1.00 & 0.30 & 0.33 \\
0.39 & 0.30 & 1.00 & 0.31 \\
0.32 & 0.33 & 0.31 & 1.00
\end{bmatrix}
\]
\caption{Correlation matrix $\bm{\rho}$.}
\label{Rho_matrix}
\end{subtable}
\hfill
\begin{subtable}{0.38\textwidth}
\centering
\begin{tabular}{ccccc} 
\Xhline{1.2pt}
                       & $c$    & $\beta$& $\kappa$& $\lambda$       \\ 
\hline
$C_1$                  & $5000$ & $2500$ & $0.05$  & $3.0$   \\
$C_2$                  & $4000$ & $2000$ & $0.05$  & $3.0$\\ 
$C_3$                  & $3000$ & $1500$ & $0.05$  & $3.0$\\
$C_4$                  & $1000$ &  $500$ & $0.05$  & $3.0$\\
\Xhline{1.2pt}
\end{tabular}
\caption{Parameters for the carbon intensities.}
\label{CI_parameters}
\end{subtable}
\caption{Model parameters for the stocks, correlation matrix, and carbon intensities.}\label{Market_parameters}
\end{table}

The results are shown in Figure \ref{fig:dyn}, suggesting that our model successfully balances the risk–return trade-off of each stock with its corresponding carbon risk. In particular, stock $4$, which maintains the lowest carbon intensity throughout the investment horizon, is consistently assigned the highest weight, reflecting the model’s sensitivity to sustainability criteria. However, the proposed portfolio selection methodology does not reduce to a naive exclusion of carbon intensive assets. Indeed, although stock 1 has the highest carbon intensity, it is assigned higher optimal portfolio weights than stocks $2$ and $3$ throughout most of the investment horizon. This is justified by the fact that stock $1$ exhibits the highest market price of risk, indicating that its financial attractiveness is sufficient to outweigh its carbon risk. This outcome highlights that our penalisation mechanism does not rigidly exclude high carbon intensive assets, but instead adjusts allocations based on a balanced evaluation of both environmental and financial attributes.
\begin{figure}[H]
\centering
\includegraphics[width=0.49\linewidth]{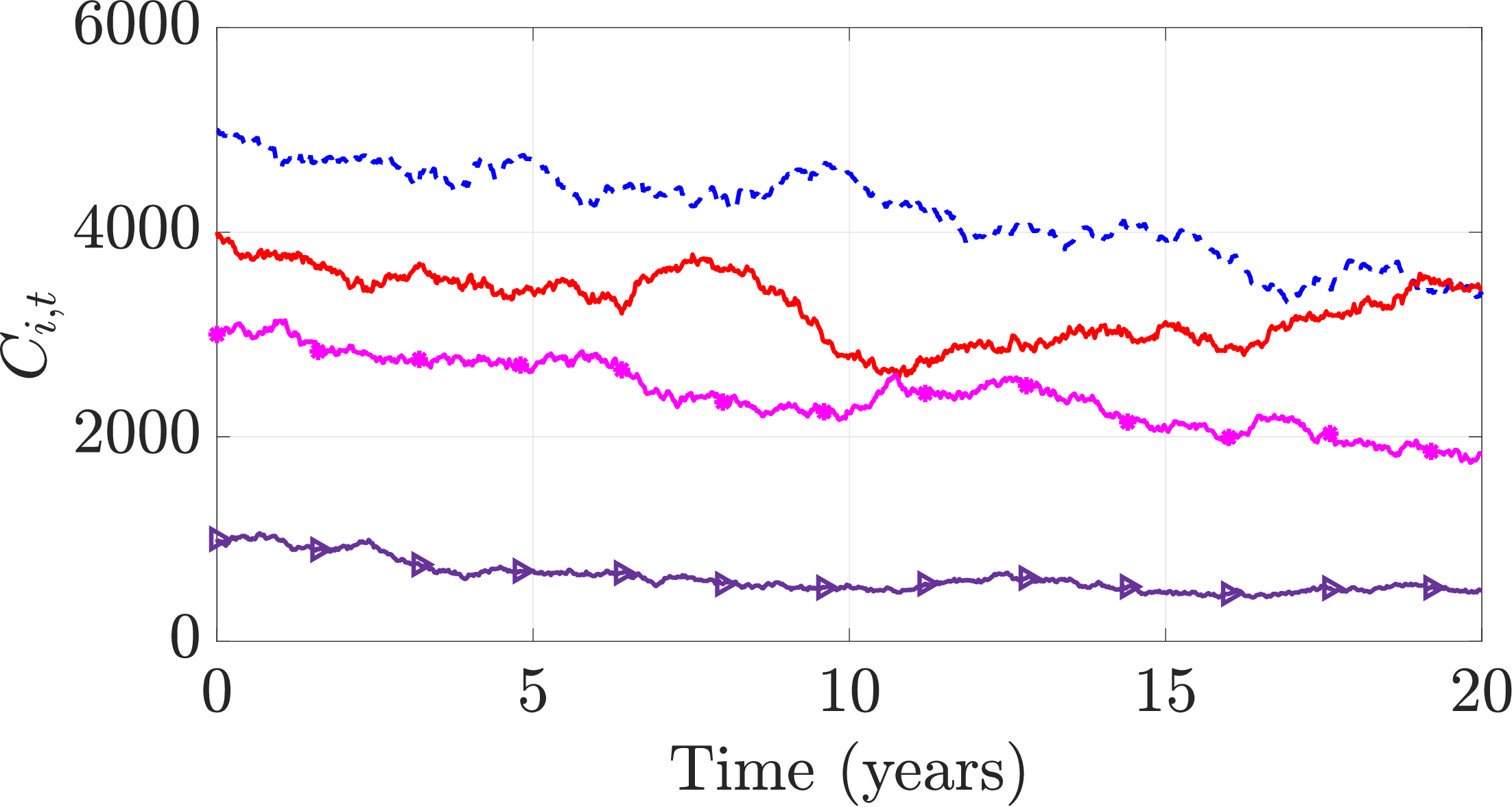}
\hfill 
\includegraphics[width=0.49\linewidth]{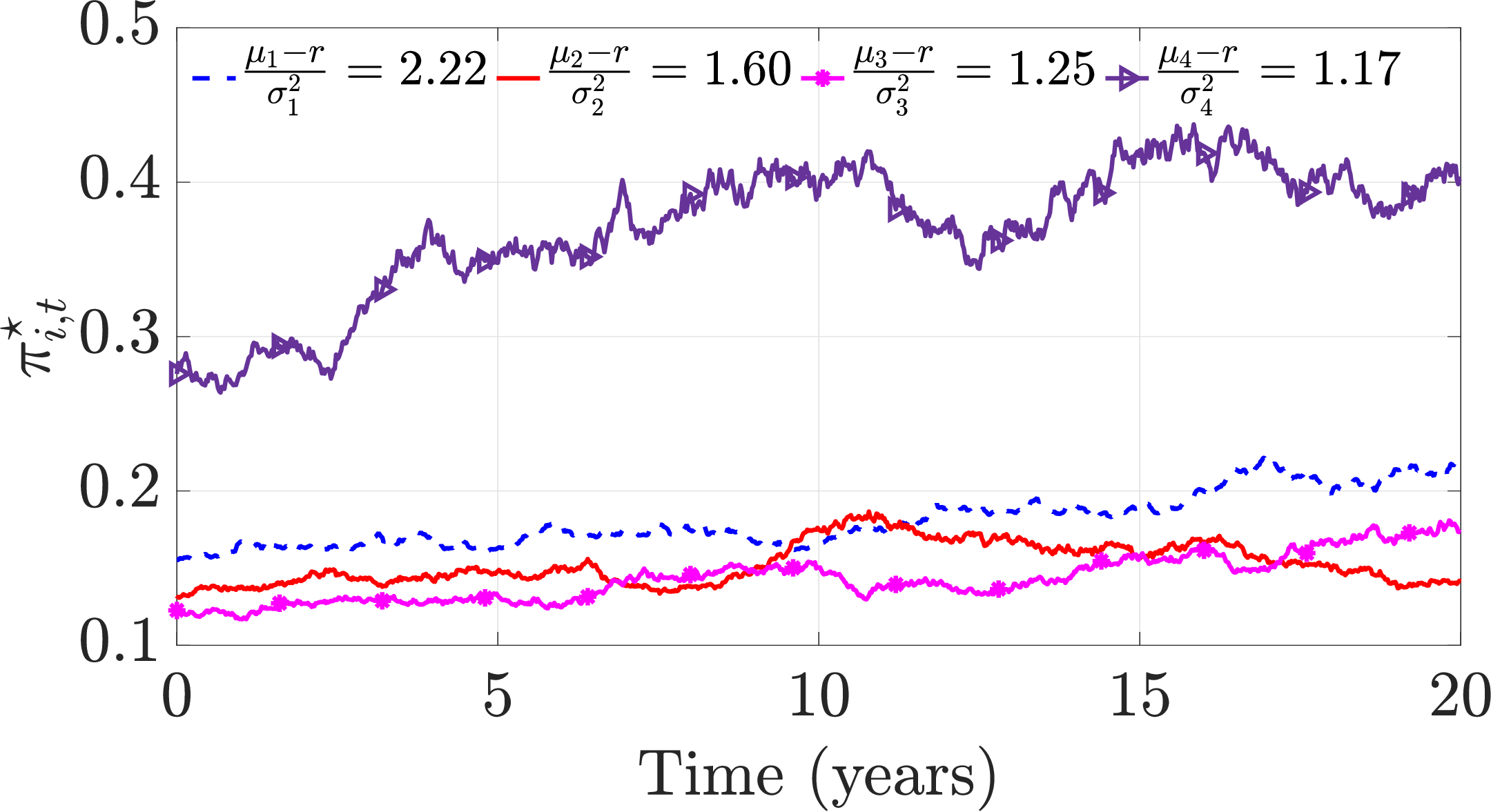}
\caption{The left panel shows the paths of the carbon intensity $\bfC$ of four stocks, simulated according to the CIR process (see Section \ref{sect:CIR_case}). The right panel shows the corresponding dynamics of the optimal portfolio weights $\bmpi^{\star}$.}
\label{fig:dyn}
\end{figure}
\section{Risk-minimizing strategies for unit linked policies under carbon risk}\label{sect:locally_risk_minimizing}
In the sequel, we consider an insurance company who proposes to its clients unit linked life insurance contracts with a sustainable label and seeks to hedge their payoffs. As usual in this type of life insurance policy, the value is linked to the performance of an underlying fund, and its maturity relates to the policyholder's remaining lifetime. 
We assume that the financial underlying of such policies is the investment fund derived in Section \ref{sec:optimal_inv}, and we analyse unit linked contract of pure endowment and term insurance types (this also provides results for endowment insurance contracts which are linear combinations of the two above).\\
The insurance company faces exposure to financial, carbon, and mortality risks. 
Since the carbon emissions of the firms and mortality risk do not correspond to tradeable assets, the underlying combined financial/insurance market is incomplete. Hence, the insurance company is unable to replicate perfectly the payoff of the policy via a self financing strategy. Part of the risk related to this contract remains unhedgeable, and generates a loss that the insurance company would like to properly \emph{minimize}. In this paper, we follow a quadratic approach and determine the strategies that minimize the residual cost (due to incompleteness) in $L^2$-sense.  

\subsection{A market for green unit linked contracts}
We consider the financial market of Section \ref{sec:setup}, and extend it by including the mortality risk of an insured individual. To do this, we consider a probability space $(\Omega^I, \mathcal H, \P^I)$. We consider an individual of age $\iota$, and let \(\tau\) be a nonnegative random variable representing their remaining lifetime having deterministic hazard rate $\{\gamma(t), t \ge 0\}$ such that 
for all $t\geq 0$,
\begin{align}
\P^I\left(\tau > t  \right) 
&=  \exp\left(- \int_0^{t} \gamma(s) \de s \right).
\end{align}
We set $H_t := \mathbbm{1}_{\{\tau \leq t\}}$, for every $t \geq 0$, and define the filtration $\mathbb{H} := \{\mathcal{H}_t\}_{ t \geq 0}$ generated by the process $H$, i.e.,
\begin{equation}\label{eq:filtrationH}
\mathcal{H}_t := \sigma\{ H_s, 0 \leq s \leq t \}, \quad \text{ for all } t \geq 0. 
\end{equation}
In the sequel, we assume that $\E^I\left[\int_0^T\gamma(s)\de s\right]<\infty$, then $\left\{H_t-\int_0^{t\wedge\tau} \gamma(s) \,\de s,\,t\geq 0\right\}$, is a $(\mathbb{H}, \P^I)$-martingale. Note that we can interpret $\P^I$ as the probability measure corresponding to the mortality table used by the insurance company. 
\begin{remark}
In this paper, we work under the assumption that the mortality intensity is deterministic, following, e.g. the Gompertz-Makeham law. The extension to stochastic intensity is an interesting addition to the framework, but it does not bring any additional financial insight, in particular in the case where the mortality and the financial market are independent. Therefore, we decided to opt for a simplified setting.       
\end{remark}
Next, we define the combined financial-insurance market, which is the natural framework for unit-linked life insurance contracts. We consider the  probability space $(\Omega, \mathcal{G}, \P)$ with the filtration $\mathbb{G}$ as our mathematical environment, where $\Omega=\Omega^M\times\Omega^I$, $\mathcal{G}=\mathcal{F}\vee \mathcal{H}$, $\P=\P^M\times  \P^I$ and  $\mathbb{G}=\left\lbrace\mathcal{G}_t\right\rbrace_{t\in[0,T]}$ given by $\mathcal G_t=\mathcal F_t \vee \mathcal H_t$, for every $t\in[0,T]$.
Note that, according to this construction, the insurance market is independent of the a priori given financial market. For the case where the financial and the insurance market are correlated, we refer to \cite{ceci2017unit}, where also partial information on the mortality intensity is assumed.\\ 
We assume that the insurance company issues a {green} unit linked life insurance policy, for a an individual with mortality intensity $\gamma$. Under this type of life insurance contract the value of the benefit depends on the investment fund, optimized in Section \ref{sec:optimal_inv}, and the time of payment depends on the death time of the individual. 

In this paper, we compute hedging strategies for two types of unit linked contracts, namely the pure endowment and the term insurance. Benefits of these contracts are illustrated below in Sections \ref{sec:pure} and \ref{sec:term}. In particular, we use similar arguments as in \cite{ceci2015hedging}, adapted to a market with an additional non-traded carbon intensity factor.\\

We consider the combined financial-insurance market which consists of $(d+1)$-tradeable assets and a unit linked contract. This market is incomplete for two reasons. First, in view of equation \eqref{eq_optimal_portf}, the variations of carbon emission levels $\mathbf C$ which appear in an integrated form, directly affect the payoff structure. This risk cannot be hedged by trading in $\mathbf S$. Second, mortality risk cannot be covered through market investments. 
\begin{remark}
Partial hedging of mortality risk can be obtained by investing in, e.g., longevity bonds. However, we opt against introducing these investment opportunities for the following reasons. Longevity bonds are not widely traded, and the market for them is extremely illiquid. This lack of depth creates pricing uncertainty and high transaction costs, which undermines their usefulness in dynamic hedging strategies. In practice, it becomes difficult to adjust or unwind positions, making them unreliable for continuous risk management in a stochastic setting.
Moreover, longevity bonds are typically linked to general population mortality indices, such as national life tables. However, insurance companies often underwrite lives from more specific, healthier subsets (e.g., individuals with medical underwriting or higher socio-economic status). This mismatch creates basis risk, meaning that changes in the bond's value may not accurately reflect the longevity risk exposure of the insurance company’s actual portfolio.
\end{remark}

For an illustration, in this section we assume that the carbon intensity process follows the dynamics in equation \eqref{eq:CIR_PROCESS}. However the same computations could be applied to any of the three possible model for the carbon intensity discussed in examples in Section \ref{sec:examples}, using their own semimartingale decomposition. \\

In the sequel, since $\mathbf{\Sigma}$ is invertible we define $\bm{\theta}:= \mathbf{\Sigma}^{-1}\left(\bm{\mu}- r\1 \right)$, and the probability measure $\Q^M$, equivalent to $\P^M$, with the density 
\begin{equation}\label{eq:mmm^M}
\frac{\de \Q^M}{\de \P^M}\bigg|_{\mathcal F_t}= \exp\left(-\bm{\theta}^\top \bm{Z}_t - \frac{1}{2} \|\bm{\theta}\|^2t\right)=:L_t, \quad t \in [0,T].
\end{equation}
We let $\whbfZ=\{\whbfZ_t\}_{t\in[0,T]}$, given by $\whbfZ_t=\bfZ_t+\bm{\theta} t, \quad t\in[0,T]$. 
Then $\whbfZ$ is a  $(\mathbb F, \Q^M)$-Brownian motion and the vector of discounted risky asset price processes $\widetilde{\mathbf S}=\{\widetilde{\mathbf S}_t\}_{t \in [0,T]}$, with $\widetilde{\mathbf S}_t:=e^{-rt}\mathbf S_t$, is a $(\mathbb F,\Q^M)$-martingale. Hence, $\Q^M$ is an equivalent martingale measure for the discounted stock-price process
$\widetilde{\mathbf S}$ in the filtration $\mathbb F$. It is unique when the stock-price market is
considered with respect to the filtration generated by the traded financial risk only. In the 
filtration $\mathbb F$, which also contains the non-traded carbon-intensity information, equivalent martingale measures are in general not unique.\\
In view of the independence between the financial market, the carbon intensity and the mortality we define the probability measure ${\mathbb Q}^*$ on \((\Omega,\mathcal G)\) as
\begin{equation}\label{eq:mmm}
\frac{\de{\mathbb Q}^*}{\de \P}\bigg|_{\mathcal G_t}= L_t,\quad t\in[0,T].
\end{equation}
It is easy to verify that $L$ is a square-integrable $(\mathbb G, \P)$-martingale; moreover $\Q^*$ is a martingale measure for the combined financial-actuarial market and $\whbfZ$ is a $(\mathbb G, {\mathbb Q}^*)$-Brownian motion. 
In particular, since $\Q^*$ changes only the drift of the traded Brownian risk and leaves the non-traded carbon and mortality risks unchanged, it corresponds to the so-called {\em minimal martingale measure} (see, e.g., \cite{follmer2010minimal}), i.e. that equivalent martingale measure such that any $(\mathbb G, \P)$ martingale orthogonal to the martingale part of $\bm{S}$ is a $(\mathbb G, \Q^*)$-martingale.

\begin{remark}
Because of market incompleteness, the martingale condition for the discounted stock-price process only fixes the change of measure
associated with the traded Brownian risk. This component is given by $\bm{\theta}= \mathbf{\Sigma}^{-1}\left(\bm{\mu}- r\1 \right)$. The remaining degrees of freedom concern the non-traded sources of risk. Under the CIR
specification \eqref{eq:CIR_PROCESS} adopted in this section, where in particular $\bfW=\left(W_{1,t},\dots,W_{d,t}\right)^\top$ denotes the Brownian motion driving
the carbon-intensity process, a generic equivalent martingale measure $\Q$ may be represented by
a density of the form
\begin{equation}
L^{\mathbb{Q}}_t=\frac{\de\mathbb Q}{\de \P}\bigg|_{\mathcal G_t}=\mathcal{E}\left(-\int_0^\cdot\bm{\theta}_s^\top\de\bfZ_s+\int_0^\cdot\bm{\psi}_s^{\mathbb{Q} \top}\de\bfW_s+\int_0^\cdot \varphi^{\Q}_s(\de H_s-(1-H_s)\gamma(s)\de s)\right)_t,\quad t\in [0,T],
\end{equation}
where the $\mathbb G$-predictable processes $\bm{\psi}^{\mathbb Q}=\{\bm{\psi}_t^{\mathbb Q}\}_{t \in [0,T]}$ and $\varphi^{\Q}=\{\varphi_t^{\Q}\}_{t\in [0,T]}$, with $\varphi_t^\Q>-1$, for each $t \in [0,T]$, represent possible changes of measure for carbon-intensity risk and mortality risk,
respectively, and are assumed to satisfy the usual integrability conditions ensuring that \(L^\Q\)
is a true $(\mathbb G, \mathbb P)$-martingale. Under $\Q$, the mortality intensity is given by
$\gamma^{\Q}=(1+\varphi^{\Q})\gamma$. Here, $\mathcal{E}(Y)$ denotes the Doléans-Dade exponential of a $(\mathbb{G}, \mathbb{P})$-semimartingale $Y$.
In view of the fact that $\P^I$ is the probability measure used by the insurance company to estimate the mortality of the insured individual (i.e., it corresponds to a specific mortality table used by the insurance company), we may also assume that it already incorporates the mortality premium, that is $\varphi^{\Q}=0$. Note that,  because of independence between the measure $\P^I$ and the financial market, this assumption does not affect the computation of the hedging strategy from a technical perspective. In particular, what changes is that $\gamma^{\Q}$ would replace $\gamma$. Moreover, consistently with the measure ${\mathbb Q}^*$ defined above, we leave the carbon-intensity risk unchanged by taking \(\bm{\psi}^\mathbb{Q}\equiv0\).
\end{remark}
The theoretical derivation in Sections \ref{sec:pure} and \ref{sec:term} is presented for a single policyholder, consistently
with the construction above. The extension to portfolios of independent policies follows by aggregation
of the individual contracts and is used in the numerical analysis of Section \ref{sec:hedging}. Before deriving the explicit formulas for the insurance contracts, we recall the quadratic hedging criterion that we apply to this setting.

\subsection{The quadratic hedging criterion}
Suppose that the insurance company issues a unit linked life insurance contract with the maturity $T$ and the payoff given by a $\mathcal{G}_T$ measurable random variable $G_T$. The insurer seeks to hedge such a contract using available financial instruments. Because of market incompleteness,  we adopt the {\em local risk-minimization} approach to determine a hedging strategy,  not necessarily self-financing, that replicates the payoff at maturity and minimizes the strategy costs in a suitable way  (see \cite{schweizer2001} for further details). The first step is to introduce the class of all admissible hedging strategies.
\begin{defn}
The space $\Theta(\mathbb G)$ consists of all $\mathbb R^d$-valued $\mathbb G$-predictable processes ${\bm{\eta}}=\{{\bm{\eta}}_t\}_{t \in [0,T]}$ satisfying the following integrability condition:
\begin{equation}
\mathbb E^{\P}\left[\int_0^T\left(\|{\bm \eta}_u^\top {\rm diag}(\widetilde{\mathbf{S}}_u)\mathbf{\Sigma}\|^2 +|{\bm \eta}_u^\top {\rm diag}(\widetilde{\mathbf{S}}_u)(\bm{\mu}-r\1)|\right)\de u\right]<\infty.
\end{equation}
\end{defn}
\begin{defn}
An admissible strategy is a pair $(\bm \eta, \zeta)$, where $\bm \eta \in \Theta(\mathbb G)$ and ${\zeta}=\{{ \zeta}_t\}_{t \in [0,T]}$ is a $\mathbb R$-valued $\mathbb G$-adapted process such that the associated discounted value process $V(\bm \eta, \zeta)=\bm \eta^\top \widetilde{\mathbf{S}} + \zeta$ is right-continuous and square-integrable, i.e., $ V_t(\bm \eta, \zeta)\in L^2(\mathcal G_t,\mathbb P)$ for every $t \in [0,T]$.
\end{defn}
Here, the processes $\bm \eta$ and $ \zeta$ represent, respectively, the units of all risky and riskless assets held in the portfolio. For any admissible hedging strategy $(\bm \eta, \zeta)$, we can define the associated {\em cost process} $C(\bm \eta, \zeta)=\{C_t(\bm \eta, \zeta)\}_{t \in [0,T]}$, which is the $\mathbb R$-valued $\mathbb G$-adapted process given by
\begin{equation}
C_t(\bm \eta, \zeta)=V_t(\bm \eta, \zeta) - \int_0^t \bm \eta_u^\top \de \widetilde{\mathbf S}_u, \quad t \in [0,T].
\end{equation}
Although admissible strategies that replicate the payoff at maturity, i.e., $V_T(\bm \eta, \zeta)=G^{PE}$, will in general not be self-financing, it turns out that good admissible strategies are self-financing \emph{on average} in the  following sense.
\begin{defn}
An admissible strategy $(\bm \eta, \zeta)$ is called {\em mean-self-financing} if the associated cost process $C(\bm \eta, \zeta)$ is a $(\mathbb G,\mathbb P)$-martingale.
\end{defn}
Following the idea of \cite{schweizer2001}, we now introduce the concept of pseudo-optimal strategy in this framework.
\begin{defn}
Let $G_T\in L^2(\mathcal{G}_T, \mathbb P)$ be a random (discounted) payoff. An admissible strategy $(\bm \eta, \zeta)$, such that $V_T(\bm \eta, \zeta) = G_T$ $\mathbb P$-a.s., is called {\em pseudo-optimal} for $G_T$ if and only if $(\bm \eta, \zeta)$ is mean-self-financing and the $(\mathbb G, \mathbb P)$-martingale $C(\bm \eta, \zeta)$ is strongly orthogonal to the $\mathbb{P}$-martingale part of $\widetilde{\mathbf{S}}$.
\end{defn}
Note that, in our setup, to be strongly orthogonal to the martingale part of \( \widetilde{\mathbf{S}}\) is equivalent to be strongly orthogonal to $\bfZ$.
The key result for finding pseudo-optimal strategies is the F\"ollmer–Schweizer decomposition.
\begin{defn}
A random (discounted) payoff \(G_T\in L^2(\mathcal{G}_T, \mathbb P)\) admits the F\"ollmer–Schweizer decomposition with respect to \( \widetilde{\mathbf{S}}\) if there exist a random variable $G_0\in L^2(\mathcal{G}_0, \mathbb P)$, a process $\bm \eta^G \in\Theta(\mathbb G)$, and a square-integrable $(\mathbb G,\mathbb P)$-martingale $O=\{O_t\}_{t \in [0,T]}$ with $O_0=0$ strongly orthogonal to $\bf Z$, such that
\begin{equation}\label{eq:FS_decomp}
G_T =G_0 + \int_0^T\bm \eta_t^{G,\top} \de \widetilde{\mathbf{S}}_t + O_T, \quad \mathbb P\mbox{-}{\rm a.s.}
\end{equation}
\end{defn}
\begin{prop}\label{prop:FS}
A random (discounted) payoff \(G_T\in L^2(\mathcal{G}_T,\mathbb P)\) admits a unique pseudo-optimal strategy $(\bm{\eta}^{\star}, \zeta^{\star})$, with $V_T(\bm{\eta}^{\star}, \zeta^{\star})=G_T$ $\mathbb P$-a.s., if and only if $G_T$ admits the F\"ollmer-Schweizer decomposition \eqref{eq:FS_decomp}.
The strategy $(\bm{\eta}^{\star}, \zeta^{\star})$ is explicitly given by
\begin{equation}
\bm \eta_t^{\star}=\bm{\eta}_t^{G}, \quad t \in [0,T],
\end{equation}
with minimal cost
\begin{equation}
C_t(\bm{\eta}^{\star}, \zeta^{\star})=G_0 + O_t, \quad t \in [0,T],
\end{equation}
and the corresponding discounted value process is 
\begin{equation}
V_t(\bm{\eta}^{\star}, \zeta^{\star})= G_0 + \int_0^t  \bm \eta_u^{G,\top} \de \widetilde{\mathbf S}_u + O_t, \quad t \in [0,T],
\end{equation}
so that $\zeta_t^{\star}=V_t(\bm{\eta}^{\star}, \zeta^{\star})-\bm{\eta}_t^{\star,\top}\widetilde{\mathbf S}_t$, for every $t \in [0,T]$.
\end{prop}

\subsection{Hedging of a pure endowment contract}\label{sec:pure}
We illustrate in this section the methodology for the computation of the risk minimizing strategy for a unit linked pure endowment contract, and we briefly discuss the extension to the unit linked term insurance in the next section, postponing all the computations to the Appendix.

\begin{remark}[Single constract vs portfolio of policies]
We remark that it is not difficult to generalize the following methodology to the case of a portfolio of independent unit linked life insurance contracts (both of  pure endowment and term insurance types). Such generalization will follow using only the property of the linearity of the expectation and therefore in the sequel we provide theoretical results for the case of a single contract, and we discuss the effect of diversification for a portfolio of contracts in the numerical part. 
\end{remark}

We let $T$ be the maturity of the unit linked pure endowment contract with discounted payoff 
\begin{equation} \label{eq:payoff_PE_portf}
G^{PE}=  e^{-rT}\phi(X^{\bmpi^{\star}}_T)\mathds{1}_{\{\tau > T\}}, 
\end{equation}
where $X^{\pi^\star}$ is given by
\begin{equation}\label{eq_optimal_portf}
\dfrac{\de X_t^{\bmpi^{\star}}}{ X_t^{\bmpi^{\star}}}=\left[r+\bm{\pi}_t^{\star,\top}\left(\bm{\mu}-r\1\right)\right]\de t+\bm{\pi}_t^{\star,\top}\bfSigma\de\bfZ_t,\quad X_0^{\bmpi^{\star}}=x_0,
\end{equation}
with $\bmpi^\star_t=\bmpi^\star(t,\mathbf{C}_t)$ as in equation \eqref{eq:sol_CRRA}. Here, $X^{\bmpi^{\star}}$ is
the value of the optimized fund (see  Theorem \ref{thm:general_thm}). The function  $\phi:\mathbb{R}^+ \to \mathbb{R}^+ $ defines the benefit structure, that is, $\phi(X_T^{\bmpi^{\star}})$
provides the amount to be paid at maturity $T$, if the policyholder is still alive at time $T$. We assume that $\phi(X_T^{\bmpi^{\star}}) \in L^2(\mathcal F_T,\P^M)$, and hence $G^{PE} \in L^2(\mathcal G_T, \P)$.
A typical example of benefit is $\phi(x)=\min(\max(x, k), \bar{k})$, which includes a minimum guarantee $k>0$ and a maximum amount $\bar{k}>0$ for the policyholder.   

Now, we compute the pseudo-optimal strategy $(\bm \eta^\star,\zeta^\star)$ for the pure endowment contract, with  payoff $G^{PE}$, using its 
its F\"ollmer-Schweizer decomposition and  Proposition \ref{prop:FS}. We let $\Theta(\mathbb{F})$ be the space of all $\R^d$-valued $\mathbb F$-predictable processes $\bm \beta=\{\bm \beta_t\}_{t \in [0,T]}$ satisfying the following integrability condition
\begin{equation}
\mathbb E^{\P^M}\left[\int_0^T\left(\|{\bm \beta}_u^\top {\rm diag}(\widetilde{\mathbf{S}}_u)\mathbf{\Sigma}\|^2 +|{\bm \beta}_u^\top {\rm diag}(\widetilde{\mathbf{S}}_u)(\bm{\mu}-r\1)|\right)\de u\right]<\infty.
\end{equation}
We assume that the discounted payoff ${ e^{-rT}}\phi(X_T^{\bmpi^{\star}})$ admits the Föllmer–Schweizer decomposition with respect to $\widetilde{\mathbf S}$ and $\mathbb{F}$, i.e.,
\begin{equation}\label{eq:FS_phi_i}
{ e^{-rT}}\phi(X_T^{\bmpi^{\star}}) = U_0 + \int_0^T \bm{\beta}_t^{\top} \, \de \widetilde \bfS_t + A_T \quad \mathbb{P}^M\text{-a.s.}, 
\end{equation}
where $U_0 \in L^2(\mathcal{F}_0, \mathbb{P}^M)$, $\bm{\beta} \in \Theta(\mathbb{F})$, and $A = \{A_t\}_{t \in [0,T]}$ is a square-integrable $(\mathbb{F}, \mathbb{P}^M)$-martingale with $A_0 = 0$, strongly orthogonal to the $\mathbb{P}^M$-martingale part of $\widetilde{\mathbf{S}}$. Since $\widetilde{\mathbf{S}}$ is a continuous process and satisfies the {\em structure condition} (see, e.g. \cite{ansel2006unicite}), then, the F\"ollmer–Schweizer decomposition of $G^{PE}$ with respect to $\widetilde{\mathbf S}$ and $\mathbb{F}$ coincides with its Galtchouk-Kunita-Watanabe decomposition under the measure $\mathbb{Q}^M$ (see, e.g., \cite[page 552-553]{schweizer2001}), which identifies the minimal martingale measure in the financial market. Using the fact that $\Q^*=\Q^M\times \P^I$, we can take the conditional expectation  of $G^{PE}$ with respect to ${\mathcal{G}}_t$ under the minimal martingale measure (for the combined market) $\mathbb{Q}^*$, and we get
\begin{align}
\E^{\Q^*}\left[G^{PE} \mid  {\mathcal{G}}_t \right] 
&= \E^{\Q^*}\left[ e^{-rT}\phi(X_T^{\bmpi^{\star}}) \ind[\tau>T] \Big{|} {\mathcal{G}}_t \right]\\
\label{eq:value}&=\E^{\Q^M}\left[  e^{-rT}\phi(X_T^{\bmpi^{\star}})\mid {\mathcal{F}}_t \right] \ind[\tau>t] e^{-\int_t^T\gamma(s) \de s} ,\quad t\in[0,T], 
\end{align}
where in equation \eqref{eq:value} we have used the independence between the financial market and the insured individuals. We now define the processes $B=\{B_t\}_{t\in [0,T]}$ and $U=\{U_t\}_{t\in [0,T]}$ as follows:
\begin{align}
B_t &:= \ind[\tau>t] e^{-\int_t^T\gamma(s) \de s}, \label{eq:B} \\
U_t &:= \E^{\Q^M}[ { e^{-rT}}\phi(X_T^{\bmpi^{\star}}) \mid \mathcal{F}_t] =\E^{\Q^M} \left[ U_0 + \int_0^T \bm{\beta}_u^{\top} \, \de\widetilde \bfS_u + A_T \, \Big| \, \mathcal{F}_t \right] \\&= U_0 + \int_0^t \bm{\beta}_u^{\top} \, \de\widetilde \bfS_u + A_t, \label{eq:U} 
\end{align}
for each $t \in [0,T]$, where for \eqref{eq:U} we relied on decomposition \eqref{eq:FS_phi_i} and the martingale-preserving property of the minimal martingale measure. Using integration by parts, and the fact that
\(\{H_t-\int_0^t(1-H_s)\gamma(s)\,\de s\}_{t \in [0,T]}
\)
is a \((\mathbb H,\P^I)\)-martingale, the process $B$ in equation \eqref{eq:B} admits the representation
\begin{equation}\label{eq:B_integral}
\de B_t
=
-\exp\left(-\int_t^T\gamma(u)\,\de u\right)
\bigl(\de H_t-(1-H_t)\gamma(t)\,\de t\bigr),
\qquad
B_0=\exp\left(-\int_0^T\gamma(u)\,\de u\right).
\end{equation}
\begin{prop}\label{prop:FS_decomp_PE}
The discounted pure endowment contract $G^{PE}$ admits the F\"ollmer-Schweizer decomposition given by
\begin{equation}
G^{PE} = G_0^{PE} + \int_0^T  B_{s-} \bm{\beta}_s^{\top}\, \de\widetilde \bfS_s + O_T^{PE} \quad  \mathbb{P}\text{-a.s.}, 
\end{equation}
where
\begin{equation}
G_0^{PE} = \E^{\Q^M} \left[ { e^{-rT}}\phi(X_T^{\bmpi^{\star}}) \mid \mathcal{F}_0\right]  \ \E^{\P^I}\left[\ind[\tau>T]\right]=  U_0   e^{-\int_0^T \gamma_s \de s},
\end{equation}
and
\begin{equation}
O_t^{PE}
=
\int_0^t B_{s-}\,\de A_s
-
\int_0^t U_{s-}e^{-\int_s^T\gamma(u)\,\de u}
\bigl(\de H_s-(1-H_s)\gamma(s)\,\de s\bigr),
\qquad t\in[0,T].
\end{equation}
where $B$ is given by \eqref{eq:B_integral},  $ \bm{\beta}$ is the integrand with respect to $\widetilde{\bfS}$ in the F\"ollmer-Schweizer decomposition of $\phi(X_T^{\bmpi^{\star}})$, see \eqref{eq:FS_phi_i}.
Then, the pseudo-optimal strategy $(\bm{\eta}^{\star}, \zeta^{\star})$ is given by
\begin{equation}
\bm{\eta}^{\star}_t = B_{t-} \bm{\beta}_t, \quad \zeta_t^{\star} = V_t(\bm{\eta}^{\star}, \zeta^{\star}) -  B_t \bm{\beta}^{\top}_t \widetilde \bfS_t,\quad t\in[0,T],
\end{equation}
and the optimal value process $V(\bm{\eta}^{\star}, \zeta^{\star})=\{V_t(\bm{\eta}^{\star}, \zeta^{\star})\}_{t \in [0,T]}$ is given by
\begin{equation}
V_t(\bm{\eta}^{\star}, \zeta^{\star}) = G_0^{PE} + \int_0^t B_{r-} \bm{\beta}_r^{\top} \de \widetilde \bfS_r + O_t^{PE},\quad t\in[0,T].
\end{equation}
\end{prop}
\begin{proof}
See Section \ref{sect:proof_FS_decomp_PE} in Appendix.
\end{proof}
In the next result, we provide a characterization of the process $\bm{\beta}$. \begin{prop}\label{Prop:characterization_beta_B_PE}
Let $({\bm{\beta}},{\zeta})$ be the pseudo-optimal strategy for the payoff $\phi(X_T^{\bmpi^{\star}})$ with respect to $\mathbb{F}$. Then, the optimal discounted value process $\{{V}_t({\bm{\beta}},{\zeta})\}_{t \in [0,T]}$ is given by
\begin{equation}
V_t({\bm{\beta}},{\zeta}) = \E^{\Q^M}[e^{-rT} \phi(X_T^{\bmpi^{\star}}) \mid \mathcal{F}_t], \quad t \in [0,T],
\end{equation}
and we have that $\bm{\beta}_t =  X_t^{{\bm\pi}^\star}\,\frac{\partial F}{\partial x}(t, X_t^{\bmpi^{\star}}, \bfC_t) \bmpi^{\star, \top}_t \mathrm{diag}(\widetilde {\bfS}_t)^{-1} $, for every $t\in[0,T]$, where the function $F(t, x, c)$, is the unique solution of the backward equation  
\begin{align}\label{eq:pdeF}
{\frac{\partial F}{\partial t}(t,x,\bfc)}+
\frac{\partial F}{\partial x}(t,x,\bfc) x r + \frac{1}{2}\frac{\partial^2 F }{\partial x^2}(t,x,\bfc) x^2 \bmpi^{\star,\top} \bfSigma \bfSigma^\top \bmpi^\star + \mathcal{L}^\bfC F(t,x,\bfc)=0,
\end{align}
for all $(t,x,\bfc)\in[0,T)\times\mathbb{R}_+\times\mathcal{D},$ with the final condition 
$F(T, x, \bfc)=e^{-rT} \phi(x)$, for all $(x,\bfc)\in\mathbb{R}_+\times\mathcal{D}$.
\end{prop}
\begin{proof}
See Section \ref{Prop:proof_characterization_beta_B_PE} in Appendix.
\end{proof}

\subsection{Hedging of a Term Insurance contract}\label{sec:term}

Our second example is the term insurance contract, which is a mortality benefit that pays the agreed amount to the beneficiary at the time of insured death, provided that it occurs before maturity $T$. Its discounted payoff is given by
\begin{equation}\label{eq:TI_Port_Payoff}
G^{TI}= e^{-r \tau}\psi(\tau, X^{\bmpi^{\star}}_{\tau})\mathbf 1_{\{\tau \le T\}} = \int_0^T{ e^{-rs}}\psi(s, X^{{\bmpi}^\star}_s) \de H_s,
\end{equation}
for some function 
$\psi:[0, T]\times \mathbb{R}^+ \to \mathbb{R}^+ $ such that 
$\psi(t, X_t^{\bmpi^{\star}}) \in L^2(\mathcal F_t,\P^M)$, for each $t \in [0,T]$. Hence, the contract has  payoff $\psi(\tau, X^{\bmpi^{\star}}_{\tau})\mathbf 1_{\{\tau \le T\}} \in L^2(\mathcal G_{\tau}, \P)$. 

\begin{prop}\label{prop:TI}
Let for every $0\le t\le u\le T$ 
\begin{align}
B(u,t) &:= (1-H_t) \gamma(u) e^{-\int_t^u \gamma(s) \de s}. \label{eq:Bu_expl}
\end{align}
Assume that $e^{-ru} \psi(u, X^{\bmpi^{\star}}_{u})$  admits the F\"ollmer–Schweizer decomposition and let $\bm{\beta}(u)$ be the integrand in such decomposition (see equation $(D.3)$ in Section \ref{app:C_TI} in Appendix. The pseudo-optimal strategy $(\bm{\eta}^{\star}, \zeta^{\star})$ is given by:
\begin{align}
\bm{\eta}^{\star}_t &= \int_t^T  B_{t-}(u)\bm{\beta}_{t}(u)\,\de u, \quad  t \in [0, T], \\
\zeta^{\star}_t &= V_t(\bm{\eta}^{\star}, \zeta^{\star}) - \left( \int_t^T B_{t-}(u)\bm{\beta}_t(u)\,du \right)^\top \widetilde{\bfS}_t, \quad t \in [0, T],
\end{align}
and the optimal (discounted) value process $V(\bm{\eta}^{\star}, \zeta^{\star})$ satisfies:
\begin{align}
V_t(\bm{\eta}^{\star}, \zeta^{\star}) = G^{TI}_0 + \int_0^t \bm{\eta}^{\star,\top}_r\,\de \widetilde{\bfS}_r + O_t^{TI},\quad t\in[0,T], 
\end{align}
where $G^{TI}_0=\mathbb{E}^{\mathbb{Q}^*}[G^{TI}]$ and $\{O^{TI}_t\}_{t \in [0,T]}$ is a $(\mathbb{G}, \mathbb{Q}^*)$ martingale whose explicit expression is given in equation \eqref{eq:K1} of Section \ref{app:C_TI} in Appendix.
\end{prop}
The proof of this proposition is provided in Section \ref{app:C_TI} in Appendix.

\section{Numerical analysis of the hedging results}\label{sec:hedging}
In this section, we present an illustrative simulation study aimed at assessing the performance of the risk minimizing hedging strategy for green unit linked insurance contracts. In particular, we consider the three types of policies introduced in Section \ref{sect:locally_risk_minimizing} and begin by specifying the payoff functions that define their benefits. For the pure endowment contract, we assume that, conditional on survival up to time $T$, the policyholder receives at maturity a benefit equal to $\phi(X_T^{\bmpi^\star})$, where $\phi$ is defined by
\begin{equation}\label{eq:PE_PAYOFF}
\phi(x):=\min(\bar{k},\,\max(k,x)),
\end{equation}
with $k\geq 0$ denoting the guaranteed minimum amount and $\bar{k}\geq 0$ the maximum benefit. For the term insurance contract, we assume that, if death occurs before maturity at time $\tau\in[0,T]$, the insurance company pays at time $\tau$ the benefit $\psi(\tau,X_\tau^{\bmpi^\star})$, where the function $\psi$ is given by
\begin{equation}\label{eq:TI_PAYOFF}
\psi(t,x):=\min\bigl(\bar{k}(t),\max(k(t),x)\bigr), \quad (t,x)\in[0,T]\times\mathbb R_+.
\end{equation}
Here, $k(\cdot):[0,T]\to[0,+\infty)$ and $\bar{k}(\cdot):[0,T]\to[0,+\infty)$ are non-negative functions satisfying $k(t)\less \bar{k}(t)$ for every $t\in[0,T]$, and represent, respectively, the guaranteed minimum amount and the maximum possible payment. Throughout this section, we assume that $k(t)=xe^{rt}$ and $\bar{k}(t)=xe^{10rt}$ for every $t\in[0,T]$, where $x$ denotes the initial value of the underlying investment fund, normalized to one, that is, $x=1$, and the risk free interest rate is set equal to $r=0.05$.\\
The endowment insurance contract combines the previous two benefit structures. More precisely, if death occurs before maturity, at time $\tau\in[0,T)$, we assume that the policy pays at time $\tau$ the benefit $\tilde\psi(\tau,\,X_\tau^{\bmpi^\star})$ in equation \eqref{eq:disc_TI_PAYOFF}, which is a discounted version of the payoff defined in equation \eqref{eq:TI_PAYOFF}, namely
\begin{equation}\label{eq:disc_TI_PAYOFF}
\tilde\psi(t,x):=\varrho e^{-rt}\psi(t,x),\quad \varrho\in(0,1].
\end{equation}
If instead the policyholder survives until maturity, the survival benefit paid at time $T$ is given by $\tilde\phi(X_T^{\bmpi^\star})$ in equation \eqref{eq:disc_PE_PAYOFF}, which is a discounted version of the payoff in equation \eqref{eq:PE_PAYOFF}, namely
\begin{equation}\label{eq:disc_PE_PAYOFF}
\tilde\phi(x):=e^{-rT}\phi(x).
\end{equation}
Moreover, we assume that the sustainable investment fund underlying the policies consists of four stocks, whose carbon intensity processes follow the CIR dynamics introduced in equation \eqref{eq:CIR_PROCESS}. For simplicity, in this simulation exercise we assume a constant long run mean, that is, $\bar C(t)=\bar C\geq 0$ for every $t\in[0,T]$. The parameters of the four stocks and of their corresponding carbon intensity processes are reported in Table \ref{Market_parameters}.\\
We assume that the mortality intensity of the representative policyholder follows the Gompertz-Makeham law, namely
\begin{equation}\label{GM_intensity}
\gamma_t = \xi +\frac 1b \exp\left({\frac{\iota+t+m}{b}}\right),
\end{equation}
where $\iota\g 0$ is the initial age of the policyholder, $\xi\geq 0$ captures the age-independent component of mortality, while $b>0$ and $m\in\mathbb{R}$ determine, respectively, the shape and the level of the age-dependent component of mortality. The mortality parameters are reported in Table \ref{tab:mortality_intensity}. The remaining parameters are set as in Section \ref{sect:dynamic}: $\delta=1$ and $\alpha=0.0025$.
\begin{table}[H]
\centering
\begin{tabular}{cccc} 
\hline
$\xi$        & $b$          & $m$          & $\iota$\\ 
\hline
$0.0041959$  & $11.5818911$ & $79.6921211$ & $60$   \\
\hline
\end{tabular}
\caption{Parameters of the Gompertz-Makeham mortality intensity as in equation \eqref{GM_intensity} for the representative policyholder.}\label{tab:mortality_intensity}
\end{table}
To compute the hedging strategies for the three types of insurance contracts, we rely on the following numerical procedure. First, we simulate the carbon-intensity dynamics under the CIR specification by means of a multidimensional Ninomiya-Victoir weak second-order scheme, similarly to  \cite{alfonsi2024high} and \cite{AL2025SIFIN}). We then compute the optimal portfolio weights $\bmpi^\star$ of the underlying sustainable investment fund and, conditional on their realizations, simulate the fund dynamics; this allows us to avoid simulating all stock paths directly. Moreover, conditional on the realized path of $\bmpi^\star$, the logarithm of the discretized fund is Gaussian, with integrated variance approximated by trapezoidal rule. This conditionally Gaussian structure allows us to apply variance-reduction techniques in the pricing and hedging step, resulting in a substantial computational gain (see Table \ref{Table_VarReduc} in Appendix). Full details of the numerical procedure are reported in Section \ref{sec:numerics}.\\
\begin{figure}[h!]
\centering
\begin{subfigure}[h]{0.49\textwidth}
\centering
\includegraphics[width=\textwidth]{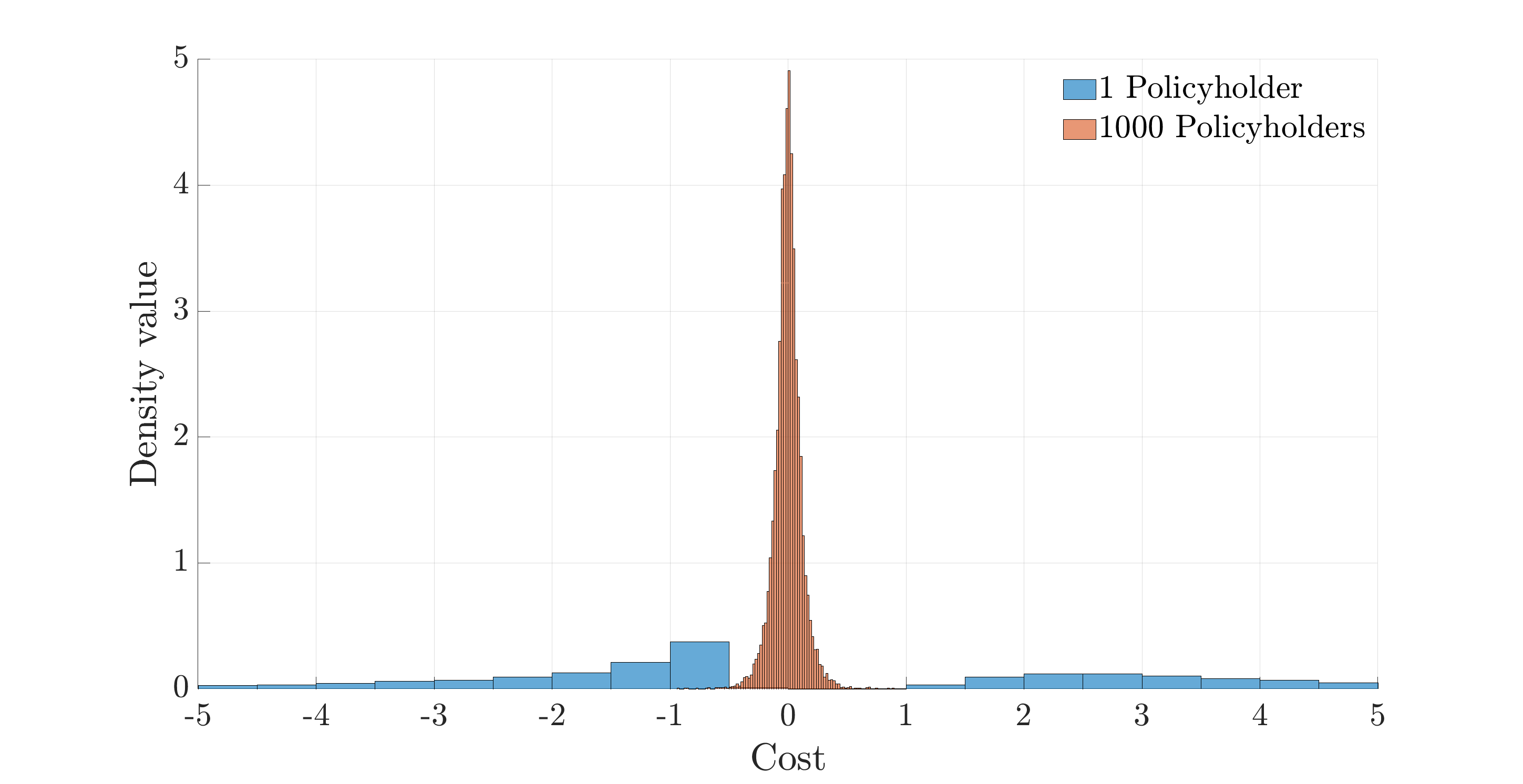}
\caption{Pure Endowment}
\label{fig:PE1vs1000}
\end{subfigure}
\hfill
\begin{subfigure}[h]{0.49\textwidth}
\centering
\includegraphics[width=\textwidth]{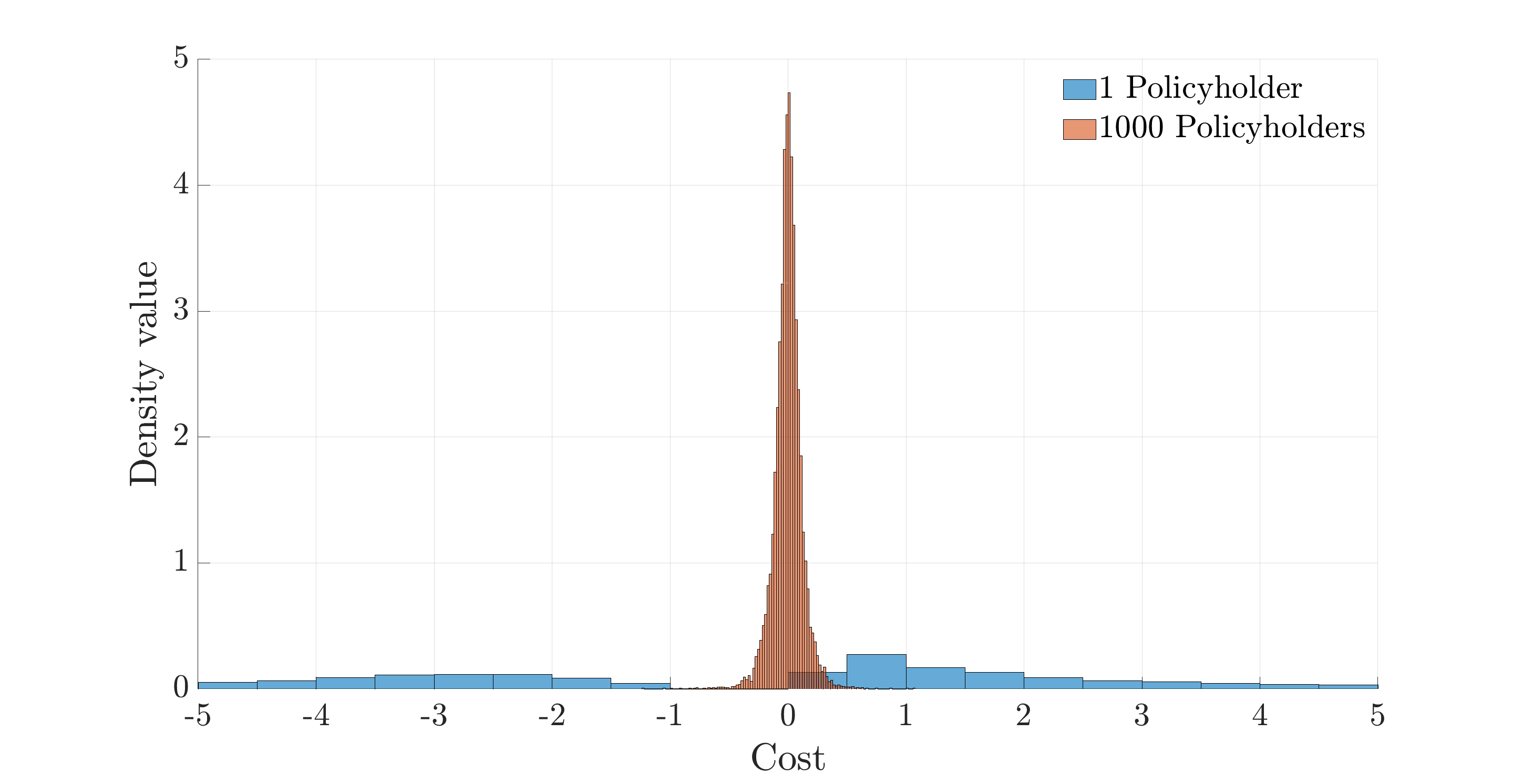}
\caption{Term Insurance} 
\label{fig:TI1vs1000}
\end{subfigure}
\hfill
\begin{subfigure}[h]{0.49\textwidth}
\centering
\includegraphics[width=\textwidth]{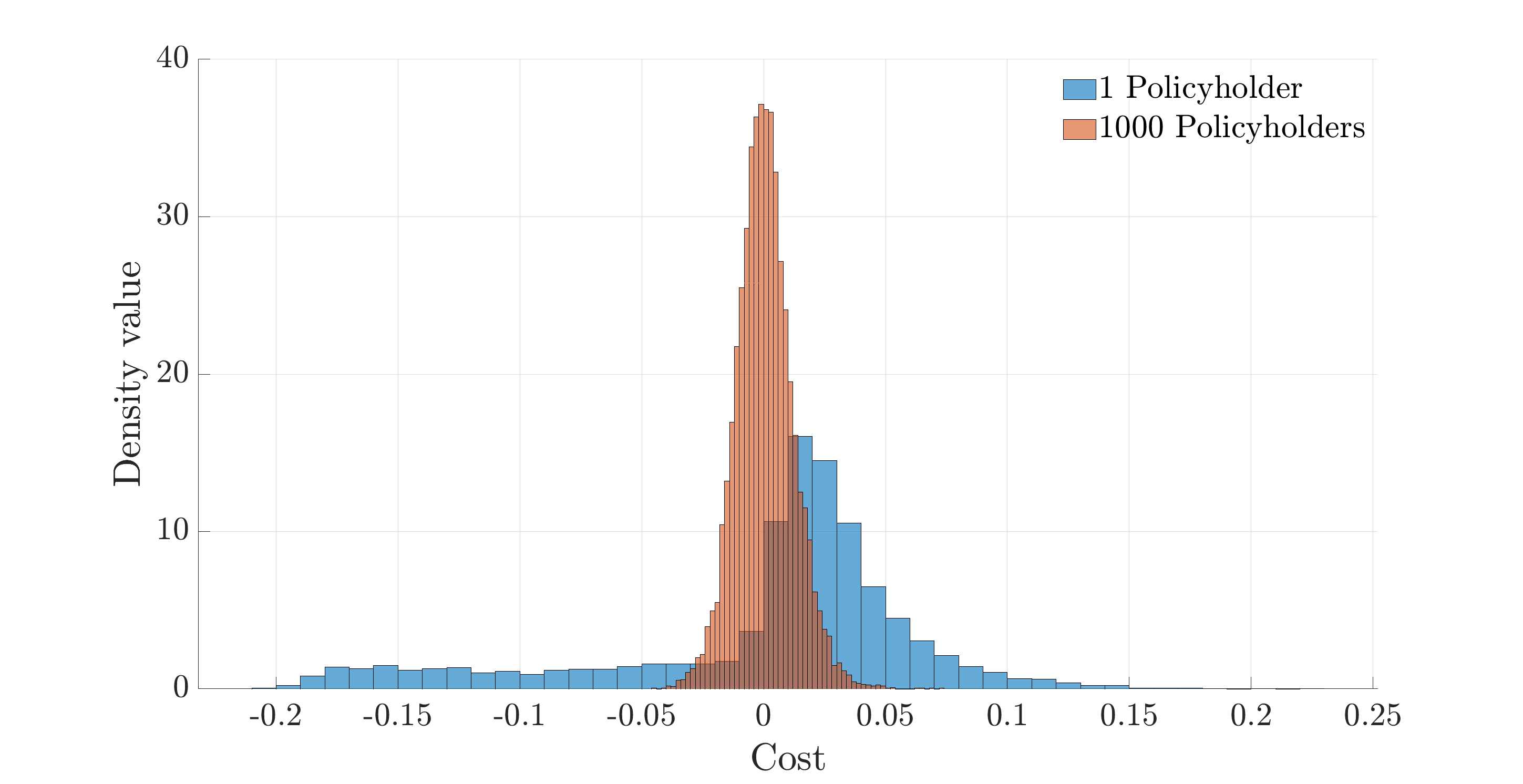}
\caption{Endowment Insurance} 
\label{fig:EI1vs1000}
\end{subfigure}
\caption{Cost of the continuous hedging: 1 Policyholder Vs Portfolio of 1000 Policyholders.}\label{fig:hedge_1vs1000} 
\end{figure}

The insurance payoffs that we aim to hedge are subject to three sources of risk, namely market risk, carbon risk and mortality risk. The primary risk that can be hedged via risk minimization using traditional financial instruments is market risk.
It is clear from the computations of the error-minimizing strategy that both carbon and mortality risk cannot be perfectly hedged, due to the lack of financial instruments that can be used for the scope, hence they enter in the strategy costs. From a mathematical perspective these two sources of risk contribute the orthogonal martingale.    
While unhedgeable in the conventional sense, carbon and mortality risk can be efficiently reduced through structural and strategic design choices. Specifically, carbon risk is mitigated ex-ante through the construction of the ad-hoc investment fund, by selecting assets that exhibit low carbon intensity. Mortality risk, on the other hand, can diversified when considering large portfolios of insurance policies, in view of the law of large numbers.
In that regard, Figure \ref{fig:hedge_1vs1000} illustrates the distribution of the hedging cost of continuously rebalanced hedging strategies for three types of unit linked contracts (pure endowment, term insurance, and endowment insurance) with a maturity $T=20$ years and the initial fund value $x=1$. Each panel corresponds to a different contract type and displays the distribution for a single policy and the distribution for a portfolio of $1000$ policies. In both cases, costs are small on average (for the pure endowment contract, single policyholder, the average loss is $-0.0171$ and for the case of $1000$ policyholders is $-0.0016$0). However, these plots confirm that the hedging is more efficient when considering large portfolios of insurance policies, for which we observe considerably less dispersion. For example, the standard deviation of the pure endowment and single policyholder is $3.8281$, whereas for the case of 1000 contracts is $0.1184$. All the results of the current section are obtained on 10000 simulations.\\

Figure \ref{fig:3_strategies} shows the distribution of the hedging cost for a homogeneous portfolio of $1000$ life insurance contract, with policyholder's initial age of 60 years. We analyse pure endowment on the upper left panel, term insurance on the upper right panel, and endowment insurance on the lower panel. Our goal is to compare three different strategies: the continuous hedging strategy via risk minimization, the static strategy, and no-hedging. Specifically, we refer to the static hedging case when the replicating portfolio is constructed at time $t=0$, and the optimal portfolio weights are not updated until maturity. No hedging corresponds to the case where the insurance company collects the policy premiums at time $t=0$ and deposits the entire amount in the bank account, using it to pay the policyholder benefits at maturity. 
We observe that no hedging is by far the most costly approach, with an average loss of $5.842$ and standard deviation of $2.84$ for the pure endowment case (note that the plot of the density function is truncated, and for example the $90\%$ quantile is $q_{90}=9.362$). The static hedging shows both profits and losses, although on average we observe losses of $2.194$ . This implies that static hedging is an improvement over the unhedged payoffs, but variability is still quite high, equal to $1.807$.
The dynamic risk minimizing strategy on the
other hand significantly reduces both the average loss, which is equal to $-0.0016$, and the variability, with a standard deviation $0.1184$.   
The behaviour of the strategies for the case of term insurance and endowment insurance is similar and the numbers are reported in Table \ref{Tab:averages}. 
\begin{table}[H]
\centering
\begin{tabular}{cccc} 
\Xhline{1.2pt}
             & \multicolumn{3}{c}{Pure Endowment}               \\ 
\hline
             & No Hedging & Static Hedging & Continuous Hedging  \\ 
\hline
Mean & $5.842$    &  $2.186$       & $-0.0016$                   \\
St. Dev.     & $2.84$     &  $1.807$       & $0.1184$                   \\ 
\Xhline{1.2pt}
             & \multicolumn{3}{c}{Term Insurance}               \\ 
\hline
             & No Hedging & Static Hedging & Continuous Hedging  \\ 
\hline
Mean & $1.863$   & $0.503$         & $-0.0004$\\
St. Dev.     & $0.184$   & $0.474$         & $0.121$ \\ 
\Xhline{1.2pt}
             & \multicolumn{3}{c}{Endowment Insurance}          \\ 
\hline
             & No Hedging & Static Hedging & Continuous Hedging  \\ 
\hline
Mean & $7.721$    & $2.708$        & $0.0006$                    \\
St. Dev.     & $2.846$    & $1.727$        & $0.012$ \\
\Xhline{1.2pt}
\end{tabular}
\caption{Mean and standard deviation of the distribution of the hedging cost for a portfolio of 1000 pure endowment policies (top panel), Term Insurance policies (middle panel), and Endowment Insurance policies (bottom panel) under no hedging, static hedging, and continuous hedging.}\label{Tab:averages}
\end{table}

\begin{figure}[h!]
\centering
\begin{subfigure}[h]{0.49\textwidth}
\centering
\includegraphics[width=\textwidth]{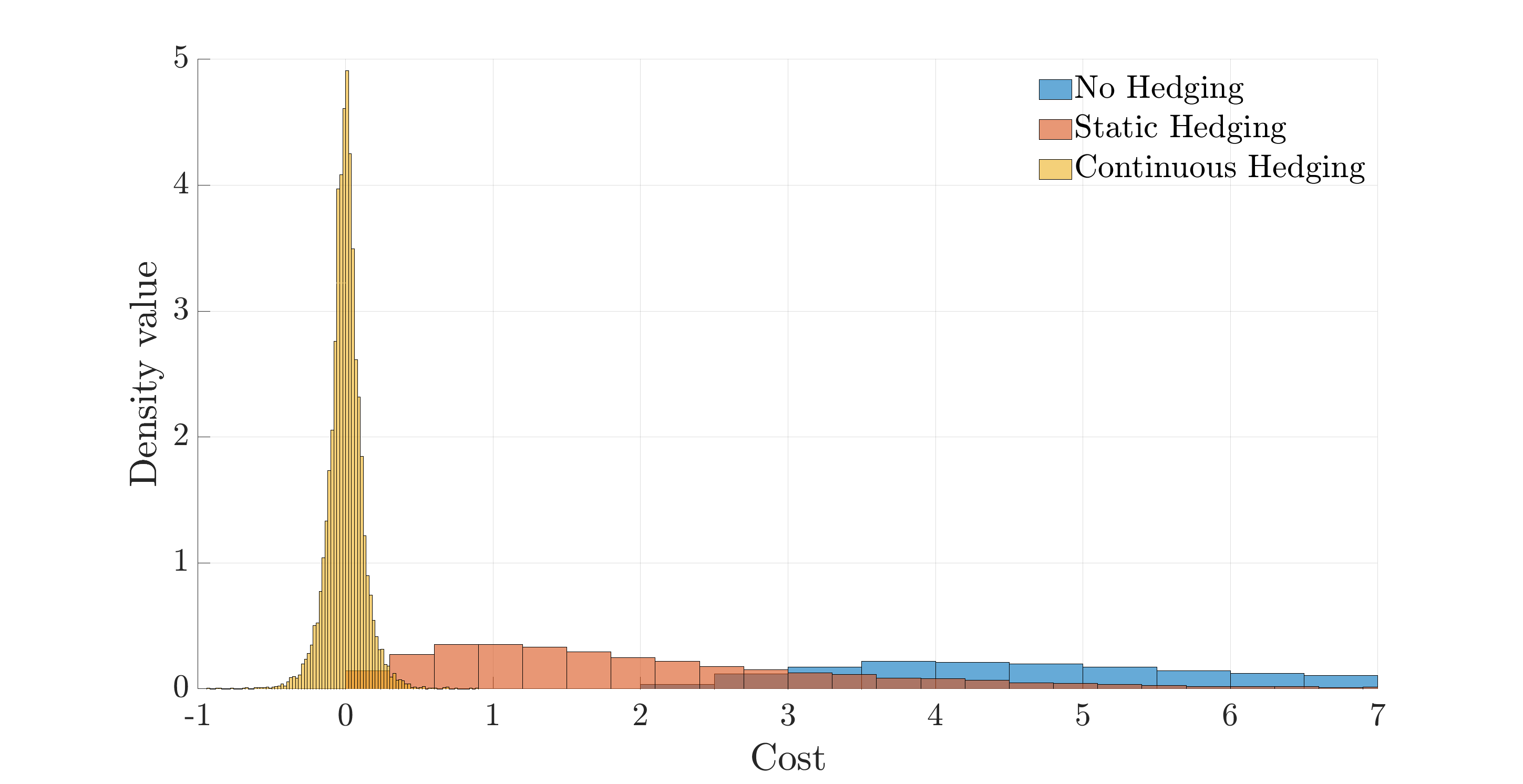}
\caption{Pure Endowment}
\label{fig:PE_3_strategies}
\end{subfigure}
\hfill
\begin{subfigure}[h]{0.49\textwidth}
\centering
\includegraphics[width=\textwidth]{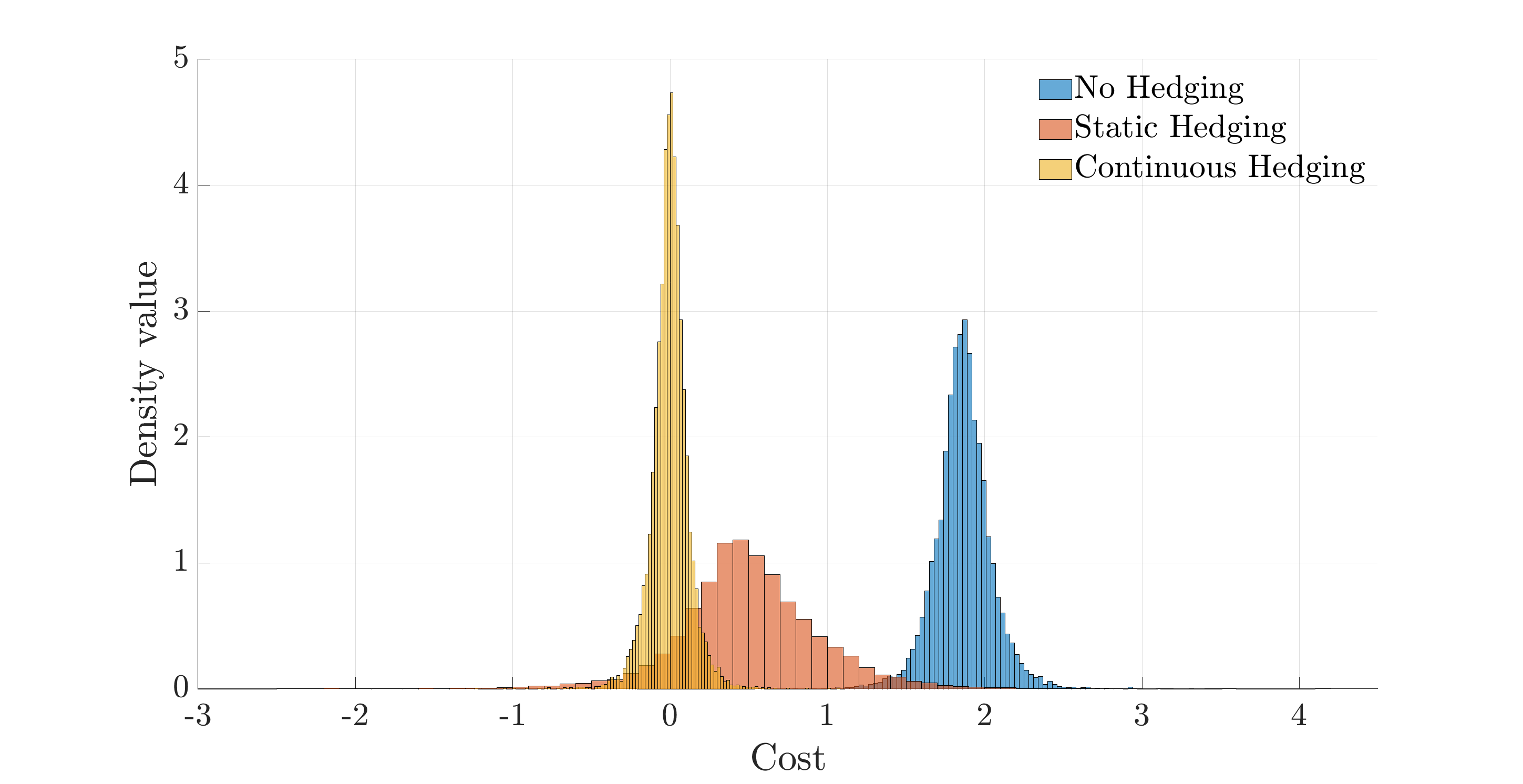}
\caption{Term Insurance} 
\label{fig:TI_3_strategies}
\end{subfigure}
\hfill
\begin{subfigure}[h]{0.49\textwidth}
\centering
\includegraphics[width=\textwidth]{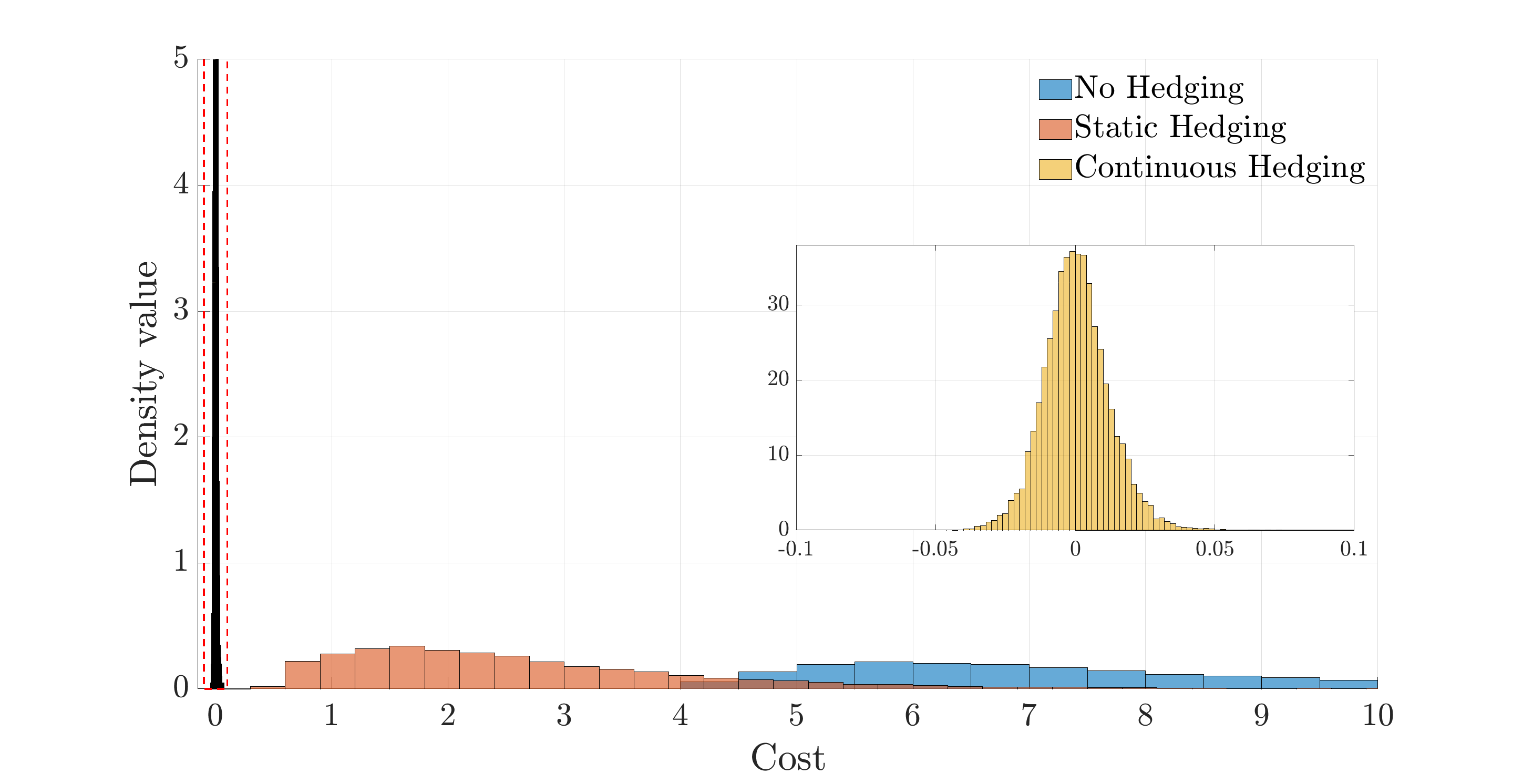}
\caption{Endowment Insurance} 
\label{fig:EI_3_strategies}
\end{subfigure}
\caption{Cost of different hedging strategies}\label{fig:3_strategies}
\end{figure}

We note that the analyses presented in Figures \ref{fig:hedge_1vs1000} and \ref{fig:3_strategies} are performed on policies written on homogeneous policyholders, all aged $\iota=60$ at time $t=0$. In contrast, Figure \ref{fig:HOMOvsHETERO} shows the distribution of the hedging cost for a portfolio of policies written on heterogeneous policyholders, whose initial ages range from $55$ to $65$ years. This picture shows similar outcomes for the two portfolios, revealing that the approach is robust across ages.  

\begin{figure}[h!]
\centering
\begin{subfigure}[h]{0.49\textwidth}
\centering
\includegraphics[width=\textwidth]{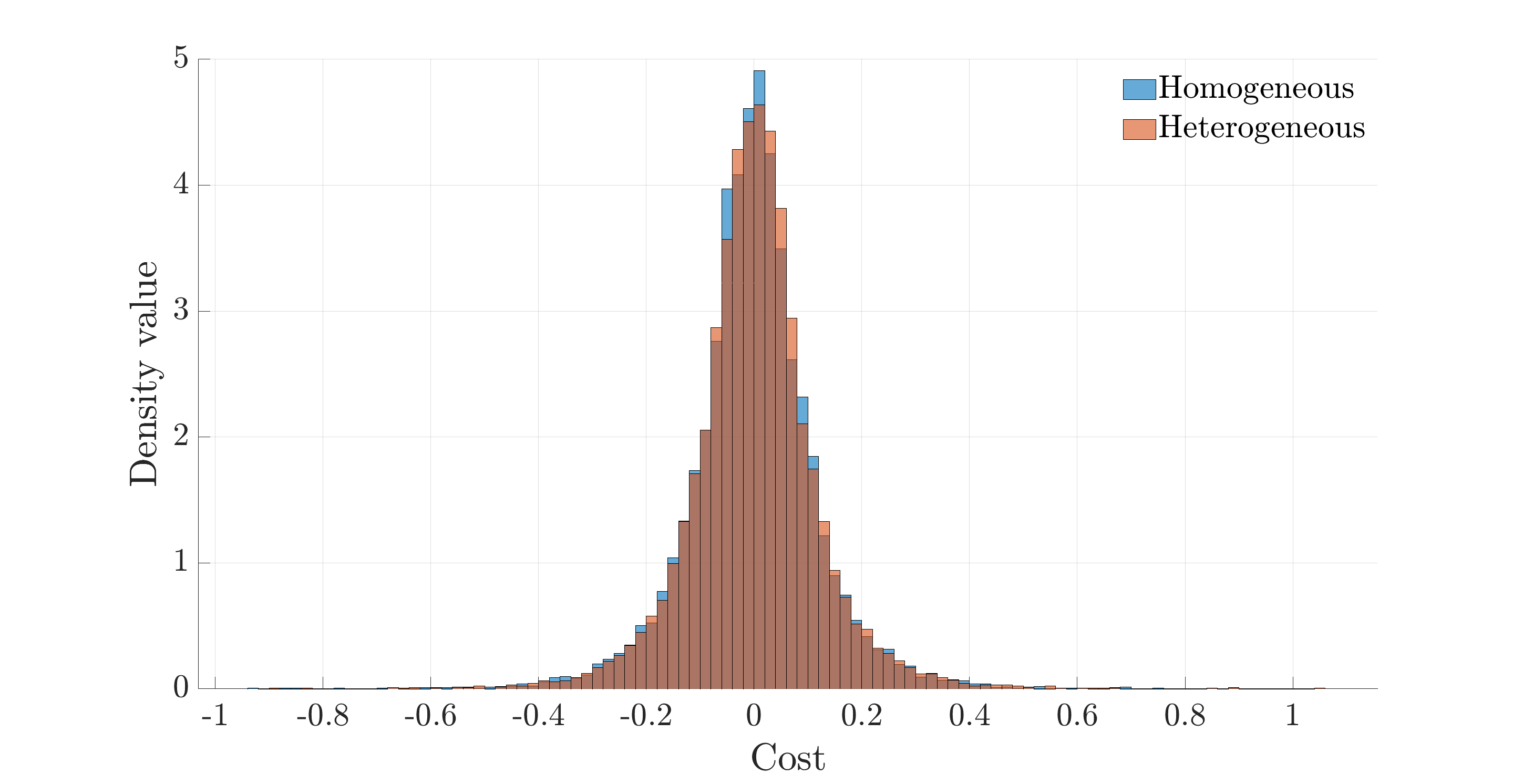}
\caption{Pure Endowment}
\label{fig:PE_HOMOvsHETERO}
\end{subfigure}
\hfill
\begin{subfigure}[h]{0.49\textwidth}
\centering
\includegraphics[width=\textwidth]{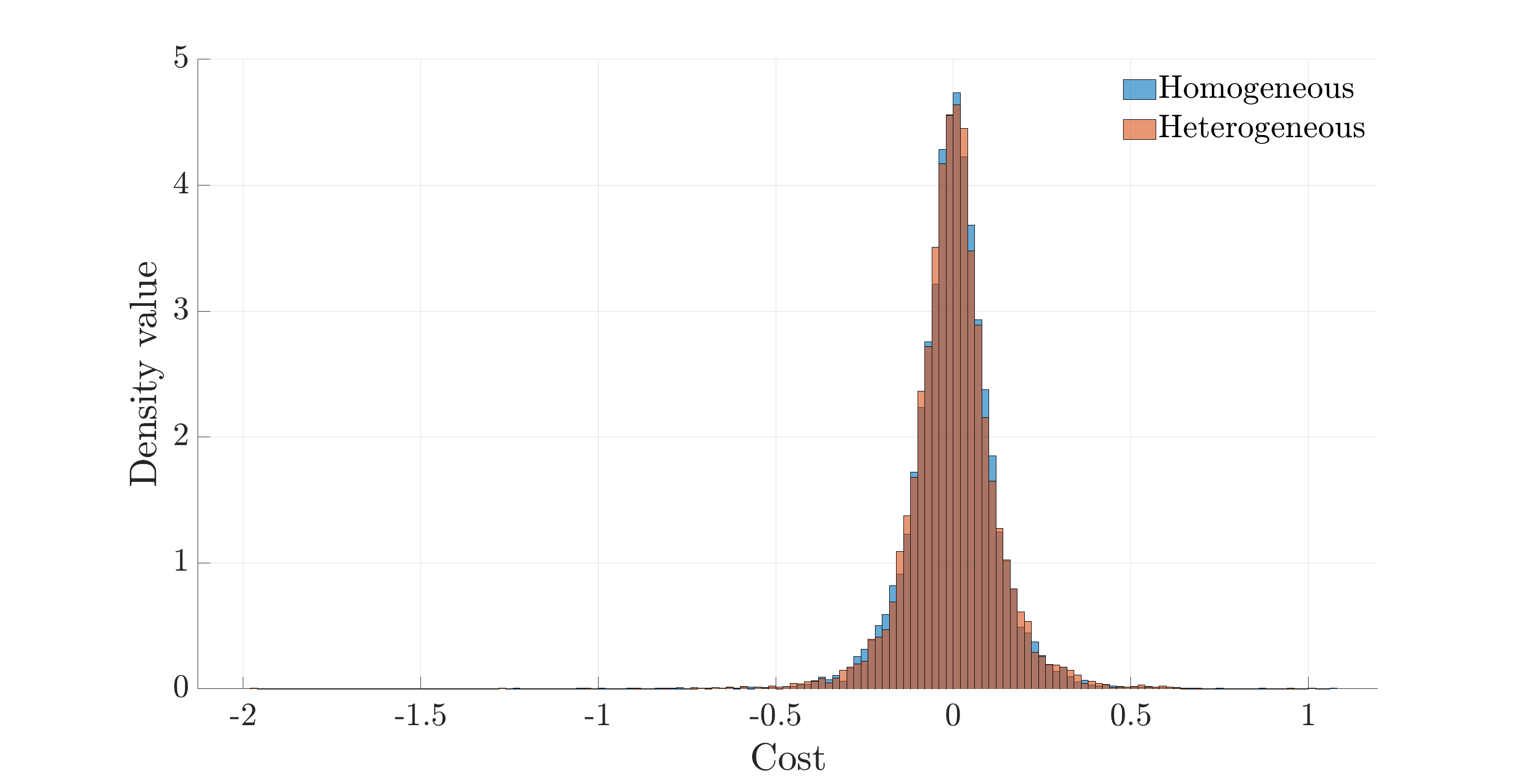}
\caption{Term Insurance} 
\label{fig:TI_HOMOvsHETERO}
\end{subfigure}
\hfill
\begin{subfigure}[h]{0.49\textwidth}
\centering
\includegraphics[width=\textwidth]{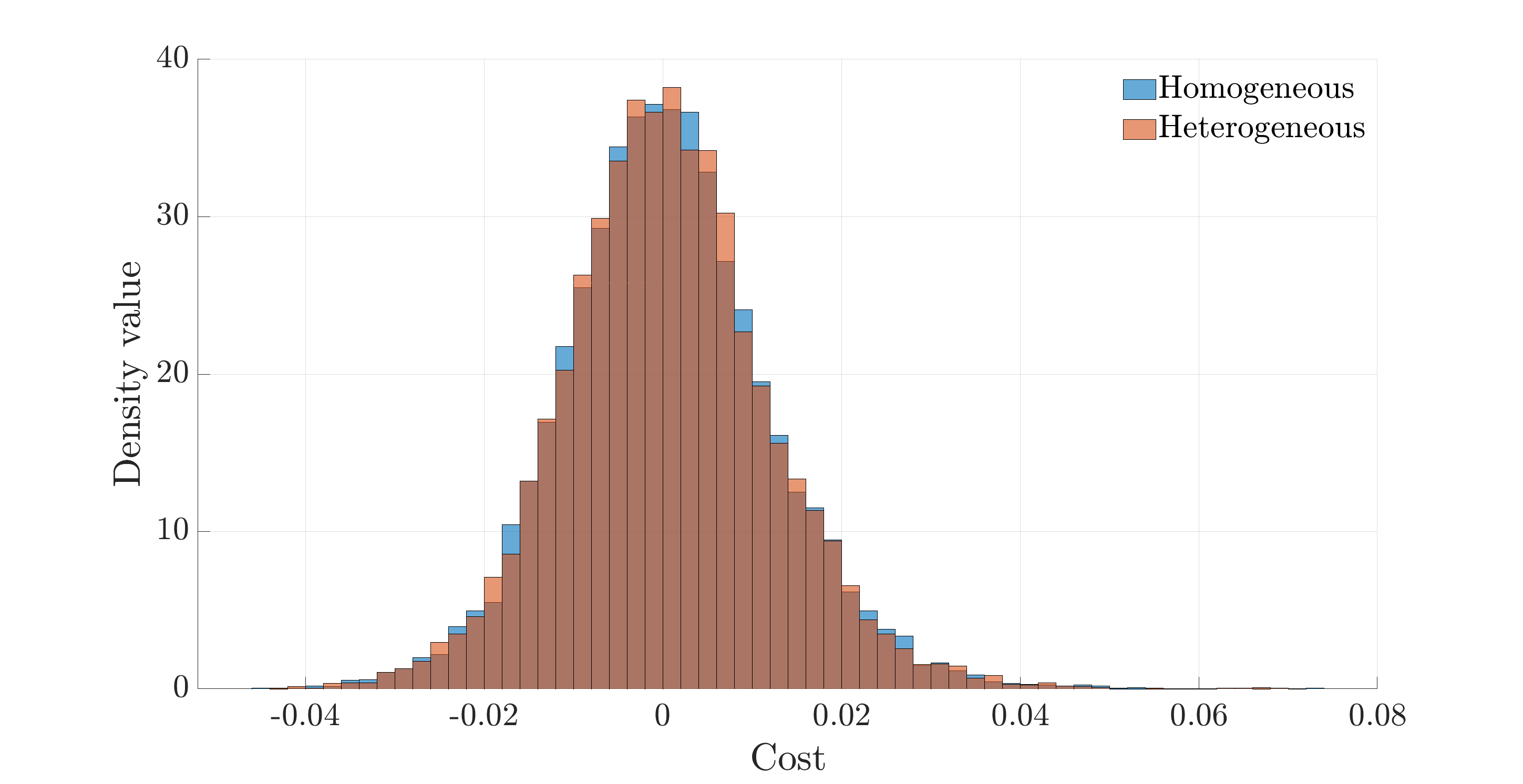}
\caption{Endowment Insurance} 
\label{fig:EI_HOMOvsHETERO}
\end{subfigure}
\caption{Cost of the continuous hedging: portfolio of 1000 policyholders of the same age of $60$ years (Homogeneous portfolio) Vs portfolio of 1000 policyholders with ages between $55$ and $65$ years (Heterogeneous portfolio).}\label{fig:HOMOvsHETERO}
\end{figure}

\section{Concluding remarks}\label{sect:conclusions}

In this paper, we develop a unified actuarial framework for the design and risk management of sustainable unit-linked life insurance products. The proposed setting jointly accounts for the construction of a carbon-efficient investment fund and the hedging of insurance liabilities written on such a fund. As a consequence, the insurer is exposed to multiple sources of risk, namely financial risk, mortality risk, and carbon risk, the latter acting as an additional non-hedgeable factor that contributes to market incompleteness.
On the asset side, we introduce a portfolio selection methodology in which carbon considerations are incorporated endogenously through a penalisation of terminal wealth. This approach allows for a flexible trade-off between financial performance and environmental impact, avoiding rigid exclusion rules and yielding portfolios that balance risk–return characteristics with carbon exposure. The associated stochastic control problem is solved via dynamic programming techniques under alternative specifications for the carbon-intensity process, including Ornstein–Uhlenbeck, CIR, and finite-state Markov dynamics. From a theoretical perspective, we establish a verification result and provide conditions ensuring the well-posedness of the value function. Empirical evidence based on S\&P 500 data indicates that even moderate levels of penalisation lead to a substantial reduction in portfolio carbon intensity.
On the liability side, we consider several standard unit-linked life insurance contracts and derive risk-minimising hedging strategies in an incomplete market framework. The hedging problem is addressed via quadratic criteria, leading to strategies that minimise the variance of the hedging costs. Numerical results show that the proposed dynamic hedging approach achieves a significant loss reduction compared to alternative methods. From a computational standpoint, the combination of a weak second-order discretisation scheme with conditioning-based variance reduction allows for an accurate and efficient implementation of the hedging strategy.
Overall, the framework highlights the strong interaction between asset allocation and liability hedging in the presence of sustainability constraints. 
Several extensions may be considered. In particular, introducing dependence between financial and insurance risks, e.g., through common stochastic factors as in \cite{ceci2022optimal}, would allow one to capture systemic shocks such as pandemics or climate-related catastrophes, and to assess their joint impact on asset returns and mortality risk. This would lead to more involved optimisation and hedging problems, that we leave for future research.

\section*{Acknowledgements and fundings}
The work of Katia Colaneri and Edoardo Lombardo has been partially funded by the European Union - Next Generation EU - Project PRIN 2022 [2022BEMMLZ - CUP E53D23005660006] with the title {\em Stochastic Control and Games and the Role of Information}. The work of Alessandra Cretarola has been partially funded by the European Union-Next Generation EU - Project PRIN [2022FPLY97 - CUP J53D23004540006] with the title {\em Modeling and Valuation of Financial Instruments for Climate and Energy Risk Mitigation}. This work has
been completed while Daniele Mancinelli was affiliated with University of Rome Tor Vergata and has been funded by European Union - Next Generation EU, Mission 4, Component 2 as part of the GRINS project - Growing Resilient, INclusive and Sustainable (PE0000018, CUP: E83C22004690001) - National Recovery and Resilience Plan (PNRR). The views and opinions expressed are solely those of the authors and do not necessarily reflect those of the European Union, nor can the European Union be held responsible for them. Katia Colaneri and Alessandra Cretarola are members of
Gruppo Nazionale per l’Analisi Matematica, la Probabilità e le loro Applicazioni (GNAMPA) of Istituto
Nazionale di Alta Matematica (INdAM).

\section*{Declaration of generative AI in scientific writing}
During the preparation of this work the authors used \textit{Writeful GPT} in the writing process in order to improve the readability and language of the manuscript. After using this tool, the authors reviewed and edited the content as needed and take full responsibility for the content of the published article.
\section*{Conflict of interest}
The authors declare no competing interests.
\bibliographystyle{apalike}
\bibliography{Biblio_1_corrected}

@article{pedersen2021responsible,
  author = {Pedersen, L. H. and Fitzgibbons, S. and Pomorski, L.},
  title = {Responsible investing: The ESG-efficient frontier},
  journal = {Journal of financial economics},
  volume = {142},
  number = {2},
  pages = {572--597},
  year = {2021}
}

@article{pastor2021sustainable,
  author = {P{\'a}stor, L. and Stambaugh, R. F. and Taylor, L. A.},
  title = {Sustainable investing in equilibrium},
  journal = {Journal of Financial Economics},
  volume = {142},
  number = {2},
  pages = {550--571},
  year = {2021}
}

@article{ceci2022optimal,
  author = {Ceci, C. and Colaneri, K. and Cretarola, A.},
  title = {Optimal reinsurance and investment under common shock dependence between financial and actuarial markets},
  journal = {Insurance: Mathematics and Economics},
  volume = {105},
  pages = {252--278},
  year = {2022}
}

@misc{ivass2024report,
  author       = {{IVASS}},
  title        = {Analysis of {IBIP} {P}olicies with {ESG} {C}haracteristics},
  year         = {2024},
  url          = {https://www.ivass.it/pubblicazioni-e-statistiche/pubblicazioni/altre-pubblicazioni/2024/analisi_polizze_ibips_esg/index.html}
}

@misc{eiopa2023costs,
  author       = {{EIOPA}},
  title        = {Costs and Past Performance Report},
  year         = {2023},
  url          = {https://www.eiopa.europa.eu/publications/costs-and-past-performance-report-2023_en}
}

@article{bhatia2025emission,
  author = {Bhatia, M. and Gugnani, R. and Yaqub, M. Z. and Tripathi, P. M. and Broccardo, L.},
  title = {Emission reduction strategies and negative emission solutions-pathways, drivers, and challenges},
  journal = {Journal of Cleaner Production},
  volume = {500},
  pages = {145263},
  year = {2025}
}

@article{verbist2025carbon,
  author = {Verbist, F. and Meus, J. and Moncada, J. A. and Valkering, P. and Delarue, E.},
  title = {Carbon removals meet {E}mission {T}rading {S}ystem design: {A} precautionary path towards integration},
  journal = {Energy Economics},
  volume = {145},
  pages = {108389},
  year = {2025}
}

@article{moller2001risk,
  author = {M{\o}ller, T.},
  title = {Risk-minimizing hedging strategies for insurance payment processes},
  journal = {Finance and Stochastics},
  volume = {5},
  pages = {419--446},
  year = {2001}
}

@article{vandaele2008locally,
  author = {Vandaele, N. and Vanmaele, M.},
  title = {A locally risk-minimizing hedging strategy for unit-linked life insurance contracts in a L{\'e}vy process financial market},
  journal = {Insurance: Mathematics and Economics},
  volume = {42},
  number = {3},
  pages = {1128--1137},
  year = {2008}
}

@article{alfonsi2010high,
  author = {Alfonsi, A.},
  title = {High order discretization schemes for the {CIR} process: application to affine term structure and {H}eston models},
  journal = {Mathematics of Computation},
  volume = {79},
  number = {269},
  pages = {209--237},
  year = {2010}
}

@article{alfonsi2024high,
  author = {Alfonsi, A. and Lombardo, E.},
  title = {High order approximations of the {C}ox--{I}ngersoll--{R}oss process semigroup using random grids},
  journal = {IMA Journal of Numerical Analysis},
  volume = {44},
  number = {4},
  pages = {2277--2322},
  year = {2024}
}

@article{bolton2021investors,
  author = {Bolton, P. and Kacperczyk, M.},
  title = {Do investors care about carbon risk?},
  journal = {Journal of Financial Economics},
  volume = {142},
  number = {2},
  pages = {517--549},
  year = {2021}
}

@incollection{LeGuenedalRoncalli2023,
  author    = {Le Guenedal, T. and Roncalli, T.},
  title     = {Portfolio construction with climate risk measures},
  booktitle = {Climate {I}nvesting: {N}ew {S}trategies and {I}mplementation {C}hallenges},
  pages = {49--86},
  publisher = {Emmanuel Jurczenko, Wiley},
  year = {2023}
}

@article{de2023esg,
  author={De Spiegeleer, J. and H{\"o}cht, S. and Jakubowski, D. and Reyners, S. and Schoutens, W.},
  title = {E{SG}: {A} new dimension in portfolio allocation},
  journal={Journal of Sustainable Finance and Investment},
  volume = {13},
  number = {2},
  pages = {827--867},
  year = {2023}
}

@article{bolton2022net,
  author = {Bolton, P. and Kacperczyk, M. and Samama, F.},
  title = {Net-zero carbon portfolio alignment},
  journal = {Financial Analysts Journal},
  volume = {78},
  number = {2},
  pages = {19--33},
  year = {2022}
}

@article{peng2024optimal,
  author = {Peng, F. and Yan, M. and Zhang, S.},
  title = {Optimal investment of defined contribution pension plan with environmental, social, and governance ({ESG}) factors in regime-switching jump diffusion models},
  journal = {Communications in Statistics-Theory and Methods},
  pages = {1--27},
  year = {2024}
}

@article{anquetin2022scopes,
  author = {Anquetin, T. and Coqueret, G. and Tavin, B. and Welgryn, L.},
  title = {Scopes of carbon emissions and their impact on green portfolios},
  journal = {Economic Modelling},
  volume = {115},
  pages = {105951},
  year = {2022}
}

@article{lagerkvist2020preferences,
  author = {Lagerkvist, C. J. and Edenbrandt, A. K. and Tibbelin, I. and Wahlstedt, Y.},
  title = {Preferences for sustainable and responsible equity funds - {A} choice experiment with {S}wedish private investors},
  journal = {Journal of Behavioral and Experimental Finance},
  volume = {28},
  pages = {100406},
  year = {2020}
}

@article{hartzmark2019investors,
  author = {Hartzmark, S. M. and Sussman, A. B.},
  title = {Do investors value sustainability? {A} natural experiment examining ranking and fund flows},
  journal = {The Journal of Finance},
  volume = {74},
  number = {6},
  pages = {2789--2837},
  year = {2019}
}

@article{andersson2016hedging,
  author = {Andersson, M. and Bolton, P. and Samama, F.},
  title = {Hedging climate risk},
  journal = {Financial Analysts Journal},
  volume = {72},
  number = {3},
  pages = {13--32},
  year = {2016}
}

@book{teschl2012ordinary,
  author = {Teschl, G.},
  title = {{O}rdinary {D}ifferential {E}quations and {D}ynamical {S}ystems},
  year = {2012},
  publisher = {Providence, Rhode Islande: American Mathematical Society}
}

@book{rogers2013optimal,
  author = {Rogers, L. C. G.},
  title = {Optimal investment},
  year = {2013},
  publisher = {Berlin, Heidelberg: Springer-Verlag}
}

@article{follmer2010minimal,
  author = {F{\"o}llmer, H. and Schweizer, M.},
  title = {Minimal martingale measure},
  journal = {Encyclopedia of Quantitative Finance},
  volume = {3},
  pages = {1200--1204},
  year = {2010},
  publisher={Wiley New York}
}

@incollection{schweizer2001,
  title = {A guided tour through quadratic hedging approaches},
  author = {Schweizer, M.}, 
  year = {2001},  
  editor = {Jouini, E. and Cvitanic, J. and Musiela, M.}, 
  booktitle = {Option {P}ricing, {I}nterest {R}ates and {R}isk {M}anagement},
  publisher = {Cambridge University Press, Cambridge},
  pages = {538--574}
}

@book{AF_book,
  title = {Partial differential equations of parabolic type},
  author = {Friedman, A.},
  year = {2008},
  publisher = {Courier Dover Publications}
}

@article{ceci2017unit,
  title = {Unit-linked life insurance policies: {O}ptimal hedging in partially observable market models},
  author = {Ceci, C. and Colaneri, K. and Cretarola, A.},
  journal = {Insurance: Mathematics and Economics},
  volume = {76},
  pages = {149--163},
  year = {2017}
}

@article{ceci2015hedging,
  title = {Hedging of unit-linked life insurance contracts with unobservable mortality hazard rate via local risk-minimization},
  author = {Ceci, C. and Colaneri, K. and Cretarola, A.},
  journal = {Insurance: Mathematics and Economics},
  volume = {60},
  pages = {47--60},
  year = {2015}
}

@article{colaneri2021classical,
  author = {Colaneri, K. and Frey, R.},
  title = {Classical solutions of the backward {PIDE} for Markov modulated marked point processes and applications to {CAT} bonds},
  journal = {Insurance: Mathematics and Economics},
  volume = {101},
  pages = {498--507},
  year = {2021}
}

@book{krylov1996lectures,
  author={Krylov, N. V.},
  title = {Lectures on elliptic and parabolic equations in {H}older spaces},
  volume = {12},
  year = {1996},
  publisher = {American Mathematical Soc.}
}

@article{BK2006exact,
  title={Exact simulation of stochastic volatility and other affine jump diffusion processes},
  author={Broadie, M. and Kaya, {\"O}.},
  journal={Operations research},
  volume={54},
  number={2},
  pages={217--231},
  year={2006}
}

@article{AL2025SIFIN,
  title = {High order approximations and simulation schemes for the log-{H}eston process},
  author = {Alfonsi, A. and Lombardo, E.},
  journal = {SIAM Journal on Financial Mathematics},
  volume = {16},
  number = {2},
  pages = {516--544},
  year = {2025}
}

@article{cai2017exact,
  title = {Exact simulation of the SABR model},
  author = {Cai, N. and Song, Y. and Chen, N.},
  journal = {Operations Research},
  volume = {65},
  number = {4},
  pages = {931--951},
  year = {2017}
}

@article{biagini2017risk,
  title = {Risk-minimization for life insurance liabilities with dependent mortality risk},
  author = {Biagini, F. and Botero, C. and Schreiber, I.},
  journal = {Mathematical Finance},
  volume = {27},
  number = {2},
  pages = {505--533},
  year = {2017}
}

@article{biagini2016risk,
  title = {Risk-minimization for life insurance liabilities with basis risk},
  author = {Biagini, F. and Rheinl{\"a}nder, T. and Schreiber, I.},
  journal = {Mathematics and Financial Economics},
  volume = {10},
  number = {2},
  pages = {151--178},
  year = {2016}
}

@article{moller1998risk,
  title = {Risk-minimizing hedging strategies for unit-linked life insurance contracts},
  author = {M{\o}ller, T.},
  journal = {ASTIN Bulletin: The Journal of the IAA},
  volume = {28},
  number = {1},
  pages = {17--47},
  year = {1998}
}

@incollection{ansel2006unicite,
  author = {Ansel, J. P. and Stricker, C.},
  title = {Unicit{\'e} et existence de la loi minimale},
  booktitle = {S{\'e}minaire de Probabilit{\'e}s XXVII},
  pages = {22--29},
  year = {2006},
  publisher = {Springer}
}

\newpage
\section*{Appendix}
\appendix
\numberwithin{equation}{section}
\renewcommand{\theequation}{\thesection\arabic{equation}}

\section{Introduction to Appendix}
This appendix complements the paper by providing technical details, proofs of the theoretical results, and a description of the numerical procedures used in the simulation study. The document is organised as follows. Section  \ref{sec:proofs} collects the proofs related to the stochastic control problem and the optimal investment strategy. Section \ref{sec:proofHJB} provides technical details on the existence and uniqueness of the solution of the Hamilton-Jacobi-Bellman equation when the carbon intensity process is modeled as a CIR process. Section \ref{sec:proofHedging} contains the derivation of the quadratic hedging strategies for pure endowment and term insurance contracts. Section \ref{sec:numerics} details the numerical procedures used to implement the model and perform the simulation study.

\section{Proofs of some technical results of Section \ref{sec:optimal_inv}}\label{sec:proofs}

\subsection{Proof of Theorem \ref{thm:ver_thm}}\label{sect:proof_of_ver_thm}
To begin we let discontinuous martingale of $\mathbf{C}$ be given by   
\begin{equation}
\de \bm{M}^{disc}_t = \int_{\mathcal Z} \bm{z} (m(\de t, \de \bm{z}) - \nu_t(\de \bm{z}) \de t),
\end{equation}
for some set $\mathcal Z \subset \mathbb{R}^d$, where $m([0,t)\times E)= \sum_{s \le t}\ind[\bfC_{s}-\bfC_{s-} \in E]$ counts the jumps of $\bfC$ in $[0,t]$ with sizes in the set $E \subset \mathcal{Z}$ and $\nu_t(\de \bm{z})$ is the compensator of $m(\de t, \de \bm{z})$. Recall that, to avoid excessive technicalities, we assume that $\mathcal{Z}$ is a compact subset of $\mathbb{R}^d$, that is $\nu_t(\mathcal{Z})<\bar{\nu}$ for some $\bar{\nu}<+\infty$, and for every $t \in [0,T]$. Consider the function $w(t,x, \bfc)$. Since this function is regular we apply It\^o formula  
\begin{align}
w(T, \tilde{X}^{\bmpi}_{T},\bfC_{T}) = & w(t, x,\bfc)+\int_t^{T} \left\{\partial_t w(s,\tilde{X}_{s}^{\bmpi},\bfC_{s})+\mcL w(s,\tilde{X}_{s}^{\bmpi},\bfC_{s}) \right\}\de s \\
&+ \int_t^{T} \partial_x w(s,\tilde X^{\bmpi}_s,\bfC_s)\tilde{X}^{\bmpi}_s\bmpi_s^\top\bfSigma\de\bfZ_s+\sum_{i=1}^{d}\int_t^{T}\partial_{c_i}w(s,\tilde{X}^{\bmpi}_s,\bfC_s)\de M^{i,cont}_s \\ 
& + \int_t^T \int_{\mathcal Z}\Delta_z w(s, \tilde{X}^{\bmpi}_s,\bfC_s)  (m(\de s, \de z) - \nu_s(\de z) \de s),
\end{align}
with $\Delta_z w(s, \tilde{X}^{\bmpi}_s,\bfC_s)=w(s, \tilde{X}^{\bmpi}_s,\bfC_{s-}+z) - w(s, \tilde{X}^{\bmpi}_s,\bfC_{s-})$, and where we intend $X^{\bmpi}_s, \bfC_s$ as the processes at time $s$, starting at time $t$ at values $x, \bfc$. 
We define the increasing sequence of stopping times as follows. 
We fix $t \ge 0$ and let, for every $u>t$
\begin{align}
I^1_u&:=\int_t^u  (\partial_x w(s,\tilde X^{\bmpi}_s,\bfC_s)\tilde{X}^{\bmpi}_s)^2\bmpi_s^\top\bfSigma \bfSigma^\top \bmpi_s \de s,\\
I^2_u&:= \sum_{i=1}^{d}  \int_t^u  \left(\partial_{c_i}w(s,\tilde{X}^{\bmpi}_s,\bfC_s)\right)^2\de \langle \bm{M}^{i,cont}\rangle_s,\\
I^3_u&:=\int_t^u \left(\sup_{\mathcal Z}\Delta_z w(s, \tilde{X}^{\bmpi}_s,\bfC_s)\right)^2   \nu_s(\mathcal Z) \de s;
\end{align}
then we define
\begin{equation}
\tau_n = \inf\left\{ u\ge t : I^1_u +  I^2_u + I^3_u\ge n\right\}.
\end{equation}
Then, we get that
\begin{align}
&w(T\wedge{\tau_n}, \tilde{X}^{\bmpi}_{T\wedge{\tau_n}},\bfC_{T\wedge{\tau_n}}) =w(t, x,\bfc)+\int_t^{T\wedge{\tau_n}} \left\{\partial_t w(s,\tilde{X}_{s}^{\bmpi},\bfC_{s})+\mcL w(s,\tilde{X}_{s}^{\bmpi},\bfC_{s}) \right\}\de s \\
&+ \int_t^{T\wedge{\tau_n}} \partial_x w(s,\tilde{X}^{\bmpi}_s,\bfC_s)\tilde{X}^{\bmpi}_s\bmpi_s^\top\bfSigma\de\bfZ_s+\sum_{i=1}^{d}\int_t^{T\wedge{\tau_n}}\partial_{c_i}w(s,\tilde{X}^{\bmpi}_s,\bfC_s)\de M^{i,cont}_s \\
& + \int_t^{T\wedge{\tau_n}} \int_{\mathcal Z}\Delta_z w(s, \tilde{X}^{\bmpi}_s,\bfC_s)  (m(\de s, \de z) - \nu_s(\de z) \de s).
\end{align}
Since $w$ satisfies the equation \eqref{eq:HJB}, we get that
\begin{align}
&w(T\wedge{\tau_n}, \tilde{X}^{\bmpi}_{T\wedge{\tau_n}},\bfC_{T\wedge{\tau_n}}) \le w(t, x,\bfc) \\
&+ \int_t^{T\wedge{\tau_n}} \partial_x w(s,\tilde X^{\bmpi}_s,\bfC_s)\tilde{X}^{\bmpi}_s\bmpi_s^\top\bfSigma\de\bfZ_s+\sum_{i=1}^{d}\int_t^{T\wedge{\tau_n}}\partial_{c_i}w(s,\tilde{X}^{\bmpi}_s,\bfC_s)\de M^{i,cont}_s \\ &+ \int_t^{T\wedge{\tau_n}} \int_{\mathcal Z}\Delta_z w(s, \tilde{X}^{\bmpi}_s,\bfC_s)  (m(\de s, \de z) - \nu_s(\de z) \de s).
\end{align}
By the admissibility condition \eqref{eq:adm_cond} on $\bmpi$ and the growth assumption on $w$, all stochastic integrals are true martingales up to the stopping time $\tau_n\wedge T$, with vanishing expected value.\\
Taking the expectation on both sides of the inequality and using the terminal condition of the Hamilton Jacobi Bellman equation,  we get that
\begin{equation}
w(t, x, \bfc) \ge \E^{\P^M}[w(T\wedge{\tau_n}, \tilde{X}^{\bmpi}_{T\wedge{\tau_n}},\bfC_{T\wedge{\tau_n}})].
\end{equation}
By the polynomial growth condition on $w$ and standard localization arguments, we can pass to the limit as $n \to \infty$, and using Lebesgue dominated convergence theorem we get that 
\begin{equation} w(t, x,\bfc) \ge \E^{\P^M}[U(\tilde{X}^{\bmpi}_{T})], 
\end{equation}
hence $w(t, x,\bfc)\ge v(t, x,\bfc)$. This shows point $(i)$ of the statement. 
Moreover, the equality holds for the maximiser $\bmpi^\star$, which concludes the proof. 
\subsection{Proof of Theorem \ref{thm:general_thm}}\label{sect:proof_general_thm}
Suppose that a classical solution $w$ of the Hamilton Jacobi Bellman equation \eqref{eq:HJB} can be written as
\begin{equation}\label{eq:ANSATZ}
w(t,x,\bfc)=
\begin{cases}
\dfrac{x^{1-\delta}}{1-\delta}\varphi(t,\bfc), & \delta\in(0,1)\cup(1,+\infty),\\
\log(x)+\varphi(t,\bfc), & \delta=1,
\end{cases}
\end{equation}
where $\varphi$ does not depend on $x$ and it is a positive function.  We let $\Phi^{\bmpi}(t,\bfc)$ be defined as 
\begin{equation}\label{eq:obj_fun}
\Phi^{\bmpi}(t,\bfc):=\bmpi^\top\left(\bmmu-\1r\right)-\dfrac{1}{2}\bmpi^\top\left(\delta\bfSigma\bfSigma^\top+\bfe(t,\bfc)\odot\bfD\bfD^\top\right)\bmpi.
\end{equation}
Then, we can rewrite the equation \eqref{eq:HJB} as follows.
\begin{itemize}
\item[(i)] If $\delta\in(0,1)\cup(1,+\infty)$, we get that \eqref{eq:HJB} is equivalent to 
\begin{equation}\label{eq:HJB_eq_CRRA}
\begin{cases}
\partial_t\varphi(t,\bfc)+\mcLC\varphi(t,\bfc)+\varphi(t,\bfc)\left[(1-\delta)r+(1-\delta)\displaystyle\max_{\bmpi\in\mathbb{R}^d}\Phi^{\bmpi}(t,\bfc)\right]=0,&(t,\bfc)\in[0,T]\times\mcD,\\
\varphi(t,\bfc)=1,&\bfc\in\mcD.
\end{cases}
\end{equation}
\item[(ii)] If $\delta=1$, corresponding the logarithmic case, the equation \eqref{eq:HJB} becomes:
\begin{equation}\label{eq:HJB_eq_LOG}
\begin{cases}
\partial_t\varphi(t,\bfc)+r+\mcLC \varphi(t,\bfc)+\displaystyle\max_{\bmpi\in\mathbb{R}^d}\Phi^{\bmpi}(t,\bfc)=0,&(t,\bfc)\in[0,T]\times\mcD,\\
\varphi(T,\bfc)=0,\quad \bfc\in\mcD.
\end{cases}
\end{equation}
\end{itemize}
We let $\bmpi^{\star}:=\argmax \Phi^{\bmpi}(t,\bfc)$. Taking the gradient and the Hessian of $\Phi^{\bmpi}$ with respect to $\bmpi$ we get that 
\begin{align}
	\nabla_{\bmpi}\Phi^{\bmpi}(t,\bfc)&=\left(\bmmu-\1r\right)-\left(\delta\bfSigma\bfSigma^\top+\bfe(t,\bfc)\odot\bfD\bfD^\top\right)\bmpi,\\
	\text{Hess}_{\bmpi}\Phi^{\bmpi}(t,\bfc)&=-(\delta\bfSigma\bfSigma^\top+\bfe(t,\bfc)\odot\bfD\bfD^\top).
\end{align}
Then, setting $\nabla_{\bmpi}\Phi^{\bmpi}(t,\bfc)=\mathbf{0}$, provides the candidate optimal strategy $\bmpi^{\star}(t, \bfc)$ given by \eqref{eq:sol_CRRA}. Moreover, since $\text{Hess}_{\bmpi}\Phi^{\bmpi}(t,\bfc)$ is negative definite for every $\bmpi$, this ensure that $\bmpi^\star(t, \bfc)$ in equation \eqref{eq:sol_CRRA} is the well defined global maximiser. Next we show that $\bmpi^\star(t, \bfc)$ is an admissible control. We note that 
\begin{align}
\|\bmpi^\star(t,\bfc)\|&\le\max\left(\text{Sp}\left(\left(\delta\bfSigma\bfSigma^\top+\bfe(t,\bfc)\odot\bfD\bfD^\top\right)^{-1}\right)\right)\|\bm{\mu}-\1r\|\\
	&\le\dfrac{\|\bm{\mu}-\1r\|}{\min\left(\text{Sp}\left(\delta\bfSigma\bfSigma^\top+\bfe(t,\bfc)\odot\bfD\bfD^\top\right)\right)}\\
	&\le\dfrac{\|\bm{\mu}-\1r\|}{\min\left(\text{Sp}\left(\delta\bfSigma\bfSigma^\top\right)\right)},
\end{align}
for every $(t,\bfc)\in[0,T]\times\mathcal{D}$, where $\text{Sp}(\cdot)$ denotes the spectrum of a matrix. Consequently, for every $\Xi\ge\frac{\|\bm{\mu}-\1r\|}{\min\left(\text{Sp}\left(\delta\bfSigma\bfSigma^\top\right)\right)}$, it holds that 
\begin{equation}
\mathbb{E}^{\mathbb{P}^M}\left[\int_0^T\|\bmpi^\star_s\|^2\de s\right]\le\mathbb{E}^{\mathbb{P}^M}\left[\int_0^T\Xi^2\de s\right]=\Xi^2T\less\infty,
\end{equation}
Hence, condition \eqref{eq:adm_cond} is satisfied. Replacing $\bmpi$ with $\bmpi^\star$ in equations \eqref{eq:HJB_eq_CRRA} and \eqref{eq:HJB_eq_LOG} yields the Cauchy problem in equation \eqref{eq:PDE}. If $\varphi$ is a classical solution of \eqref{eq:PDE}, then $w$ is also regular and $|w(t,x,\bfc)|\le c_1(1+|x| + |x|^{1-\delta})$ and solves $(2.5)$, hence the result follows from Theorem \eqref{thm:ver_thm}. 
\section{Proofs of some technical results of Section \ref{sec:examples}}\label{sec:proofHJB}
\subsection{Proof of Lemma \ref{cor:CIR_P}}\label{sect:proof_CIR_P}
Let $S\in[0,T)$,  $\mcR=(\frac{1}{R},R)^d$, for some $R>1$, $Q=[0,S)\times\mcR$ and consider the PDE with final condition 
\begin{equation}\label{cauchy_prob}
\begin{cases}
\partial_tv(t,\bfc)+ \mcLC v(t,\bfc)+H(t,\bfc) v(t,\bfc)= f(t,\bfc),\qquad&\mbox{ in }Q,\\
v(t,\bfc)= \varphi(t,\bfc)  ,\qquad&\mbox{ in }\partial_0Q,
\end{cases}	
\end{equation}
$\partial_0Q$ denoting the parabolic boundary of $Q$. Under the assumption that $\alpha_i(t),\,\bar{C}_i(t)$ are Lipschitz continuous for every $i=1, \dots, d$, it holds that $H$ and $f$ are Lipschitz continuous and bounded; then, they are in particular $p$-Hölder continuous for all $p\in(0,1)$, so the coefficients of the PDE \eqref{cauchy_prob} satisfy in $Q$ all the assumptions of, e.g. Theorem 9 and Corollary 2 in Chapter 3, Sec. 4 in \cite{AF_book}. Therefore a unique bounded solution $v\in C^{1,2}([0,S)\times\mcR)\cap C([0,S]\times\bar{\mcR})$, with $\bar \mcR=[\frac{1}{R},R]^d$, exists.
We let $\tau_\mcR$ be the first exit time of the process $\bfC$ starting in $(t,\bfc)\in[0,S)\times\mcR$ from the set $\mcR$, i.e. 
\[\tau_\mcR=\inf\{s>t: \bfC_s\notin \mcR\}.\]
We define, for every $t < s$ 
$$Z_s:=\exp\Big( \int_t^{s\wedge \tau_\mcR} H(u,\bfC_u)\de u \Big)v(s,\bfC_s) - \int_t^{s\wedge \tau_\mcR} \exp\Big( \int_t^u H (r,\bfC_r)\de r \Big) f (u,\bfC_u)\de u,$$
then $\{Z_s\}_{s \in [t, T\wedge \tau_\mcR]}$ is a martingale. Moreover, it holds that $(S\wedge \tau_\mcR,\bfC_{S\wedge \tau_\mcR})\in \partial_0 Q$, then, 
\begin{align*}
	v(t,c) &=\E^{t,\bfc}(Z_t)=\E^{t,\bfc}(Z_{S\wedge \tau_\mcR}) \\
	=&\E^{t,c}\Big[\exp\Big( \int_t^{S\wedge \tau_\mcR} H(s,\bfC_s)\de s \Big) \varphi(S\wedge \tau_\mcR,\bfC_{S\wedge \tau_\mcR}) - \int_t^{S\wedge \tau_\mcR} \exp\Big( \int_t^s H (u,\bfC_u)\de u \Big) f (s,\bfC_s)\de s\Big].
\end{align*}
By the strong Markov property,
\begin{multline*}
	\varphi(S\wedge \tau_\mcR,\bfC_{S\wedge \tau_\mcR})\\
	=\E\Big[\exp\left( \int_{S\wedge \tau_\mcR}^T H(s,\bfC_s)\de s  \right)\ind[\delta \in (0,1)\cup(1,+\infty)]- \int_{S\wedge \tau_\mcR}^T \exp\Big( \int_t^s H (u,\bfC_u)\de u \Big) f (s,\bfC_s)\de s \,\big|\, \mathcal{F}_{S\wedge \tau_{\mcR}}\Big].
\end{multline*}
By replacing above, we conclude that $v$ coincides with $\varphi$ on $Q$, and by letting $S \uparrow T$ and $R \uparrow  \infty$, the result follows on the whole domain. Hence, we get that $\varphi(t, \bfc)$ is a solution of \eqref{eq:PDE}. \\

Next we show uniqueness. Assume that $2\kappa_i\inf_{t\in[0,T]}\bar{C}_i(t)\ge \lambda_i^2$ for all $i=1,\dots,d$. Let $v\in\mcC([0,T]\times (0,\infty)^n)$ denote a solution to \eqref{eq:PDE}. We prove that $v=\varphi$.
Let $S_n<T$ and let $\mcR_n=(\frac{1}{R_n}, R_n)^d$ denote a sequence of rectangles, such that $Q_n=[0,S_n)\times \mcR_n\uparrow[0,T)\times (0,\infty)^n$. Let $v_n$ be the unique solution to
\begin{equation}\label{eq:pde_sequence}
	\begin{cases}
		\partial_tv_n(t,\bfc)+ \mcLC v_n(t,\bfc)+H(t,\bfc) v_n(t,c)= f(t,\bfc),\qquad&\mbox{ in }Q_n,\\
		v_n(t,\bfc)= v(t,\bfc)  ,\qquad&\mbox{ in }\partial_0Q_n.
	\end{cases}	
\end{equation}    
Since $v$ trivially solves the above PDE problem \eqref{eq:pde_sequence}, by uniqueness, we get that $v_n=v$ and
\begin{align}
	v(t,\bfc)= & \E^{t,x,y}\Big[\exp\Big( \int_t^{S_n\wedge \tau_{\mcR_n}} H(s,\bfC_s)\de s \Big)v(S_n\wedge \tau_{\mcR_n},\bfC_{S_n\wedge \tau_{\mcR_n}}) \\&- \int_t^{S_n\wedge \tau_{\mcR_n}} \exp\Big( \int_t^s H (u,\bfC_u)\de u \Big) f (s,\bfC_s)\de s\Big].
\end{align}
By Feller condition we get that, for a process $\bfC$ starting at $t$ at value $\bfC_t=\bfc$, with $c_i>0$ for every $i =1, \dots d$, $\P(\bfC_s\in(0,\infty)^d, \text{ for all } s)=1$. Then,  taking $n\to\infty$, one has $\tau_{\mcR_n}\uparrow \infty$ and since $v$ is continuous, we obtain $v\equiv \varphi$.
\section{Proofs of some technical results of Section \ref{sect:locally_risk_minimizing}}\label{sec:proofHedging}
\subsection{Proof of Proposition \ref{prop:FS_decomp_PE}}\label{sect:proof_FS_decomp_PE}
Using equations \eqref{eq:value}, \eqref{eq:B}, \eqref{eq:U}, and applying It\^{o}'s product rule, we compute:
\begin{equation}
	\de \E^{\mathbb{Q}^*}[G^{PE} \mid {\mathcal{G}}_t] 
	=B_{t-} \de U_t +   U_{t-} \de B_t,\label{eq:product}
\end{equation}
where the equality follows from independence of the financial market and the mortality risk under $\mathbb Q^*=\mathbb Q^M \times \mathbb P^I$; hence, $\langle B, U \rangle_t = 0$ and $\sum_{s\le t} \Delta B_s \Delta U_s = 0$, for every $t \in [0,T]$. Using \eqref{eq:U} and the martingale representation \eqref{eq:B_integral},
we get that equation \eqref{eq:product} becomes
	\[
	\de \E^{Q^\ast}[G^{PE}\mid\mathcal G_t]
	=
	B_{t-}\bm \beta_t^\top \de \widetilde{ \bf S}_t
	+
	B_{t-}\,\de A_t
	-
	U_{t-}e^{-\int_t^T\gamma(u)\,\de u}
	\bigl(\de H_t-(1-H_t)\gamma(t)\,\de t\bigr).
	\]
	Therefore,
	\[
	\E^{Q^\ast}[G^{PE}\!\mid\mathcal G_t]
	=
	G_0^{PE}
	+
	\int_0^t B_{s-}\bm \beta_s^\top \de \widetilde{\bm S}_s
	+
	\int_0^t\! B_{s-}\,\de A_s
	-
	\int_0^t\! U_{s-} e^{-\int_s^T\!\gamma(u)\,\de u}
	\bigl(\de H_s-(1-H_s)\gamma(s)\,\de s\bigr),
	\]
	where
	\(
	G_0^{PE}
	=
	U_0e^{-\int_0^T\gamma(s)\,ds}.
	\)
Hence,  $\{ B_{t-} \bm{\beta}_t\}_{t \in [0,T]}$ belongs to $\Theta({\mathbb{G}})$. Indeed, $0 \le B_t \le 1$,  for each $t \in [0,T]$, and $\bm \beta\in \Theta(\mathbb F))$. Moreover, define the process $O^{PE}=\{O_t^{PE}\}_{t \in [0,T]}$ as
\begin{equation}
	O_t^{PE}:= \int_0^t B_{s-}\,\de A_s
	-
	\int_0^t U_{s-}e^{-\int_s^T\gamma(u)\,\de u}
	\bigl(\de H_s-(1-H_s)\gamma(s)\,\de s\bigr), \quad t \in [0,T],
\end{equation}
we obtain
	\[
	G^{PE}
	=
	G_0^{PE}
	+
	\int_0^T B_{s-}\bm \beta_s^\top \de \widetilde{\bm S}_s
	+
	O_T^{PE},
	\qquad \P\text{-a.s.}
	\]
Then, $O^{PE}$ is a square-integrable $(\mathbb G,\mathbb P)$-martingale strongly orthogonal to the
$\mathbb G$-martingale part of $\widetilde \bfS$.
Taking $t=T$ gives the F\"ollmer-Schweizer decomposition of $G^{PE}$ and this concludes the proof.

\subsection{Proof of Proposition \ref{Prop:characterization_beta_B_PE}}\label{Prop:proof_characterization_beta_B_PE}
First, we observe that the PDE \eqref{eq:pdeF} has a unique classical solution under very general conditions on the infinitesimal generator $\mathcal{L}^\bfC$, see, e.g. \cite{colaneri2021classical}. 
For convenience of notation, let $\widetilde\varphi:=({\bm{\beta}},{\zeta})$ denote the pseudo-optimal strategy for the payoff
	\(e^{-rT}\varphi(X_T^{\pi^\star})\) with respect to \(\mathbb F\). In view of \cite[Proposition 6.2]{ceci2015hedging} and equation \eqref{eq:FS_phi_i}, we get that the (normalized) value of the strategy $\widetilde\varphi$ is given by
\begin{align}
	V_t(\widetilde\varphi)& =\mathbb{E}^{\mathbb Q^M}\!\left[e^{-rT}\phi\!\left(X_T^{\pi^\star}\right)\middle|\mathcal F_t\right] =U_0 + \int_0^t {\bm \beta}_u \de \widetilde{\bf S}_u + A_t, \quad t \in [0,T],
\end{align}
where 
$$
{\bm \beta}_t=
\frac{\frac{\de \langle V(\widetilde{\varphi}), \widetilde{\bf S}\rangle_t}{\de t}}{\frac{\de \langle \widetilde{\bf S}\rangle_t}{\de t}}.
$$ 
On the other hand, if we denote $Y_t^{\bmpi^{\star}} = e^{-rt} X_t^{\bmpi^{\star}}$, for $t \in [0,T]$, the discounted fund value process $Y^{\bmpi^{\star}}=\{Y_t^{\bmpi^{\star}}\}_{t \in [0,T]}$ is also a $(\mathbb F,\mathbb{Q}^M)$-martingale.  Then, by \cite[Proposition 6.2]{ceci2015hedging} and the Galtchouk-Kunita-Watanabe decomposition of $e^{-rT}\phi\!\left(X_T^{\pi^\star}\right)$ with respect to $Y^{\bmpi^{\star}}$ under $\mathbb{Q}^M$ we have
$$
V_t(\widetilde{\varphi}) = \bar{U}_0 + \int_0^t \widetilde{\beta}_u \de Y_u^{\bmpi^{\star}} + \bar{A}_t,\quad t \in [0,T]
$$
where $$
\widetilde{\beta}_t=
\frac{\frac{\de \langle V(\widetilde{\varphi}), Y^{\bmpi^{\star}}\rangle_t}{\de t}}{\frac{\de \langle Y^{\bmpi^{\star}}\rangle_t}{\de t}},
$$ 
{ $\bar U_0 \in L^2(\mathcal{F}_0, \mathbb{Q}^M)$ and $\bar A = \{\bar A_t\}_{t \in [0,T]}$ is a square-integrable $(\mathbb{F}, \mathbb{Q}^M)$-martingale with $\bar A_0 = 0$, strongly orthogonal to $Y^{\bmpi^{\star}}$.}
It is easy to see that $U_0=\bar{U_0}$ as they both correspond to $e^{-rT}\mathbb{E}^{\mathbb Q^M}\!\left[\phi\!\left(X_T^{\pi^\star}\right)\right]$ and that $A_t=\bar{A}_t$, since any martingale that is orthogonal to  $\widetilde{\bf S}$ is also orthogonal to $Y^{\bmpi^\star}$. Finally, using the fact that $Y^{\bmpi^\star}$ can be expressed in terms of $\widetilde{\bf S}$ and the uniqueness of the Galtchouk-Kunita-Watanabe decomposition we get that  
\begin{equation}\label{eq:beta-prices}
	\bm\beta_t
	=\widetilde{\beta}_t Y_t^{{\bm\pi}^\star}\, \bm\pi^{\star, \top}_t \mathrm{diag}(\widetilde{\bf S}_t)^{-1}, \quad t \in [0,T].
\end{equation}
From the dynamics of the process $Y^{\bmpi^\star}$ under $\Q^M$, 
\[
\de Y^{\bmpi^\star}_t=Y^{\bmpi^\star}_t \bmpi_t^{\star, \top} \mathbf{\Sigma} \de\mathbf{\hat{Z}}_t,
\]
we immediately get that $\de \langle Y^{\bmpi^{\star}}\rangle_t= (Y^{\bmpi^\star}_t)^2 \bmpi_t^{\star, \top}\mathbf{\Sigma}\mathbf{ \Sigma}^\top \bmpi_t^{\star}$. 
Next, by Markovianity, we have 
\begin{align}
	V_t(\widetilde{\varphi}) =F(t,X_t^{\bm \pi^{\star}}, \bfC_t), \quad t \in [0,T]. 
\end{align}
Moreover, $V(\widetilde{\varphi})$ is a $(\mathbb{F}, \Q^M)$ martingale, which implies that the function $F(t,x, \bfc)$ solves the backward PDE \eqref{eq:pdeF}. Applying It\^o's formula to $F(t,X_t^{{\bm \pi}^\star},\mathbf C_t)$ we get that 
\begin{align}
	\de F(t,X_t^{{\bm \pi}^\star},\mathbf C_t) =&\Big({\frac{\partial F}{\partial t}(t,X_t^{{\bm \pi}^\star},\mathbf C_t)}+
	\frac{\partial F}{\partial x}(t,X_t^{{\bm \pi}^\star},\mathbf C_t) X_t^{{\bm \pi}^\star} r \\
	&+ \frac{1}{2}\frac{\partial^2 F }{\partial x^2}(t,X_t^{{\bm \pi}^\star},\mathbf C_t) (X_t^{{\bm \pi}^\star})^2 \bmpi^{\star,\top} \bfSigma \bfSigma^\top \bmpi^\star + \mathcal{L}^\bfC F(t,X_t^{{\bm \pi}^\star},\mathbf C_t)\Big) \de t\\
	&+\frac{\partial F}{\partial x}(t,X_t^{{\bm \pi}^\star},\mathbf C_t) X_t^{{\bm \pi}^\star} \bmpi^{\star,\top} \bfSigma \de \whbfZ_t  +{\nabla_{\bfc}F}(t,X_t^{{\bm \pi}^\star},\mathbf C_t) \de \mathbf C_t^m,
\end{align}
where the notation $\mathbf C_t^m$ is used to indicate the martingale part of $\mathbf C_t$ and $\nabla_{\bfc}F(t, x, \bfc)$ the gradient with respect to $\bfc=(c_1, \dots, c_d)$. Using \eqref{eq:pdeF} to eliminate the drift and taking the predictable covariation with respect to $Y^{\bmpi^\star}$
yields
\[
\frac{\de \langle V(\widetilde{\varphi}), Y^{\bmpi^{\star}}\rangle_t}{\de t} = \frac{\partial F}{\partial x}(t,X_t^{{\bm \pi}^\star},\mathbf C_t) X_t^{{\bm \pi}^\star} Y^{\bmpi^{\star}}_t \bmpi^{\star,\top}  \bfSigma \bfSigma^\top \bmpi_t^{\star}=  \frac{\partial F}{\partial x}(t,X_t^{{\bm \pi}^\star},\mathbf C_t) (X_t^{{\bm \pi}^\star})^2e^{-rt}  \bmpi_t^{\star,\top}  \bfSigma \bfSigma^\top \bmpi_t^{\star}
\]
and hence, the following representation for $\widetilde{\beta}$:
\begin{equation}\label{eq:beta-returns}
	\widetilde{\beta}_t
	= e^{rt} \frac{\partial F}{\partial x}\!\big(t,X_t^{\pi^\star},\mathbf C_t\big), \quad t \in [0,T].
\end{equation}
Equivalently, writing the decomposition with respect to $\widetilde{\bf S}$, (see equation \eqref{eq:beta-prices}) we get the claimed result. 

\begin{remark}
	Clearly, one could use the the Galtchouk-Kunita-Watanabe decomposition with respect to $\widetilde{\bf S}$ and compute directly the strategy ${\bm \beta}$ that turns out to be $${\bm \beta}=\nabla_{\widetilde{\bf s}} G(t, \widetilde{\bf S}, {\bf C}),$$ where the notation $\nabla_{\widetilde{\bf s}} G(t, \widetilde{\bf s}, {\bf c})$ indicates the gradient with respect to variables $\widetilde{\bf s}=(\widetilde{s}_1, \dots, \widetilde{s}_d)$ and the function $G$ solves the equation 
	\begin{equation}
		\frac{\partial G}{\partial t}(t, \widetilde{\bf s}, {\bf c}) + \frac{1}{2}\sum_{j, l=1}\frac{\partial^2 G }{\partial \widetilde{s_j}\widetilde{s_l}}(t, \widetilde{\bf s}, {\bf c})  \widetilde{s_j} \widetilde{s_l} (\bfSigma \bfSigma^\top)_{j,l}  + \mathcal{L}^\bfC G(t, \widetilde{\bf s}, {\bf c})=0.
	\end{equation}
	However, this approach is less convenient for the numerics as it would require the simulation of the paths of all stock prices and the computation of all the derivatives of $G$ with respect to the components $(\widetilde{s}_1, \dots, \widetilde{s}_d)$.
\end{remark}

\subsection{Proof of Proposition \ref{prop:TI}}\label{app:C_TI}
We aim to compute the pseudo-optimal strategy $(\bm{\eta}^{\star}, \zeta^{\star})$ for the term insurance contract whose payoff is given by \eqref{eq:FS_phi_i}. Similarly to the case of pure endowment contracts, the optimal value process $V(\bm{\eta}^{\star}, \zeta^{\star})=\{V_t(\bm{\eta}^{\star}, \zeta^{\star})\}_{t \in [0,T]}$ is characterized by 
	\[
	V_t(\bm \eta^\star,\zeta^\star)
	=
	\mathbb E^{\Q^\ast}\left[
	\int_0^T e^{-ru}\psi(u,X_u^{\bmpi^\star})\,\de H_u
	\Bigm| \mathcal G_t
	\right].
	\]
	Therefore,
	\[
	V_t(\bm \eta^\star,\zeta^\star)
	=
	\int_0^t e^{-ru}\psi(u,X_u^{\bmpi^\star})\,\de H_u
	+
	\mathbb E^{\Q^\ast}\left[
	\int_t^T e^{-ru}\psi(u,X_u^{\bmpi^\star})\,\de H_u
	\Bigm| \mathcal G_t
	\right],
	\qquad t\in[0,T].
	\]
Since the process $\{\psi(t, X^{\bmpi^{\star}}_{t})\}_{t\in(0,T]}$ is ${\mathbb{G}}$-predictable, we obtain:
	\[
	\mathbb E^{\Q^\ast}\left[
	\int_t^T e^{-ru}\psi(u,X_u^{\bmpi^\star})\,\de H_u
	\Bigm| \mathcal G_t
	\right]
	=
	\mathbb E^{\Q^\ast}\left[
	\int_t^T e^{-ru}\psi(u,X_u^{\pi^\star})(1-H_u)\gamma(u)\,\de u
	\Bigm| \mathcal G_t
	\right].
	\]
	Using the independence between the financial market and the mortality risk under
	$\Q^\ast=\Q^M\times \P^I$, we get
	\[
	V_t(\bm \eta^\star,\zeta^\star)
	=
	\int_0^t e^{-ru}\psi(u,X_u^{\bmpi^\star})\,\de H_u
	+
	\int_t^T U(u,t)\,B(u,t)\,\de u,
	\qquad t\in[0,T],
	\]
	where
	\[
	U(u,t):=
	\mathbb E^{\Q^M}\left[e^{-ru}\psi(u,X_u^{\bmpi^\star})\mid\mathcal F_t\right],
	\qquad
	0\le t\le u\le T,
	\]
	and
	$B(u)=\{B(u,t)\}_{t\in[0,u]}$ is given by \eqref{eq:Bu_expl}. For every \(u\in[0,T]\), we assume that the discounted payoff
	\(e^{-ru}\psi(u,X_u^{\bmpi^\star})\) admits the F\"ollmer--Schweizer decomposition with respect to
	\(\widetilde S\) and \(\mathbb F\), namely
\begin{align}\label{FS-Phi}
	e^{-ru}\psi(u, X^{\bmpi^{\star}}_{u}) = U(u, 0) + \int_0^u \bm{\beta}^{\top}_r(u) \de \widetilde{\bfS}_r + A_u(u), \quad  \P^M\text{-a.s.}, \quad  u \in [0,T],
\end{align}
where $U(u, 0) \in L^2(\mathcal{F}_0,  \P^M)$, $\bm{\beta}(u):=\{ \bm{\beta}_t(u)\}_{t\in[0,u]}\in\Theta(\mathbb F)$, and $A(u):=\{A_t(u)\}_{t\in[0,u]}$ is a square-integrable $(\mathbb{F}, \P^M)$-martingale with $A_0(u)=0$, strongly orthogonal to the martingale part of $\widetilde{\bfS}$. Since $\widetilde{\bfS}$ is an $(\mathbb{F}, \Q^M)$-martingale, for every $0 \le t \le u \le T$ we get that
\begin{align}
	U(u,t) = \E^{\mathbb{Q}^M} \left[e^{-ru} \psi(u, X^{\bmpi^{\star}}_{u}) \mid \mathcal{F}_t \right] 
	= U_0(u) + \int_0^t \bm{\beta}^{\top}_r(u) \de\widetilde{\bfS}_r + A_t(u), 
\end{align}
where $A(u)$ is a $(\mathbb{F}, \Q^M)$-martingale by definition of minimal martingale measure. Next, we assume some structure for the process $A(u)$. In particular, we suppose that there exists a process $\varsigma(u)=\{\varsigma_t(u)\}_{t\in[0,u]}$ for every $u \in [0,T]$ such that
\begin{equation}
	A_t(u) = \int_0^t \varsigma_r(u) \de A_r, \quad t \in [0,u],
\end{equation}
and that 
\begin{equation}
	\E^{\P^M} \left[ \int_0^T {\varsigma_t}^2(u) \de \langle A \rangle_t \right] < \infty,
\end{equation}
where $A$ is a square-integrable $(\mathbb{F}, \P^M)$-martingale, strongly orthogonal to the martingale part of $\widetilde{\bfS}$. This is an additional structural assumption on the family $\{A(u),  u\in[0,T]\}$, introduced to obtain a tractable representation of the orthogonal martingale term. By repeating the same arguments as in the pure endowment case, we get that 
$G^{TI}$ admits the F\"ollmer–Schweizer decomposition:
\begin{equation}\label{eq:GTI}
	G^{TI} = G^{TI}_0 + \int_0^T \left( \int_t^T B_{t-}(u)\bm{\beta}^{\top}_t(u)\,du \right) \de \widetilde{\bfS}_t + O_T^{TI} \quad {\mathbb P}\text{-a.s.},
\end{equation}
where $G^{TI}_0 = \E^{\mathbb{Q}^*}[G^{TI}]$. Moreover, the orthogonal martingale term is given by
\begin{align}
	O_t^{TI} &= \int_0^t  e^{-ru}\psi(u, X^{\bmpi^{\star}}_{u-})\,\big(\de H_u - (1-H_u) \gamma(u)\, \de u\big) \nonumber \\
	&\quad + \int_0^t  \left( \int_r^T U(u, r-)\xi_r(u)\,du \right) \big(\de H_r - (1-H_r)\gamma(r) \,\de r\big) \nonumber \\
	&\quad + \int_0^t  \left( \int_r^T B_{r-}(u)\varsigma_r(u)\,du \right) \de A_r,\quad t\in[0,T], \label{eq:K1}
\end{align}
with  $\xi_r(u)=\gamma(u) e^{-\int_r^u\gamma(s) \de s}$, $0\le r\le u\le T$. Hence, by \ref{prop:FS}, the pseudo-optimal strategy is given by
\[
\bm \eta_t^\star
=
\int_t^T B(u,t-)\beta_t(u)\,\de u,
\qquad t\in[0,T],
\]
and
\[
\zeta_t^\star
=
V_t(\bm \eta^\star,\zeta^\star)
-
\left(
\int_t^T B(u,t-)\beta_t(u)\,\de u
\right)^\top
\widetilde {\bm S}_t,
\qquad t\in[0,T].
\]
This concludes the proof.

\section{Numerics}\label{sec:numerics}
In this Section, we introduce a Monte Carlo framework to evaluate and hedge unit linked contracts on the optimized green investment fund. The methodology involves two key steps: first, the implementation of efficient simulation schemes for the fund and its underlying carbon intensity process, and second, the use of conditioning techniques. This approach allows for variance reduction, enabling the efficient pricing and hedging of contracts such as pure endowments, term insurances, and endowment insurances. 
\subsection{An efficient scheme for the Fund}
We assume that a given simulation scheme (exact or approximate) provides a discretized path for $\bfC$ on a uniform time grid $\{jh\mid j=0,\ldots,N\}$, where the time step is $h=T/N$ for some  $N\in\N^*$.  We denote the discretized path of $\bfC$ as $\whbfC=(\whC_1,\ldots,\whC_d)$.\footnote{The case of the CIR process is discussed at the end of this section.} Let $\bmpi^\star$ be the optimal control map, and define the approximated optimal weights over the same uniform time grid as
\begin{equation}\label{optimal_approx_weights}
	\widehat{\bmpi}_{jh} =\bmpi^\star(jh, \whbfC^N_{jh}) \quad \text{ for all }j\in\{0,\ldots,N\}.
\end{equation}
To introduce an efficient scheme for the fund $X^{\bmpi^\star}$, we use a conditioning technique which is  well-known in the literature of schemes for stochastic volatility processes (see, e.g., the Heston exact scheme by \cite{BK2006exact}, the SABR exact scheme by \cite{cai2017exact}, and the Heston second order scheme by \cite{AL2025SIFIN}). For hedging purposes, we consider the dynamics of the fund $X^{\bmpi^\star}$ under the minimal martingale measure  $\Q^M$, 
\begin{equation}\label{eq_optimal_portf_sec5}
	\de X_t^{\bmpi^\star}=r X_t^{\bmpi^\star}\de t+X_t^{\bmpi^\star}\bmpi_t^{\star,\top}\bfSigma\de\bfZ_t.
\end{equation}
It is easy to see that, for every $t,s>0$,
\begin{equation}\label{exact_scheme_X}
	X_{t+s}^{\bmpi^\star}=X^{\bmpi^\star}_t\exp\left(\int_t^{t+s} \Big(r-\frac 12 \bmpi^{\star,\top}_u\bfSigma\bfSigma^\top \bmpi_u^{\star}\Big)\de u + \int_t^{t+s} \bmpi_u^{\star,\top}\bfSigma\de \bfZ_u \right),
\end{equation}
Moreover, because of the independence between $\bfZ$ and $\bmpi^\star$, conditioning to the realization of $X_{t}^{\bmpi^\star}$ and $\{\bmpi^\star_u,u\in[t,{t+s}]\}$, the law of the ratio $X_{t+s}^{\bmpi^\star}/X_{t}^{\bmpi^\star}$ is $$ \log\mcN\left(\int_t^{t+s}\Big(r-\frac 12 \bmpi^{\star,\top}_u\bfSigma\bfSigma^\top \bmpi_u^{\star}\Big)\de u, \,\int_t^{t+s}\bmpi^{\star,\top}_u\bfSigma\bfSigma^\top \bmpi_u^{\star}\de u\right),$$ where $\log\mcN(a,b)$ stands for lognormal distribution of mean $a$ and variance $b$.
The latter implies that we can simulate any marginal of the process $X^{\bmpi^\star}$ given a starting point, a discretization of the optimal weights process $\bmpi^\star$ as the one in \eqref{optimal_approx_weights} and a rule to discretize the integral depending on $\bmpi^\star$ (we use the trapezoidal rule). Therefore, an approximation of $X^{\bmpi^\star}$ over all the grid $\{jh\mid j=0,\ldots,N\}$, given $\widehat{\bmpi}$, is 
\begin{align}
	\whX^N_{0} &= x, \nonumber\\
	\whX^N_{(j+1)h} &=\begin{multlined}[t]
		\whX_{jh}\exp\Bigg( rh  - \frac h4 (\whbmpi^\top_{jh}\bfSigma\bfSigma^\top \whbmpi_{jh}+  \whbmpi^\top_{(j+1)h}\bfSigma\bfSigma^\top \whbmpi_{(j+1)h}) \\ +\sqrt{ \frac h2 (\whbmpi^\top_{jh}\bfSigma\bfSigma^\top \whbmpi_{jh}+  \whbmpi^\top_{(j+1)h}\bfSigma\bfSigma^\top \whbmpi_{(j+1)h}})F_{j+1}\Bigg), \quad \text{ for all }j\in\{0,\ldots,N-1\},
	\end{multlined} \label{Nstep_whX}
\end{align}
where $(F_j)_{j=0, \dots, N-1}$ is a set of i.i.d. standard Gaussian random variable independent of $\whbmpi$. 
For some products, as the pure endowment contracts, one needs only to simulate $\whX$ at maturity $T$; in this case, one can save some computation time by using the following, equivalent in law, scheme
\begin{multline}\label{1step_whX}
	\whX^N_{T} = \whX^N_{0}\exp\Bigg( rT   - \frac h4 \sum_{j=0}^{N-1}(\whbmpi^\top_{jh}\bfSigma\bfSigma^\top \whbmpi_{jh} +  \whbmpi^\top_{(j+1)h}\bfSigma\bfSigma^\top \whbmpi_{(j+1)h}) \\+ \sqrt{ \frac h2 \sum_{j=0}^{N-1}(\whbmpi^\top_{jh}\bfSigma\bfSigma^\top \whbmpi_{jh}+  \whbmpi^\top_{(j+1)h}\bfSigma\bfSigma^\top \whbmpi_{(j+1)h})} F\Bigg).
\end{multline} 
where $F$ is a standard Gaussian random variable independent of $\whbmpi$.
\footnote{We can follow the same steps for the dynamic under the physical measure $\P^M$. The construction is the same, and the extra drift term in the dynamic boils down to an extra factor $\exp\big( \frac{h}{2}\sum_{l=0}^{j}(\whbmpi_{lh}+\whbmpi_{(l+1)h})^\top (\bmmu- r \1)\big)$ in the point $\whX^N_{(j+1)h}$ for every $j\in\{0,\ldots, N-1\}$.} A key advantage of these schemes is their efficiency. Indeed, rather than simulating the paths of all stocks, which would require $d$ sets of Gaussian random variables, one can use a single set of Gaussian random variables, or, for specific cases as for example the pure endowment, even one random variable.
\subsubsection*{A scheme for carbon intensity $\bfC$: the example of the CIR process}
In what follows, we present a numerical scheme for carbon intensity assuming that each component follows a Cox-Ingersoll-Ross (CIR) process, and we illustrate the efficient second-order scheme.
For all $i\in\{1,\ldots,d\}$ carbon intensity of each firm is
\begin{equation}\label{eq:carbon_emiss_sec5}
	\de C_{i,t}=\kappa_i\left(\bar{C}_i-C_{i,t}\right)\de t+\lambda_i\sqrt{C_{i,t}}\de W_{i,t},\quad C_{i,0}=c_i.
\end{equation}
Exact simulation schemes for the CIR process exist (e.g., based on CDF inversion or its representation as a Poisson weighted gamma mixture), but they are typically computationally intensive. This high cost makes them a suboptimal choice for our analysis, which relies on generating whole process trajectories.
Therefore, we implement a multidimensional version of the well-known Ninomiya-Victoir scheme for the CIR process (see \cite{alfonsi2010high} and \cite{alfonsi2024high} for convergence, regularity results and extensions). \footnote{Note that the scheme is well-defined if $\lambda_i^2\le 4\kappa_i \bar{C}_i$, which is weaker than the Feller's condition $\lambda_i^2\le 2\kappa_i\bar{C}_i$ for each $i$.}  We define the auxiliary maps
\begin{equation}\label{eta_i_maps}
	\eta_i(t,c,\omega)= \left(\kappa_i\bar{C}_i-\frac{\lambda_i^2}{4}\right)\psi_{\kappa_i}\left(\frac{t}{2}\right) +e^{-\kappa_i t/2}{\kappa_i}\left(\sqrt{(\kappa_i\bar{C}_i-\frac{\lambda_i^2}{4})\psi_{\kappa_i}\left(\frac{t}{2}\right) +e^{-\kappa_i t/2}c}+ \frac{\lambda_i}{2}\omega\right)^2,
\end{equation}
where $\psi_{\kappa}(t)=(1-e^{-\kappa t})/\kappa$ if $\kappa\neq 0$ and $\psi_{\kappa}(t)=t$ if $\kappa=0$. Then, the second-order scheme $\whbfC=(\whC_1,\ldots,\whC_d)$ over the uniform grid $\{jh\mid j=0,\ldots,N\}$ is defined as follows
\begin{align}\label{CI_scheme}
	\whC^N_{i,0} &= c_i, \\
	\whC^N_{i,(j+1)h} &= \eta_i(h,\whC^N_{i,jh},\sqrt{h}G_{i,j+1}), \quad \text{ for all }j\in\{0,\ldots,N-1\},
\end{align}
where $(G_{i,j})_{i=1,\ldots,d,\,j=1,\ldots,N}$ are i.i.d. standard Gaussian random variables.
	\begin{remark}
		A complete proof of the weak convergence of the proposed simulation scheme is beyond the scope of this paper and is left for future research. Nevertheless, the numerical construction is based on a quasi-exact scheme for the log-fund process, in which the integrated variance is approximated by the trapezoidal rule. Therefore, the overall scheme is expected to achieve weak second-order accuracy, provided that the variance process is itself simulated by means of an exact or weak second-order scheme. In our framework, the joint dynamics of the fund and the carbon-intensity process, 
		$(X^{\bmpi^\star},\,C)$, have the same structure as the log-Heston model. For this class of models, \cite{AL2025SIFIN} prove a second-order weak convergence result when the Cox-Ingersoll-Ross variance process is simulated either exactly or through a second-order scheme, without requiring the Feller condition.
\end{remark}
\subsection{Option evaluation and variance reduction formula}
We derive estimators for the efficient pricing and hedging of the three contracts introduced in Section \ref{sect:locally_risk_minimizing}. Although the proposed construction is quite general, it requires the specification of a mortality intensity model under which the quantities $\mathbb{P}^I(\tau\g t)$ and $\mathbb{E}^{P^I}\left[(1-H_t)\gamma_t\right]$ can be computed efficiently, either in closed form or through numerical approximation. Models with deterministic mortality intensity, such as the Gompertz-Makeham model (see equation \eqref{GM_intensity}), provide a natural benchmark setting.
\paragraph{Pure Endowment Contract.} We begin by considering a single pure endowment contract with payoff specified in equation \eqref{eq:PE_PAYOFF}, although the proposed methodology can be extended to more general payoff structures.\\
Our objective is to evaluate $\E^{\Q^*}_{0,x,\bfc}[e^{-rT}\phi(X^{\bmpi^\star}_T)\mathds{1}_{\{\tau>T\}}]$ and $\partial_x\E^{\Q^*}_{0,x,\bfc}[e^{-rT}\phi(X^{\bmpi^\star}_T)\mathds{1}_{\{\tau>T\}}]$. A straightforward approach is to approximate $X^{\bmpi^\star}$ by means of the discretization scheme \eqref{1step_whX} and exploit the independence between the financial market and the policyholder's lifetime, thus obtaining the candidate random variable
\begin{equation}\label{std_rv_PE}
	\Upsilon^{st}_{PE}(N)=e^{-rT}\P^I(\tau>T)\phi(\whX^N_{Nh}).\footnote{The corresponding standard Monte Carlo estimator is obtained by averaging $M$ independent samples.}
\end{equation}
We proceed in a different way by exploiting the conditional distribution of $\whX$ given $\whbmpi$ to construct a new random variable with the same expectation but substantially smaller variance (see panel $A$ of Table \ref{Table_VarReduc}), thereby improving the efficiency of the Monte Carlo estimator. The next result formalizes this idea. 
\begin{prop}\label{var_reduc_prop_PE}
	Let $\Phi_\mcN:\R\rightarrow(0,1)$ be the cumulative distribution function of a standard Gaussian random variable. For all $j\in\{1,\ldots,N\}$, let
	\begin{equation}\label{vNpi}
		v^j_\whbmpi = \frac h2 \sum_{l=0}^{j-1}(\whbmpi^\top_{lh}\bfSigma\bfSigma^\top \whbmpi_{lh}+  \whbmpi^\top_{(l+1)h}\bfSigma\bfSigma^\top \whbmpi_{(l+1)h}), 
	\end{equation}
	and
	\begin{equation}\label{abpi}
		a^j_{\whbmpi}(y)= \frac{\log(y/x)-rT+\frac12 v^j_\whbmpi}{\sqrt{v^j_\whbmpi}}, \quad \text{ and } \quad b^j_\whbmpi(y) = a^j_\whbmpi(y)-\sqrt{v^j_\whbmpi}.
	\end{equation}
	Then, given 
	\begin{equation}\label{vr_rv_PE}
		\Upsilon^{vr}_{PE}(N) = e^{-rT}\Big(k\Phi_\mcN(a^N_\whbmpi(k)) + \bar{k}\Phi_\mcN(-a^N_\whbmpi(\bar{k})) +xe^{rT}\big(\Phi_\mcN(b^N_\whbmpi(\bar{k}))-\Phi_\mcN(b^N_\whbmpi(k))\big)\Big)\P^I(\tau > T),
	\end{equation}
	one has
	\begin{equation}
		\E^{\Q^M}_{0,x,\bfc}[\Upsilon^{vr}_{PE}(N)]=\E^{\Q^M}_{0,x,\bfc}[\Upsilon^{st}_{PE}(N)] \quad\text{and}\quad
		\Var^{\Q^M}_{0,x,\bfc}[\Upsilon^{vr}_{PE}(N)]\le\Var^{\Q^M}_{0,x,\bfc}[\Upsilon^{st}_{PE}(N)]. 
	\end{equation}
\end{prop}
\begin{proof}
	From the definition of $v^N_\whbmpi$ and equation \eqref{1step_whX}, we obtain that  $\whX^{N}_T=x\exp(rT-\frac12v^N_\whbmpi+\sqrt{v^N_\whbmpi}F)$. Hence, 
	\begin{align*}
		\E^{\Q^M}[\phi(\whX^{N}_T)] &= \E^{\Q^M}\bigg[\min\Big(\bar{k},\max\big(k,x\exp(rT-\frac12v^N_\whbmpi+\sqrt{v^N_\whbmpi}F)\big)\Big)\bigg], \\
		&= \begin{multlined}[t]
			\E^{\Q^M}\bigg[ k\ind[F<a^N_\whbmpi(k)] + \bar{k}\ind[F>a^N_\whbmpi(\bar{k})] +x\exp(rT-\frac12v^N_\whbmpi+\sqrt{v^N_\whbmpi}F)\ind[a^N_\whbmpi(k)<F<a^N_\whbmpi(\bar{k})]  \bigg].
		\end{multlined}
	\end{align*}
	Next, conditioning respect to $v^N_\whbmpi$ and using the independence between $F$ and $v^N_\whbmpi$, we get 
	\begin{align}
		\E^{\Q^M}[\phi(\whX^{N}_T)]&=\begin{multlined}[t]
			\E^{\Q^M}\bigg[ \E^{\Q^M}\Big[ k\ind[F<a^N_\whbmpi(k)] + \bar{k}\ind[F>a^N_\whbmpi(\bar{k})] \\
			\qquad+x\exp(rT-\frac12v^N_\whbmpi+\sqrt{v^N_\whbmpi}F)\ind[a^N_\whbmpi(k)<F<a^N_\whbmpi(\bar{k})] \,\Big|\, v^N_\whbmpi \Big] \bigg],
		\end{multlined} \nonumber\\
		&=\E^{\Q^M}\Big[ k\Phi_\mcN(a^N_\whbmpi(k)) + \bar{k}\Phi_\mcN(-a^N_\whbmpi(\bar{k})) +xe^{rT}\big(\Phi_\mcN(b^N_\whbmpi(\bar{k}))-\Phi_\mcN(b^N_\whbmpi(k))\big)\Big], \label{PE_reduc_formula}
	\end{align}
	where the functions $a^N_\whbmpi(y)$ and $ b^N_\whbmpi$ are given in $\eqref{abpi}$ and $\Phi_\mcN$ denotes the cumulative distribution function of a standard Gaussian random variable. Now, by the law of total variance,
	$$\Var^{\Q^M}(\Upsilon^{st}_{PE}(N)) = \Var^{\Q^M}(\E[\Upsilon^{st}_{PE}(N)\mid v^N_\whbmpi]) +\underbrace{\E^{\Q^M}[\Var^{\Q^M}(\Upsilon^{st}_{PE}(N))\mid v^N_\whbmpi)]}_{\ge0},$$
	and we recognize  $\E[\Upsilon^{st}_{PE}(N)\mid v^N_\whbmpi]=\Upsilon^{vr}_{PE}(N)$, which completes the proof.
\end{proof}
\begin{remark}
	From Proposition \ref{var_reduc_prop_PE}, we get the following results.
	\begin{itemize}
		\item[(i)] Employing the chain rule and simple algebra, we get a simple and elegant formula for the $x$-derivative of $\Upsilon^{vr}_{PE}(N)$, given by 
		\begin{equation}\label{delta_vr_rv_PE}
			\partial_x \Upsilon^{vr}_{PE}(N) = e^{rT}(\Phi_\mcN(b^N_\whbmpi(\bar{k}))-\Phi_\mcN(b^N_\whbmpi(k)).
		\end{equation} 
		\item [(ii)] Expanding the term  $\Var^{\Q^M}_{0,x,\bfc}(\Upsilon^{st}_{PE}(N)\mid v^N_\whbmpi)$ and computing $\E^{\Q^M}_{0,x,\bfc}[\Upsilon^{st}_{PE}(N)^2\mid v^N_\whbmpi]$, the difference of the variance between the standard estimator and its variance reduced version is given by 
		\begin{align}
			&\Var^{\Q^M}[\Upsilon^{st}_{PE}(N)]-\Var^{\Q^M}[ \Upsilon^{vr}_{PE}(N)] \\
			&= e^{-2rt}\E^{\Q^M}\!\!\bigg[ k^2\Phi_\mcN(a^N_\whbmpi(k)) +\bar{k}^2\Phi_\mcN(-a^N_\whbmpi(\bar{k}))+ x^2e^{2rT}\exp\Big( v^N_\whbmpi\Big)\big(\Phi_\mcN(c^N_\whbmpi(\bar{k}))-\Phi_\mcN(c^N_\whbmpi(k))\big) \\
			\label{var_reduc_quantification}&\quad -\Big(k\Phi_\mcN(a^N_\whbmpi(k)) + \bar{k}\Phi_\mcN(-a^N_\whbmpi(\bar{k})) +xe^{rT}\big(\Phi_\mcN(b^N_\whbmpi(\bar{k}))-\Phi_\mcN(b^N_\whbmpi(k))\big)\Big)^2\bigg]\P^I(\tau  > T)^2,
		\end{align}
		where $c^N_\whbmpi(y) = a^N_\whbmpi(y) - 2\sqrt{v^N_\whbmpi}$.\footnote{For the sake of readability, we have dropped the dependence on $0,x,\bfc$ in expectations and variances in equation \eqref{var_reduc_quantification}.}
	\end{itemize}
\end{remark}
\paragraph{Term Insurance Contract.} For the term insurance contract, we consider the payoff function specified in equation \eqref{eq:TI_PAYOFF}. To evaluate $\E^{\Q^*}_{0,x,\bfc}[e^{-r\tau}\psi(\tau,X^{\bmpi^\star}_\tau)\ind[\tau<T]]$ and $\partial_x\E^{\Q^*}_{0,x,\bfc}[e^{-r\tau}\psi(\tau,X^{\bmpi^\star}_\tau)\ind[\tau<T]]$, we follow the same procedure as in the pure endowment. Setting $P_t=e^{-rt}\E^{\P^I}[(1-H_t)\gamma_t ]$, for every $t \in [0,T]$, the corresponding standard random variable is given by
\begin{equation}\label{st_rv_TI}
	\Upsilon^{st}_{TI}(N) =\frac{h}{2} \sum_{j=0}^{N-1} \big(\psi(jh,\whX^N_{jh})P_{jh}+\psi((j+1)h,\whX^N_{(j+1)h})P_{(j+1)h} \big),
\end{equation}
In the following proposition, we introduce a random variable with the same expected value as \eqref{st_rv_TI}, but with reduced variance (see panel $B$ of Table \ref{Table_VarReduc}).
\begin{prop}\label{var_reduc_prop_TI}
	Let $\Phi_\mcN:\R\rightarrow(0,1)$ be the cumulative distribution function of a standard Gaussian random variable, and let $a^j_\whbmpi$ and $b^j_\whbmpi$ be defined as in \eqref{abpi} for all $j\in\{1,\ldots,N\}$. Define
	\begin{multline}\label{vr_rv_TI}
		\Upsilon^{vr}_{TI}(N)=\frac{h}{2}\sum_{j=0}^{N-1} \left(k\Phi_\mcN(a^j_\whbmpi(k)) + \bar{k}\Phi_\mcN(-a^j_\whbmpi(\bar{k})) +xe^{ rjh}\big(\Phi_\mcN(b^j_\whbmpi(\bar{k}))-\Phi_\mcN(b^j_\whbmpi(k))\big)\right)P_{jh} \\
		+\left( k\Phi_\mcN(a^{j+1}_\whbmpi(k)) + \bar{k}\Phi_\mcN(-a^{j+1}_\whbmpi(\bar{k})) +xe^{ r(j+1)h}\big(\Phi_\mcN(b^{j+1}_\whbmpi(\bar{k}))-\Phi_\mcN(b^{j+1}_\whbmpi(k))\big)\right)P_{(j+1)h}.    
	\end{multline}
	Then
	$$\E^{\Q^M}_{0,x,\bfc}[\Upsilon^{vr}_{TI}(N)] = \E^{\Q^M}_{0,x,\bfc}[\Upsilon^{st}_{TI}(N)] \quad\text{and}\quad \Var^{\Q^M}_{0,x,\bfc}[\Upsilon^{vr}_{TI}(N)] \le \Var^{\Q^M}_{0,x,\bfc}[\Upsilon^{st}_{TI}(N)].$$
\end{prop}
\begin{proof}
	We start by observing that, by using the compensator of $H_t$, Fubini's Theorem and the independence of $\tau$ from $X^{\bmpi^\star}$ and $\bfC$, one has 
	\begin{align}
		\E^{\Q^*}[e^{-r\tau}\psi(\tau,X^{\bmpi^\star}_\tau)\ind[\tau<T]]&=\E^{\Q^*}\bigg[\int_0^Te^{-rt}\psi(t,X^{\bmpi^\star}_t) \de H_t\bigg]=\E^{\Q^*}\bigg[\int_0^Te^{-rt}\psi(t,X^{\bmpi^\star}_t)(1-H_t)\gamma_t \de t\bigg]\\
		&=\int_0^T e^{-rt}\E^{\Q^M}[\psi(t,X^{\bmpi^\star}_t)]\E^{\P^I}[(1-H_t)\gamma_t ]\de t. \label{Fubini_TI}
	\end{align}
	Now, using equation \eqref{st_rv_TI}, discretizing $X^{\bmpi^\star}$ through the scheme in equation \eqref{Nstep_whX} on the uniform grid $\{jh \mid j=0,\ldots,N\}$, and approximating the Lebesgue integral in \eqref{Fubini_TI} by the trapezoidal rule, we obtain the approximation of \eqref{Fubini_TI} given by:
	\begin{align}
		\int_0^T\E^{\Q^M}[\psi(t,X^{\bmpi^\star}_t)]P_{t}\de t&\approx\frac{h}{2}\sum_{j=0}^{N-1}\E^{\Q^M}[\psi(jh,\whX^N_{jh})]P_{jh}+\E^{\Q^M}[\psi((j+1)h,\whX^N_{(j+1)h})]P_{(j+1)h}\\
		&=\E^{\Q^M}[\Upsilon^{st}_{TI}(N)]. 
	\end{align}
	where $P_t=e^{-rt}\E^{\P^I}[(1-H_t)\gamma_t]$. The conclusion then follows by repeating the same conditioning argument used in the proof of \ref{prop:FS_decomp_PE}.
\end{proof}
\paragraph{Endowment Insurance Contract.} We next consider an endowment insurance contract whose payoff combines discounted versions of the payoffs defined in equations \eqref{eq:disc_PE_PAYOFF} and \eqref{eq:disc_TI_PAYOFF}. As in the pure endowment and term insurance settings, we aim to approximate
$$\E^{\Q^*}_{0,x,\bfc}\left[\tilde\psi(\tau,X^{\bmpi^\star}_\tau)\ind[\tau<T] +\tilde\phi(X^{\bmpi^\star}_T)\ind[\tau\ge T]\right],$$ and its derivative $\partial_x\E^{\Q^*}_{0,x,\bfc}\left[\tilde\psi(\tau,X^{\bmpi^\star}_\tau)\ind[\tau<T] +\tilde\phi(X^{\bmpi^\star}_T)\ind[\tau\ge T]\right]$. Combining the results of Propositions \ref{var_reduc_prop_PE} and \ref{var_reduc_prop_TI} yields the following standard and variance reduced random variables:  
\begin{equation}
	\Upsilon^{st}_{EI}(N) = \varrho \Upsilon^{st}_{TI}(N) + \Upsilon^{st}_{PE}(N),
	\qquad
	\Upsilon^{vr}_{EI}(N) = \varrho \Upsilon^{vr}_{TI}(N) + \Upsilon^{vr}_{PE}(N).
\end{equation}
\subsection*{Variance reduction Test} 
We perform a simulation exercise to assess the variance reduction achieved by the proposed methodology, using the parameter set described in Section \ref{sec:hedging}. For each contract, we compare the standard and variance reduced estimators and report the resulting improvement in computational efficiency. The results are summarized in Table \ref{Table_VarReduc}.\\
Panel $A$ of Table \ref{Table_VarReduc} reports the results for the pure endowment contract. Comparing the variance of $\Upsilon^{st}_{PE}(N)$ with that of $\Upsilon^{vr}_{PE}(N)$, we obtain an average reduction of $99.9535$\%. Equivalently, the same statistical error achieved with $M$ samples of $\Upsilon^{st}_{PE}(N)$ can be obtained using only about $M/2150$ samples of $\Upsilon^{vr}_{PE}(N)$, resulting in a substantial reduction in computational time.\\
Panel $B$ of Table \ref{Table_VarReduc} presents the corresponding results for the term insurance contract. The variance reduction remains very significant, with an average value of $97.4639$\%. In practical terms, this implies that the estimator $\Upsilon^{vr}_{TI}(N)$ requires only about one-fortieth of the number of simulations needed by the standard estimator $\Upsilon^{st}_{TI}(N)$ to attain the same statistical error.\\
Finally, panel $C$ of Table \ref{Table_VarReduc} summarizes the results for the endowment insurance contract. These are close to those of the pure endowment for maturities $T=\{5,10,20\}$, while for $T=30$ they become more similar to those of the term insurance contract. This is consistent with the structure of the contract, which resembles a pure endowment at shorter maturities and progressively behaves more like a term insurance contract as the maturity increases.\\
Overall, the variance reduction results are highly satisfactory. Even in the least favorable case, namely the term insurance contract with maturity $T=30$, the variance reduction is still equal to $94.7999$\%. This means that only about $5$\% of the simulations required in the standard case are needed to achieve the same level of statistical error.
\begin{table}[H]
	\centering
	\begin{tabular}{ccccc}
		\Xhline{1.2pt}
		\rowcolor{gray!20}
		\multicolumn{5}{c}{Panel $A$: Pure endowment}                                  \\
		\Xhline{1.2pt}
		& $T=5$          & $T=10$         & $T=20$         & $T=30$        \\
		\hline
		$\Var(\Upsilon^{st}_{PE}(N))$  & $2.3325$e$-02$ & $4.2705$e$-02$ & $3.7802$e$-02$ & $4.6313$e$-03$  \\
		$\Var(\Upsilon^{vr}_{PE}(N))$  & $7.3157$e$-06$ & $2.0879$e$-05$ & $2.1528$e$-05$ & $2.2617$e$-06$  \\
		\hline
		Variance reduction             & $99.9686$\%    & $99.9511$\%    & $99.9431$\%    & $99.9511$\%   \\
		\Xhline{1.2pt}
		\rowcolor{gray!20}
		\multicolumn{5}{c}{Panel $B$: Term insurance}                                  \\
		\Xhline{1.2pt}
		& $T=5$          & $T=10$         & $T=20$         & $T=30$        \\
		\hline
		$\Var(\Upsilon^{st}_{TI}(N))$  & $6.4640$e$-06$ & $4.5424$e$-05$ & $4.2306$e$-04$ & $1.1124$e$-03$  \\
		$\Var(\Upsilon^{vr}_{TI}(N))$  & $2.5549$e$-08$ & $5.7513$e$-07$ & $1.3889$e$-05$ & $5.7844$e$-05$  \\
		\hline
		Variance reduction             & $99.6047$\%    & $98.7339$\%    & $96.7170$\%    & $94.7999$\%   \\
		\Xhline{1.2pt}
		\rowcolor{gray!20}
		\multicolumn{5}{c}{Panel $C$: Endowment insurance}                             \\
		\Xhline{1.2pt}
		& $T=5$          & $T=10$         & $T=20$         & $T=30$        \\
		\hline
		$\Var(\Upsilon^{st}_{EI}(N))$  & $2.3378$e$-02$ & $4.2908$e$-02$ & $3.9107$e$-02$ & $5.9768$e$-03$  \\
		$\Var(\Upsilon^{vr}_{EI}(N))$  & $8.1701$e$-06$ & $2.8083$e$-05$ & $6.8198$e$-05$ & $8.1028$e$-05$  \\
		\hline
		Variance reduction             & $99.9651$\%    & $99.93454$\%   & $99.8256$\%    & $98.6443$\%\\  
		\Xhline{1.2pt}
	\end{tabular}
	\caption{Sample variances of the standard and variance reduced estimators for the three insurance contracts, together with the corresponding percentage variance reduction, for maturities $T=\{5,10,20,30\}$. Panel $A$ refers to the pure endowment contract, Panel $B$ to the term insurance contract, and Panel $C$ to the endowment insurance contract. Number of samples: $10^6$. Discretization steps: $N=N(T)=5T$.}
	\label{Table_VarReduc}
\end{table}
\end{document}